\documentclass[a4paper,fleqn,usenatbib]{mnras_arXiv2}

\usepackage{mathptmx}

\usepackage[T1]{fontenc}
\usepackage{ae,aecompl}

\usepackage{graphicx}	
\usepackage{amsmath}	
\usepackage{amssymb}	
\usepackage{bm}
\usepackage{tikz}
\usepackage{cuted}

\DeclareMathOperator{\e}{e}

\renewcommand*{\vec}[1]{\mathit{#1}}
\newcommand*{\mat}[1]{\mathbfss{#1}}
\newcommand*{\trans}[1]{{#1}^T}
\newcommand{\dd}{{\rm d}}
\newtheorem{theorem}{Theorem}

\newcommand{\REWRITE}[1]{\textcolor{black}{#1}}
\newcommand{\REWRITER}[1]{\textcolor{black}{#1}}

\title[Bias and robustness of eccentricity estimates]{Bias and robustness of eccentricity estimates from radial velocity data}

\author[N. C. Hara et al.]{
Nathan C. Hara,$^{1,2,3}$\thanks{E-mail:nathan.hara@unige.fr}
G. Bou\'e$^{2}$,
J. Laskar$^{2}$,
J.-B. Delisle$^{1,2}$,
N. Unger$^{4}$
\\ 
$^{1}$ Observatoire de Gen\`eve, Universit\'e de Gen\`eve, 51 ch. des Maillettes, 1290 Versoix, Switzerland
\\
$^{2}$ ASD/IMCCE, CNRS-UMR8028, Observatoire de Paris,  PSL, UPMC, 77 Avenue Denfert-Rochereau, 75014 Paris, France
\\
$^{3}$ NCCR PlanetS CHEOPS Fellow, Switzerland
\\
$^{4}$ Universidad de Buenos Aires, Facultad de Ciencias Exactas y Naturales. Buenos Aires, Argentina
}

\pubyear{2019}

\begin{document}
\label{firstpage}
\pagerange{\pageref{firstpage}--\pageref{lastpage}}
  \maketitle

\begin{abstract}
Eccentricity is a parameter of particular interest as it is an informative indicator of the past of planetary systems. It is however not always clear whether the eccentricity fitted on radial velocity data is real or if it is an artefact of an inappropriate modelling.
In this work, we address this question in two steps: we first assume that the model used for inference is correct and present interesting features of classical estimators. Secondly, we study whether the eccentricity estimates are to be trusted when the data contain incorrectly modelled signals, such as missed planetary companions, non Gaussian noises, correlated noises with unknown covariance, etc. Our main conclusion is that data analysis via posterior distributions, with a model including a free error term gives reliable results provided two conditions. First, convergence of the numerical methods needs to be ascertained. Secondly, the noise power spectrum should not have a particularly strong peak at the semi period of the planet of interest. As a consequence, it is difficult to determine if the signal of an apparently eccentric planet might be due to another inner companion in 2:1 mean motion resonance. We study the use of Bayes factors to disentangle these cases.  Finally, we suggest methods to check if there are hints of an incorrect model in the residuals. We show on simulated data the performance of our methods and comment on the eccentricities of Proxima b and 55 Cnc f.

\end{abstract}



\section{Introduction}
\label{sec:introduction}

The nearly coplanar and circular orbits of the Solar system have long been an argument in favour of Laplace and Kant's theory of formation of planets in a disk~\citep{swedenborg1734, kant1755, laplace1796}. The first observations of exoplanets suggested that such low eccentricities were rather the exception than the rule. The ``eccentricity problem'', along with the need to envision migration scenarios for hot Jupiters, triggered several theoretical studies which explored migration scenarios after the disk has dissipated. The prediction of these models were compared to measured eccentricities. For instance~\cite{juric2008} and~\cite{petrovich2014} evaluate the likelihood of formation scenarios of hot and warm Jupiters through their agreement with observed eccentricity distributions.

 For a radial velocity data set on a given star, one wants to extract two pieces of information about the eccentricity. First, a best candidate value (the estimation problem) and what are the eccentricity values that are incompatible with the data (the hypothesis testing problem). It is in particular interesting to test if an eccentricity is non zero. Both problems do not have completely obvious solutions. For instance it is known since~\cite{lucy1971} that \REWRITE{when the true eccentricity is small, its least square estimate is biased upwards}.  Other aspects of the estimation and hypothesis testing problems have been addressed in the exoplanet community.

  In~\cite{shenturner2008} the bias was found to depend on the signal to noise ratio \REWRITE{(denoted by SNR)} as well as on the time span of the observations. This was confirmed by~\cite{zakamska2011}, which further showed that the bias depends on the phase coverage, and updated the~\cite{lucy1971} null hypothesis test to determine if a null eccentricity can be rejected or not. 
  They also propose metrics for evaluating the quality of a data set. \cite{otoole2009a} showed that error bars on eccentricity from least square are under-estimated by a factor 5-10.   \cite{brown2017} shows that there might be orbits with very high eccentricities with similar goodness-of-fit as a low eccentricity one.   \cite{pont2011},~\cite{husnoo2011} and \cite{husnoo2012} used Bayesian Information Criterion to confirm non-zero eccentricities. More recently,~\cite{bonomo2017a,bonomo2017b} assessed the evidence in favour of eccentric solutions with Bayes factors. A Bayesian test with a physically motivated prior on eccentricity was devised by~\cite{lucy2013}. Also,~\cite{angladaescude2010}, \cite{wittenmyer2013} and~\cite{kurster2015} note that two planets in 2:1 mean motion resonances can appear as an eccentric planet, and propose ways to disentangle those cases. \REWRITE{This problem has also been studied in~\cite{boisvert2018, nagel2019, wyttenmeier2019}}.

 The fact that eccentricity estimates can be spuriously high for a given planet gives reasons for concern on the eccentricity distributions.  The estimation of those has been tackled in \cite{hogg2010}, which computes the posterior of the eccentricity probability distribution itself. 
 ~\cite{zakamska2011} consider the accuracy of the eccentricity catalogues obtained by Bayesian point estimates. They show that for single planet populations contaminated by white noise, estimating the eccentricity via the maximum of the marginalised posterior distribution of eccentricity with a free jitter term gives satisfactory retrieval of the input population.
 Furthermore, it has been noted by~\cite{cumming2004} that high eccentricity orbits $\gtrsim 0.6$ are more difficult to detect at fixed semi-amplitude. For a fixed mass, the detection bias is less strong~\citep{shenturner2008}.

 We contribute to this series of work by studying in depth the bias and robustness of eccentricity estimates from radial velocity data (note that our analysis also applies to astrometric measurements and to estimates of semi-amplitude and inclination). We proceed in two steps. First, we highlight key properties of classical estimators, in order to have a consistent view of eccentricity estimation. The  following questions are then considered:  is the eccentricity inference robust to modelling errors? By that, we refer to wrong noise models, planetary companion too small to be detected, etc. If not, how to mitigate the problem? One could encounter a situation similar to the spectroscopic binaries, where proximity effects or gas streams lead to spurious high-significance eccentricities if not properly accounted for~\citep{lucy2005}.

The article is structured as follows. In section~\ref{sec:originecc} we study the behaviour of eccentricity estimates when the model is correct. The least square estimates as well as Bayesian ones are studied, and is is shown that the latter are less biased at low eccentricity. \REWRITE{The problem of spurious local $\chi^2$ minima at high eccentricity is also tackled, in particular through the Proxima b case~\citep{angladaescude2016}. }
Section~\ref{sec:robustness} is devoted to studying the robustness of the estimates when the numerical method, the model or the prior is poorly chosen. Finally in section~\ref{sec:resanalysis}, we consider ways to check the validity of a noise model.
 Our methods are illustrated with the 55 Cnc HIRES data set in section~\ref{sec:application}. In section~\ref{sec:conclusion}, we conclude with a step-by-step procedure to obtain reliable eccentricities and  present perspectives for future work .

\section{Point and interval eccentricity estimates}
\label{sec:originecc}

\subsection{Problem definition}
\subsubsection{Point and interval estimates}
\label{sec:problemstatement}
Let us first define the problem under study precisely. Some generic symbols, used throughout the text, are summarized in table~\ref{tab:listofsymb}.
 \begin{table}\centering
	\caption{List of symbols. }
	\begin{tabular}{c|p{6cm}}
		$\btheta$ & Vector of parameters \\ \hline
		$\btheta_t$ & True value of the vector of parameters \\ \hline
		$\widehat{\btheta}$ & Estimator of $\btheta$ \\ \hline
		$\mathbfit{f}(\mathbfit{t}, \btheta)$ &  Deterministic model sampled at times $\mathbfit{t} = (t_k)_{k=1..N}$ and of parameters $\btheta$ (orbital parameters plus possibly offset, trend...) \\ \hline
		$p(\btheta)$ & Probability density of $\btheta$, $\mathrm{Pr} \{ \btheta \in \Theta \} = \int_{\Theta} p(\btheta) d\btheta$ for some measurable set $\Theta$  \\ \hline
		$\mathbb{E}\{\btheta\}$ & Mathematical expectancy of the random variable $\btheta$ \\ \hline
		$k,h$ & $k = e\cos \omega$, $h = e\sin \omega$ \\ \hline
	\end{tabular}
	\label{tab:listofsymb}
\end{table}
Let us consider a time series of $N$ observations, modelled as a vector $\mathbfit y = (y(t_k))_{k=1..N}$, such that
\begin{align}
\mathbfit{y}(\mathbfit{t})  = \mathbfit{f}(\mathbfit{t} ,\btheta)  + \bepsilon 
\end{align} 
where $\mathbfit{t} = (t_k)_{k=1..N}$ is the vector of measurement times, $f$ is a deterministic model depending on parameters $\btheta \in \mathbb{R}^p$ and $\bepsilon$ is a random variable modelling the noise.  An estimator of $\btheta$ is a function $\widehat{ \theta}$ of the data $\mathbfit{y}(\mathbfit{t})$ whose output is wanted to be close to the true value of $\btheta$, denoted by $\btheta_t$, in a sense chosen by the data analyst. If the mean value of $\widehat{ \theta}(\mathbfit{y})$ ($\mathbb{E}\{\widehat{ \theta}(\mathbfit{y})\}$) is not equal to $\btheta_t$, the estimator is said to be biased and 
\begin{align}
b_{\widehat{ \theta}}(\btheta_t) =  \mathbb{E}\{\widehat{ \theta}(\mathbfit{y})\} - \btheta_t
\label{eq:defbias}
\end{align}
is called the bias of the estimator $\widehat{ \theta}$ in $\btheta_t$. A common metric for the accuracy of an estimator is the mean squared error (MSE), linked to bias and the variance of the estimator, $\mathrm{Var}\{\widehat{ \theta}(\mathbfit{y})\}$, via
\begin{align}
\mathrm{MSE}(\widehat{ \theta}):=\mathbb{E}\{(\widehat{ \theta}(\mathbfit{y}) -\btheta_t )^2 \}= \mathrm{Var}\{\widehat{ \theta}(\mathbfit{y})\}  + b_{\widehat{ \theta}}(\btheta_t)^2.
\label{eq:defmse}
\end{align}

The other problem we consider is to have a testing procedure to reject or not certain values of the eccentricity. We are now interested in rejecting the hypothesis that $e\in C$ where $C$ is a subset of $[0,1]$. More precise definitions are given in the relevant sections (\ref{sec:intervalest} and~\ref{sec:bayessection}). 

\REWRITE{In the present section~\ref{sec:originecc}, we describe tools for the estimation and hypothesis testing problems and present some of their properties. In section~\ref{sec:robustness}, we study the reliability of these tools when the model is incorrect. \REWRITER{By convention, in the following, radial velocity signals are in m.s\textsuperscript{-1}. The analysis is unchanged for other units as long as the signal to noise ratio is identical.}
}

\subsubsection{Model}
\label{sec:modeldef}

The concern of the present work is the estimation of  eccentricity from radial velocity data. The model of a radial velocity planetary signal is recalled below
\begin{align}
f(t,e,K,P,\omega,M_0) & = K(\cos\left(\omega + \nu(t,e,P,\omega,M_0)\right) + e \cos \omega) \label{eq:vexprbis}\\
\cos \nu & = \frac{\cos E - e}{1 - e\cos E} \label{eq:cosnu2bis}\\
\sin \nu & = \frac{\sqrt{1-e^2}\sin E}{1 - e\cos E} \label{eq:sinnu2bis}\\
E - e \sin E &= M_0 + \frac{2 \pi}{P} t \label{eq:keplereq2bis}.
\end{align}
The symbols $t,e,K,P,\omega,M_0$ designate respectively the measurement time, eccentricity, semi-amplitude, period, argument of periastron and mean anomaly at $t=0$. \REWRITE{The symbols $E$ and $\nu$ denote the eccentric and true anomalies. }

We assume a Gaussian noise model, such that the likelihood function is
\begin{align}
p(\mathbfit{y} | \btheta, \bbeta) = \frac{1}{\sqrt{(2\pi)^N |\mathbf{V}(\bbeta) | }} \e^{ -\frac{1}{2} (\mathbfit{y} - \mathbfit{f}(\mathbfit{t} ,\btheta))^T  \mathbf{V}(\bbeta) ^{-1}  (\mathbfit{y} - \mathbfit{f}(\mathbfit{t} ,\btheta)) }
\label{eq:likelihood}
\end{align}
where $\mathbfit{f}(\mathbfit{t} ,\btheta)$ is a sum of Keplerian functions defined as~\eqref{eq:vexprbis} possibly plus some other model features (offset, trend...). The covariance matrix $ \mathbf{V}$ is parametrized by $\bbeta$ and the suffix $T$ denotes the matrix transposition. The explicit expression of $\mathbfit{f}(\mathbfit{t} ,\btheta)$ and $ \mathbf{V}(\bbeta)$  will be given in the relevant sections.

The features of least square and Bayesian estimates  are now studied respectively in sections~\ref{sec:leastsquare} and~\ref{sec:bayesianestimates}.

\subsection{Least square estimate}
\label{sec:leastsquare}

\begin{figure}
	
	\centering
	\hspace{-0.3cm}
	\includegraphics[width=8.1cm]{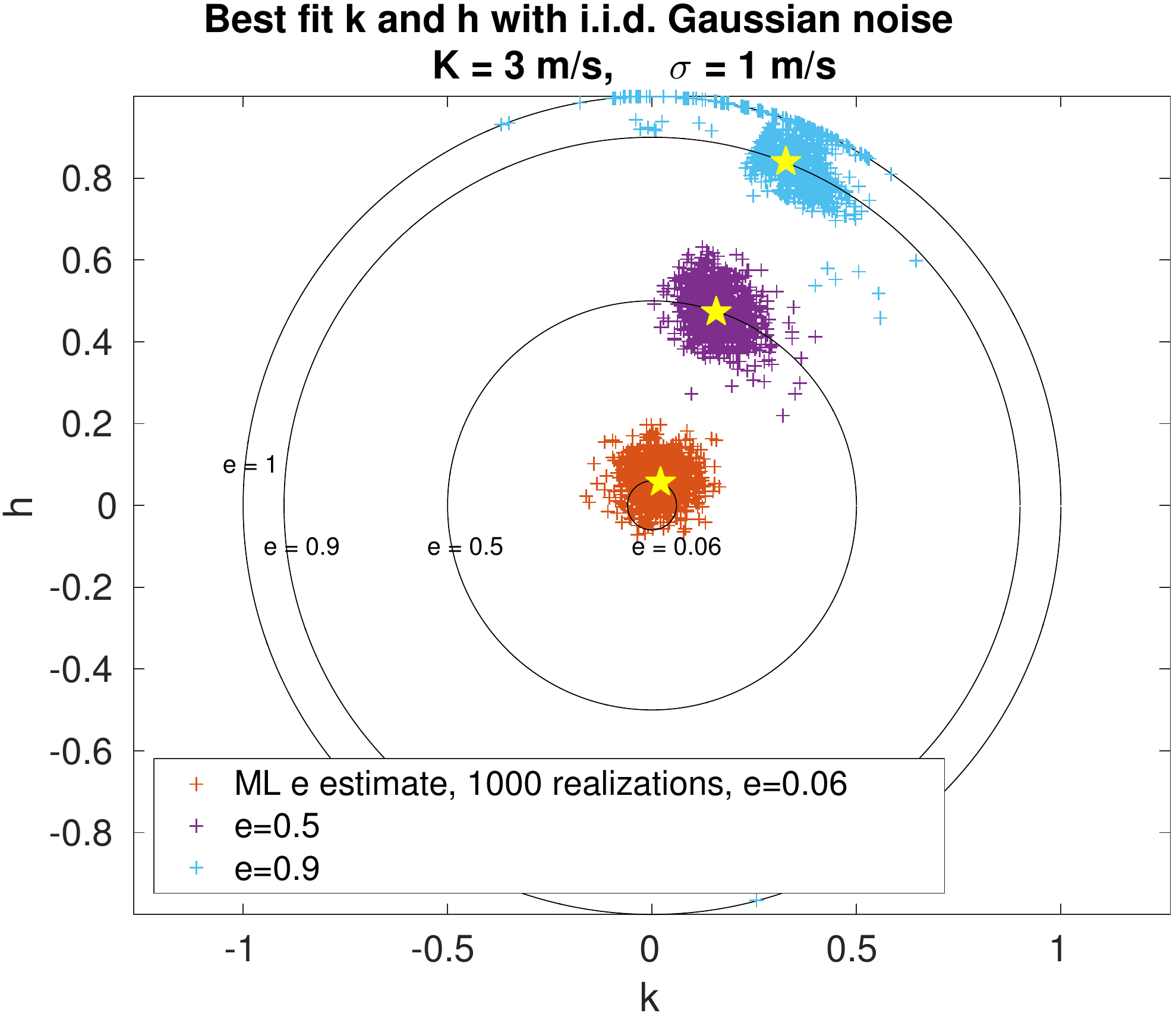}
	\caption{  Estimates of $k=e\cos\omega$ and $h=e\sin\omega$ for various true eccentricities. Each cross represents an estimate $\widehat{k}$ and $\widehat{h}$ obtained with a Keplerian model least square fit for a given noise realization. The yellow stars represents the true $(k,h)$. Estimates for $e$=0.06, 0.5 and 0.9 are respectively in red, purple and blue (3 $\times$ 1000 = 3000 estimates in total on each figure). $\sigma$ = 1 m.s\textsuperscript{-1} and $K$ = 3 m/s.
	}
	\label{target2}
\end{figure}
  \begin{figure}
 	\vspace{0.3cm}
 	
 	\hspace{-1.4cm}
	 	\includegraphics[width=9.5cm]{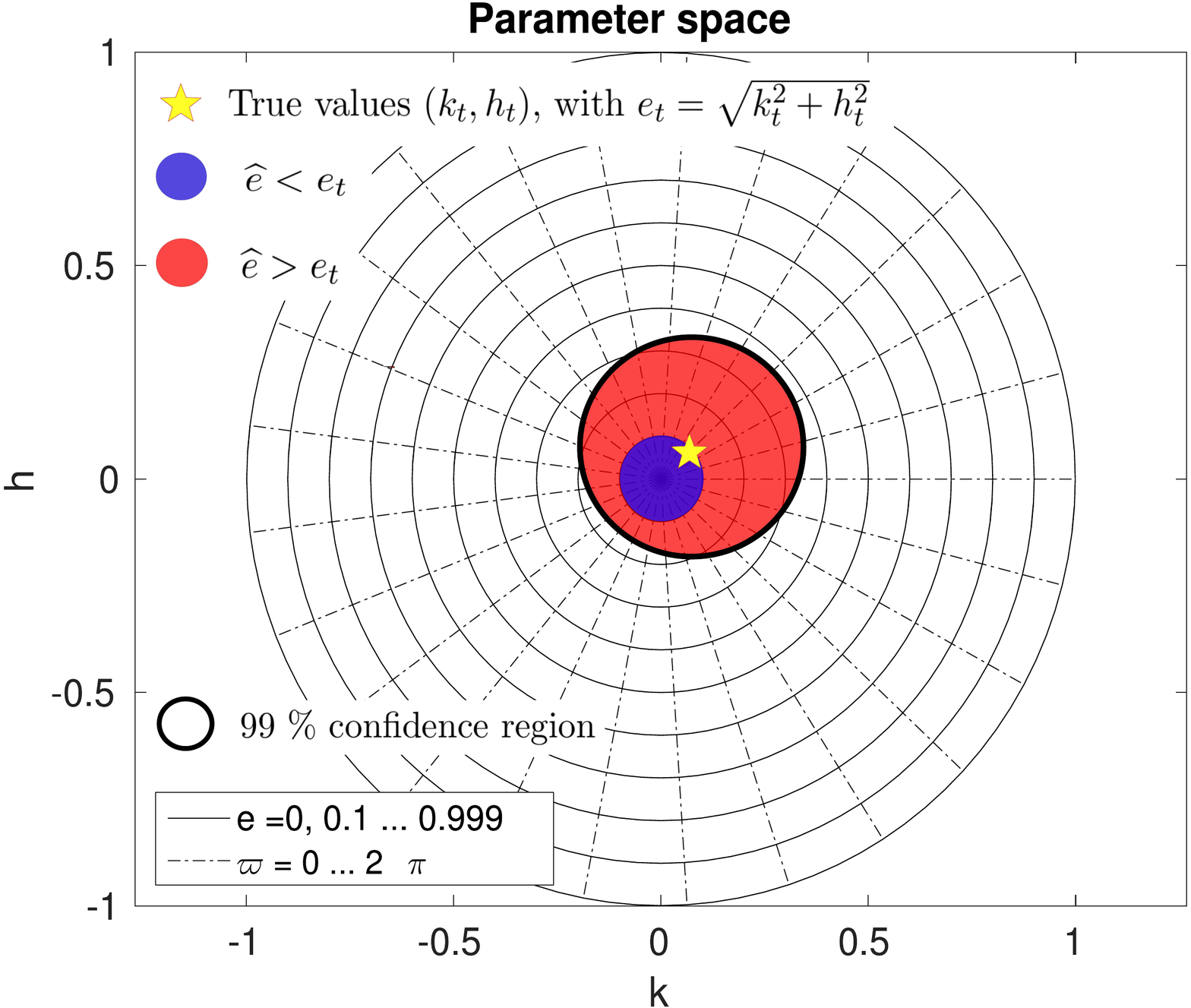}
 	\caption{Representation of the estimates $\widehat{k} = \widehat{e}\cos \widehat{\omega}$ and $\widehat{k} = \widehat{e}\sin \widehat{\omega}$ when $\widehat{k}$ and $\widehat{h}$ have a joint Gaussian distributions. The yellow star represents the true value of $k$ and $h$, the bold black line encircles the region where 99\% of the estimates are found. The red and blue region represent respectively the regions where the eccentricity estimates are overestimated and underestimated.}
 	\label{target1}
 \end{figure}

\subsubsection{Bias of the non linear least square}
\label{sec:nonlinls}
\REWRITE{
A common parameter estimator is the maximum likelihood $\widehat{ \theta}_{\mathrm{ML}}$. For the model of eq.~\eqref{eq:likelihood}, 
	\begin{align}
	\widehat{ \theta}_{\mathrm{ML}} =  \mathrm{arg} \max\limits_{\btheta \in \Theta, \bbeta \in B} p(\mathbfit{y}|\btheta, \bbeta).
	\label{eq:ml}
	\end{align}
When the parameters of the covariance, $\bbeta$, are fixed, maximising the likelihood comes down to the least square problem,  
\begin{align}
\widehat{ \theta}_{LS}(\mathbfit{y})  = \mathrm{arg} \min\limits_{\btheta \in \Theta} (\mathbfit{y} -  \mathbfit{f}(\btheta) )^T\mathbfss{V}^{-1}(\mathbfit{y}- \mathbfit{f}(\btheta) ) .
\label{eq:ls_estimate}
\end{align}
When the model $\mathbfit{f}(\btheta) $ depends linearly on $\btheta$, the least square estimate is unbiased. This is in general untrue when $\mathbfit{f}(\btheta)$ is non linear, which has been noted for instance by~\cite{hurwicz1950} and discussed in~\cite{hartley1964, bates1980, cook1985, firth1993}. The model we are concerned with (eq.~\eqref{eq:vexprbis} -~\eqref{eq:keplereq2bis}) is non linear, and indeed eccentricities obtained by least square are biased. 
}


 \subsubsection{Bias and uncertainty at low eccentricity}
 \label{sec:moremodels}
\label{sec:firstorder}

\REWRITE{
We begin with a numerical experiment. We generate Keplerian signals of eccentricity 0.06, 0.5 and 0.9 with fixed $\omega$, $M_0$ and $K = 3$ m.s\textsuperscript{-1}. The array of time $\mathbfit{t}$ is the 30 first measurements of GJ 876~\citep{correia2010}. We generate 1000 realisations of a white Gaussian noise with standard deviation of 1 m.s\textsuperscript{-1}.} For each realisation of the noise, a non-linear Keplerian model and a constant are fitted with a Levenberg-Marquardt algorithm. \REWRITE{The minimisation step is scaled so that the eccentricity never exceeds 0.999}.  The values of the estimates ($\widehat{k}:= \widehat{e\cos \omega},  \widehat{h}:= \widehat{e\sin \omega} $) are reported as crosses (red, purple and blue resp. for $e = 0.06$, 0.5 and 0.9) in Fig.~\ref{target2}.
The distributions of $\widehat{k}  $ and $\widehat{h}$ are fairly isotropic for $e_t=0.06$ and $e_t = 0.5$. For $e = 0.9$ there seems to be more complicated phenomena at work. The distribution has no circular symmetry and in some cases $\widehat{e}$ \REWRITE{is stuck at its maximal value, 0.999} (see section~\ref{sec:complicatedshape_body}). 

  In appendix~\ref{appendix_realformula}, it is shown that provided $e$ is small enough ($\leqslant 0.2$) and the number of observation is sufficient, $\widehat{k}$ and $\widehat{h}$,  follow independent Gaussian laws of same variance.  \REWRITE{This property allows us to understand the bias qualitatively.}
In Fig.~\ref{target1}, the distribution of eccentricity estimates is represented. The pair of true values $k_t,h_t$ is represented by a yellow star.  The bold black line delimits the region where 99$\%$ of the estimates are located. When the estimate falls in the blue-coloured region, the eccentricity is under-estimated. When it falls in the red-coloured region, the eccentricity is over-estimated.  As the volume of higher eccentricity models is larger in the vicinity of $k_t,h_t$, the eccentricity is more probably over-estimated. Informally, there are more and more models with eccentricity $e$ as $e$ grows.

\REWRITE{Furthermore, it is possible to obtain an analytical approximation of the bias.} Since $\widehat{k}$ and $\widehat{h}$ approximately follow a joint Gaussian distribution with same variances, $\widehat{e}= (\widehat{k}^2 + \widehat{h}^2)^{1/2}$ follows a Rice distribution, as noted in~\cite{shenturner2008}. Interestingly enough, the Rice distribution appears as a very good model for the eccentricity densities of the inner planets of the Solar System, resulting from chaotic diffusion~\citep{laskar2008}. Supposing that the measurement noise is white with standard deviation $\sigma$ and that $p$ parameters are fitted, within our approximation, $\widehat{k}$ and $\widehat{h}$ have a standard deviation $\sigma_k = \sigma_h = \sigma/K_t \sqrt{2/(N-p)}$. 
Defining the SNR $S$ as
\begin{align}
S := \frac{1}{\sigma_k } = \frac{K_t}{\sigma} \sqrt{\frac{N-p}{2}}
\label{eq:snr}
\end{align}
and denoting by $e_t$ the true eccentricity, the bias is
\begin{align}
\label{eq:bias1}
b(e_t,S) \approx  \frac{1}{S}\sqrt{\frac{\pi}{2}} L_{1/2}\left( -\frac{S^2 e_t^2}{2}\right) - e_t 
\end{align}
where $L_{1/2}$ is the Laguerre polynomial of order $1/2$. When $e_t = 0$, the eccentricity follows a Rayleigh distribution and eq.~\eqref{eq:bias1} reduces to a very simple expression,
\begin{align}
\label{eq:bias00}
b(0,S) \approx \sqrt{\frac{\pi}{2}} \frac{1}{S} = \frac{\sigma}{K_t} \sqrt{\frac{\pi}{N-p}} = \sqrt{\frac{\pi}{4-\pi}} \sigma_{\widehat{e}}.
\end{align}
Eq.~\eqref{eq:bias00} is identical to equation (18) of~\cite{lucy1971} except that  we are able to derive the effect of the correlations between parameters on the SNR through the term $-p$ (see Appendix~\ref{appendix_realformula} for justification).

Formula~\eqref{eq:bias00} is useful to see a few trends: the bias is proportional to the uncertainty on $k$ and $h$, which is proportional to the inverse of $K_t$ and $\sqrt{N-p}$. As a consequence, the bias increases as the SNR decreases, i.e. as $\sigma$ increases or as $K_t$ or $N$ decrease. This is also found by simulations in~\cite{shenturner2008} and~\cite{zakamska2011} for Bayesian estimates. We add that increasing the number of fitted parameters  $p$, increases the bias. 
 
 There are particular cases where the correlations between parameters drastically increase the uncertainties on $k$ and $h$ and therefore increase the bias, so that formula~\eqref{eq:bias1} should be taken as a lower bound.
\REWRITE{
	 However, the fact that the bias is approximately proportional to the standard deviation with a factor $\sqrt{\pi/(4-\pi)}$, as in eq.~\eqref{eq:bias00} stays true. This fact is remarkable because it means that the accuracy of the estimate (seen as the mean squared error~\eqref{eq:defmse}) is proportional to its precision (seen as the standard deviation $\sigma_{\widehat{ e}}$). In appendix~\ref{app:lindependency}, we show that poor phase coverage or short observational baseline affect the accuracy of the eccentricity estimate insofar as they decrease the precision of the estimate.  
  }

\subsubsection{Local minima at high eccentricities}
\label{sec:complicatedshape_body}

As shown in~\cite{baluev2015}, the number of local $\chi^2$ minima increases significantly in the high eccentricity region. These minima might lead a local minimisation algorithm or a Monte Carlo Markov Chain (MCMC) to be stuck in the wrong region of the parameter space. We here aim at quantifying  and understanding this feature. In this section, we provide a summary of our results, whose precise study is in appendix~\ref{app:complicatedshape}. These results are:
\begin{itemize}
\item An incorrect estimation of the noise level can lead to spurious deep local minima at high eccentricities.
\item As the SNR degrades, the probability of missing the global minimum by a non linear fit initialized on a circular orbit increases. 
\item There is a geometrical interpretation of the numerous local minima at high eccentricity: the set of models with fixed eccentricity explore more and more dimensions of the $N$-dimensional sample space as eccentricity grows. 
\end{itemize}

Let us illustrate the first point on Proxima b.
\cite{brown2017} re-analyses the data of Proxima b, a $\approx 1.27$ $\mathrm{M}_\oplus$ planet orbiting the M star Proxima Centauri with a period of 11.186 days~\citep{angladaescude2016},  and finds that there are local minima at eccentricity 0.75 and 0.95, the 0.95 eccentricity being the global least square fit. 

We compute a Keplerian periodogram~\citep{otoole2009, zechmeister2009, baluev2015} in the vicinity of the 11.186 period. That is Keplerian models (eq.~\eqref{eq:vexprbis} - \eqref{eq:keplereq2bis}) are fitted for a grid of periods, argument of periastron and eccentricity. We then represent per eccentricity the $\chi^2$ minimised over all other parameters. Using the nominal uncertainties, a single Keplerian model plus one offset per instrument, a linear and quadratic trend, we obtain the red curve in Fig.~\ref{fig:proximab}. 

The curve displays three local minima, the deepest being at eccentricity 0.92. However, let us note that 
the  best fit gives $\chi^2$ of 1057, while there are only 214 measurements. This gives a reduced $\chi^2$ of 5.16, which is unrealistic. We here simply \REWRITE{add a constant jitter term in quadrature with the nominal error bars} to obtain a reduced $\chi^2$ of 1 at the best fit. The minimum $\chi^2$ as a function of eccentricity so obtained is represented in blue in Fig.~\ref{fig:proximab}. The global minimum now occurs at $e = 0.17$. We interpret the global minima at eccentricity 0.92 as an artefact of an incorrect estimation of the error bars.

The same Keplerian periodogram calculations can be done on simulated data sets with different noise levels and different numbers of measurements. We simulate such systems with eccentricity sampled from a uniform distribution on $[0, 0.999]$, and count how many of them that have a SNR between 0 and 5, 5 and 10 etc. have $k=1,2,3...$ local minima. The histogram of Fig.~\ref{fig:nlminsfitcf_body} is obtained. Also, for each bin of SNR, we compute the proportion of  systems where the global minimum is not the closest to 0 (as in the case on Fig.~\ref{fig:proximab}, red curve) and therefore a local minimisation should miss the global minimum. \REWRITE{It appears that as the SNR increases, the fraction of cases where the global minimum is missed is decreasing, though not reaching zero.}

\begin{figure}
	\includegraphics[scale=0.37]{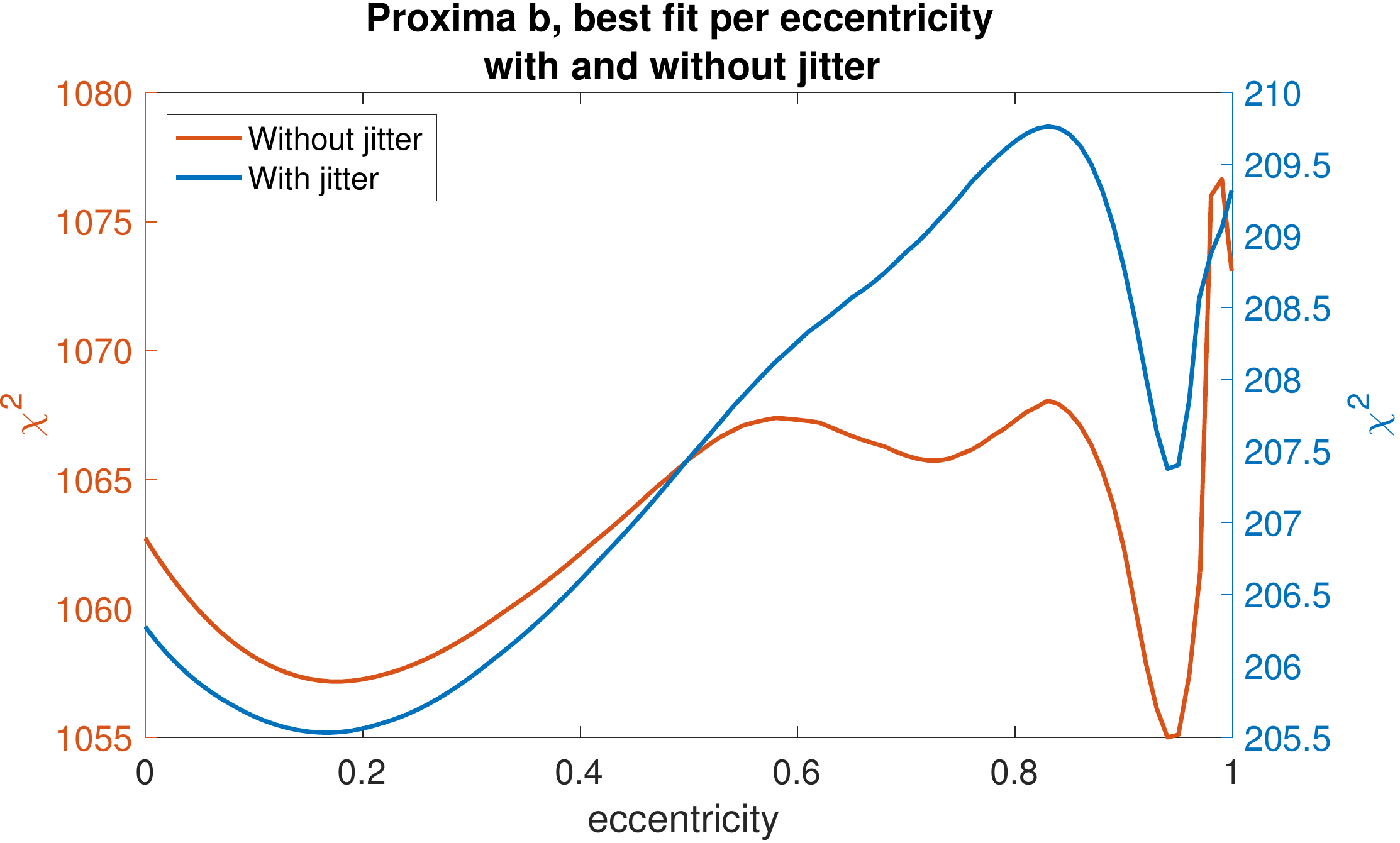}
	\caption{$\chi^2$ of the residuals of a Keplerian fit as a function of the eccentricity on Proxima b~\citep{angladaescude2016}. }
	\label{fig:proximab}
\end{figure}
\begin{figure}
	\includegraphics[width=8.5cm]{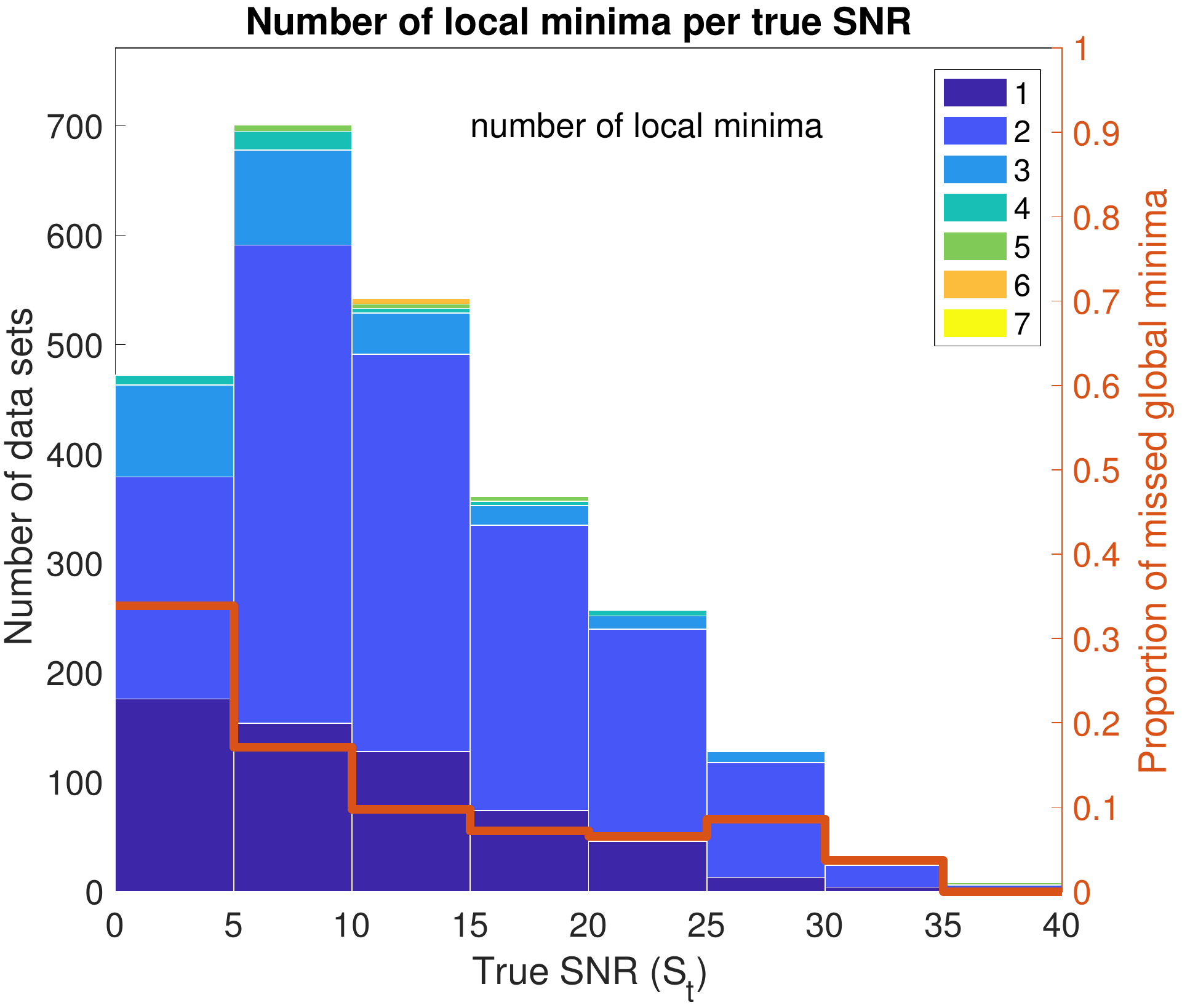}
	\caption{Blue bins: binned values of the number of systems that have a condition number lower than $10^7$ with 1,2,3,4,5, or 6 local minima. The bin size in fitted SNR $S_\mathrm{fit}$ of 5. Red curve: fraction of the binned systems where the global minimum is not attained at the one obtained with a linear fit.}
	\label{fig:nlminsfitcf_body}
\end{figure}

Let us now briefly comment on the geometrical interpretation of the higher number of local minima at high eccentricities. Finding the best fitting model  amounts to finding the closest model to the observation in a geometrical sense. 
We consider the figure drawn in $\mathbb{R}^N$ by all the  models that have an eccentricity $e$  and a period $P$, denoted by $\mathcal{M}_{e,P}$.  This figure might explore more or less dimensions. For instance, if it is close to a plane, it is nearly confined to a two-dimensional space. Otherwise, exploring many dimensions traduces a ``rough'' surface, which increases the chances of finding a local minimum of distance to the data. 
By a procedure based on singular value decomposition, detailed in appendix~\ref{app:complicatedshape} and in~Hara 2017 (PhD thesis), it is possible to obtain an approximate number of dimensions explored by $\mathcal{M}_{e,P}$ as a function of $e$. Table~\ref{table:dimension}  shows such calculations for the 214 measurement times of GJ 876~\citep{correia2010}. As eccentricity increases, $\mathcal{M}_{e,P}$ explores more and more dimensions.

\begin{table}
	\caption{Dimension of the models with fixed eccentricity as a function of the eccentricity, GJ876 measurement times}
	\begin{tabular}{p{1.3cm}|p{0.2cm}|p{0.2cm}|p{0.2cm}|p{0.2cm}|p{0.2cm}|p{0.2cm}|p{0.2cm}|p{0.2cm}|p{0.2cm}|p{0.2cm}}
		Eccentricity&  0.1 & 0.2 & 0.3 & 0.4 & 0.5 & 0.6 & 0.7 & 0.8 & 0.9 & 0.999 \\ \hline
		Dimension of $\mathcal{M}_{e,P}$ & 3 & 4 & 4 & 6 & 8 & 10& 14 & 24 & 46 & 91\\
	\end{tabular}
\label{table:dimension} 
\end{table}

\subsubsection{Interval estimates}
\label{sec:intervalest}

As said in section~\ref{sec:introduction}, one does not only want to obtain a value of the eccentricity with error bars, but also to test whether a certain value of the eccentricity is compatible with the data. This can be done in the frequentist setting with interval estimates. Since the focus will be put on Bayesian estimates in section~\ref{sec:robustness}, we simply here give their definition and refer the reader to appendix~\ref{sec:fisheriansection} for their derivation and detailed study. 

The hypothesis that the eccentricity has a certain value $e$ is rejected with a confidence level $\alpha$ if \REWRITE{the likelihood ratio $LR_e$ satisfies}
\begin{align}
&LR_e :=  \frac{\max\limits_{\btheta \in \Theta_e} p(\mathbfit{y}|\btheta)  }{\max\limits_{\btheta \in \Theta} p(\mathbfit{y}|\btheta) } \leqslant \e^{-\frac{1}{2}\beta} \label{eq:lrt_body}   \; \; \; \; \; \;\text{,where }\\
&\beta := F_{\chi^2_\rho}^{-1}(1 - \alpha) \label{eq:lrt2_body} \\
&\rho := 2 + 2S'^2\frac{e^2}{1+e^2} - \frac{\pi e}{1+e^2}L_{\frac{1}{2}}\left(-\frac{S'^2}{2}\right)
L_{\frac{1}{2}}\left(-\frac{e^2S'^2}{2}\right) \label{eq:lrt3_body}.
\end{align}
where  $\Theta_e$ is the set of parameters that have all eccentricity $e$, $p(\mathbfit{y}|\btheta)$ is the likelihood,  $F_{\chi^2_\rho}^{-1}$ is the inverse cumulative distribution function of a $\chi^2$ law with $\rho$ degrees of freedom, $S' = (\sigma /K_t) \sqrt{2/N}$  and $L_{\frac{1}{2}} $ is the Laguerre polynomial of order $1/2$. We also fit a free jitter term so that the reduced $\chi^2$ equals one. Conversely, for a certain measured $LR_e$ the FAP$(e)$ of the hypothesis $e_t=e$ is defined as
\begin{align}
\mathrm{FAP}(e) = 1 - F_{\chi_\rho}^2(-2\ln LR_e)
\label{eq:FAPe}.
\end{align}
The confidence interval of eccentricity is the set of $e$ with $\mathrm{FAP}(e)$ greater than a certain threshold (for instance $0.05$).

\subsection{Posterior distributions}
\label{sec:bayesianestimates}
\subsubsection{Point estimators}

The previous sections are devoted to the study the least square eccentricity estimate. 
 However, the standard practice in the exoplanet community is rather to compute the posterior probability  $p(\btheta|\mathbfit{y}) = p(\mathbfit{y}| \btheta) p(\btheta) /p(\mathbfit{y}) $ of the orbital elements $\btheta$ using Monte Carlo Markov Chains algorithms~\citep[e.g.][]{ford2005, ford2006}. 

From such posterior distributions, one can compute the orbital elements corresponding to the maximum \textit{a posteriori} (MAP)
\begin{align}
\widehat{ \theta}_{\mathrm{MAP}}  =  \mathrm{arg} \max\limits_{\btheta \in \Theta} p(\btheta|\mathbfit{y}).
\label{eq:map}
\end{align}
For a simple Keplerian fit $\widehat{ \theta}_{\mathrm{MAP}}  = (K,e,P,\omega, M_0)^{\mathrm{MAP}}$, an eccentricity estimate is then obtained by  $e^{\mathrm{MAP}}$.
 Alternately, one can compute the marginal distribution
 \begin{align}
  \label{eq:mapmarginal}
 p(e|\mathbfit{y}) &= \int_{\Theta_e} p(e,\tilde{\btheta}|\mathbfit{y}) \dd \tilde{\btheta}= \frac{1}{p(\mathbfit{y})}  \int_{\Theta_e} p(\mathbfit{y} |e, \tilde{\btheta}) p(\tilde{\btheta}) \dd \tilde{\btheta} 
 \end{align}
 and its maximum, mean and median
  \begin{align}
 \label{eq:maxmarginal}
 \begin{split}
 e_\mathrm{max} = \arg \max\limits_{e\in[0,1]} p(e|\mathbfit{y})  \; \;& ; \; \; e_\mathrm{mean} =\mathrm{mean}(p(e|\mathbfit{y}) ) \\  e_\mathrm{med} = & \mathrm{median}(p(e|\mathbfit{y}) ) 
 \end{split}
 \end{align}
 where $\tilde{\btheta}$ is the vector of parameters that are not eccentricity, and $\Theta_e$ is the space of parameters where the eccentricity is equal to $e$.  
 A standard result is that $e_{\mathrm{mean}}$ is the estimator that minimizes the mean squared error (see eq.~\eqref{eq:defmse}). Also, $e_{\mathrm{med}}$ minimizes the mean absolute error (MAE), defined as 
\begin{align}
\mathrm{MAE} := \mathbb{E}\{ |\widehat{ e} - e_t|\}.
	\label{eq:defmae}
\end{align}

The estimators~\eqref{eq:maxmarginal} are in general less than the maximum likelihood or the maximum a posteriori.  This is shown with a numerical experiment. We generate a circular planet of semi-amplitude 3.5 m/s and 100 realizations of Gaussian white noise at 2 m/s. The estimates $e_{\mathrm{ML}}$ and $e_{\mathrm{mean}}$ (eq.~\eqref{eq:maxmarginal})  are computed with a Monte-Carlo Markov Chain (MCMC) algorithm. The sampler is the same as in~\cite{delisle2018}, based on the adaptive Metropolis sampler of~\cite{haario2001}. \REWRITE{ The model consists of a Keplerian, an offset and a free jitter term, with uniform priors on all parameters. }
Fig.~\ref{fig:diffmlbayes} shows experimental distribution function of $e_{\mathrm{ML}} - e_{\mathrm{mean}}$. The condition $e_{\mathrm{ML}} > e_{\mathrm{mean}}$ is verified in 79 cases out of 100, with a mean value of $e_{\mathrm{ML}} - e_{\mathrm{mean}}$ equal to 0.0944.
  
  The efficiency of the estimates~\eqref{eq:maxmarginal} is understandable in terms of trade-off between model simplicity and agreement with the data. \REWRITE{ In section~\ref{sec:moremodels} and~\ref{sec:complicatedshape_body}, it appeared that for eccentricities $e_t$<0.2, in the vicinity of the true model, there is a larger volume of models with eccentricities $e>e_t$ than $e<e_t$.%
  }
  The integration over the domain $\Theta_e$ in eq.~\eqref{eq:mapmarginal} penalizes the models with high complexity, which here are the high eccentricity models. This is comparable to the penalization of models with too many planets by the marginal likelihood in the context of exoplanets detection~\citep[e.g.][]{nelson2018}.  In section~\ref{sec:bayessection}, we argue in favour of reporting $e_{\mathrm{mean}}$ \REWRITE{and/or $e_{\mathrm{med}}$}, as $e_{\mathrm{max}}$ is too biased towards low eccentricities.
 \begin{figure}
	\centering
	\includegraphics[width= 8.4cm]{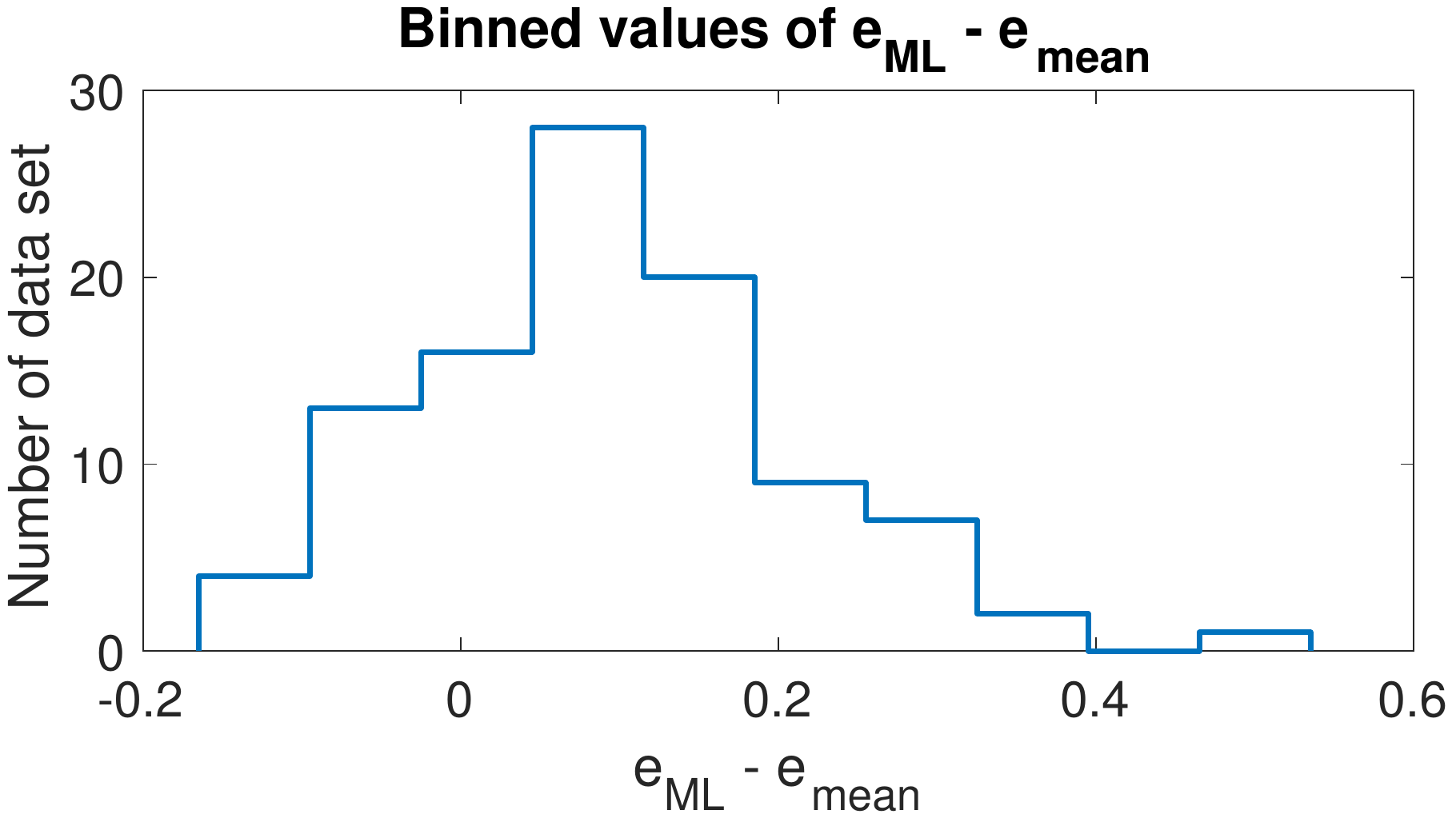}
	\caption{Binned values of the difference of maximum likelihood and posterior mean estimates, $e_{\mathrm{ML}}- e_{\mathrm{mean}}$ for 100 realisation of white noise and $K/\sigma=3.5$.}
	\label{fig:diffmlbayes}
\end{figure}

\subsubsection{Hypothesis testing}
\label{sec:bayessection}

Several of the works cited in section~\ref{sec:introduction} address the question of whether an eccentricity should be set to zero or not, which is a model selection problem. It can be addressed by computing the ratio of posterior likelihood, or odds ratio, of the models~\cite[e.g.][]{kassraftery1995},
\begin{align}
R = \frac{\mathrm{Pr}\{ e \neq 0 | \mathbfit{y}\}}{ \mathrm{Pr}\{ e = 0| \mathbfit{y} \} } = \frac{\mathrm{Pr}\{ \mathbfit{y}|e \neq 0 \}}{ \mathrm{Pr}\{\mathbfit{y}| e = 0\} } \frac{\mathrm{Pr}\{ e \neq 0 \} }{\mathrm{Pr}\{ e = 0 \}}
\label{eq:bayesratio}
\end{align}
where $\mathrm{Pr}\{ \mathbfit{y}|e \neq 0 \} = \int_{\btheta \in \Theta} p(\mathbfit{y}|\btheta) p(\btheta) d\btheta$. If this ratio is greater than a certain value, then one favours $e \neq 0$ over $e=0$. This methodology has been used in~\cite{bonomo2017a,bonomo2017b}. As the number of samples $N$ tends to infinity, assuming $\mathrm{Pr}\{ e \neq 0 \} =  \mathrm{Pr}\{ e = 0 \}$, the odds ratio is equivalent to the Bayesian Information Criterion~\citep[BIC,][]{schwarz1978}, as used by~\cite{pont2011, husnoo2011, husnoo2012}. 

More generally, one can compute a credible set, that is a set of $e$, denoted by $C\subset [0,1]$ such that
\begin{align}
\mathrm{Pr}\{ e \in C | \mathbfit{y} \}  = \int_C p(e|\mathbfit{y}) de  = \alpha 
\label{eq:credible}
\end{align}
where $\alpha \in [0, 1]$ is a probability. The set $C$ is in general taken as an interval but this need not be the case. Let us note that this approach is also a ratio of posterior likelihood as in Eq~\eqref{eq:bayesratio}. If $e=0$ is given a non null probability, then the prior probability takes the form $p(e) = p(0)\delta(0) + p(e)$ where $\delta$ is the Dirac function. 
We now focus on credible intervals, (eq.~\eqref{eq:credible}) since these are more widely used and give finer information.

\section{Robustness of eccentricity estimates}

\label{sec:robustness}
\subsection{Problem statement}
\label{sec:pbstate}

In section~\ref{sec:originecc}, several tools to make inferences on eccentricity were presented. We now study whether these are reliable even if the adopted model is incorrect.  What we call a model is a couple of prior and likelihood functions. We assume that the orbital elements are distributed according to a true prior $p_t(\btheta)$ and the observations have a true likelihood $p_t(\mathbfit{y}|\btheta)$ ($\mathcal{M}_t$) while the analysis is made with the model $\mathcal{M}$: $p(\btheta), p(\mathbfit{y}|\btheta)$.
The model $\mathcal{M}$ can be too simple: missed planetary signal, non modelled correlated noise or too complicated: for instance Gaussian processes are known to be very flexible, possibly too much.

The case where the model is too complicated will not be treated in detail. We simply point out that from section~\ref{sec:leastsquare}, we expect that the bias is higher than for a correct model. \REWRITE{Indeed, as the model grows in complexity the correlation between parameters increases, therefore the error on eccentricity $\sigma_e$ increases. This might be problematic as the bias is proportional to $\sigma_e$ (see eq.~\eqref{eq:bias00}. However, the error bars broaden, so that having too complex a model is unlikely to produce spurious conclusions. On the contrary, as we shall see, simplistic models can be problematic. }

\REWRITE{
In the following sections, we study the effect of the noise level estimate (\ref{sec:noiselevel}), numerical effects (\ref{sec:numericalerror}), incorrect noise models (\ref{sec:wronglikelihood}), priors (\ref{sec:priors}) and the comparison of two models: one eccentric planet or two planets in 2:1 mean motion resonance (\ref{sec:worstcase}). 
}

\subsection{Noise level}
\label{sec:noiselevel}

In section~\ref{sec:moremodels} and~\ref{sec:complicatedshape_body}, it appeared that an incorrect estimate of the noise norm leads to an underestimated bias and to spurious local minima at high eccentricity. 
As a consequence, it is key, as is standard practice, to adjust at least an extra jitter term $\sigma_J$ in the likelihood, 
\begin{align}
p(\mathbfit{y} | \btheta, \sigma_J) = \frac{1}{\sqrt{(2\pi)^N |\mathbf{V_0} + \sigma_J^2 \mathbf{I}| }} \e^{ -\frac{1}{2} (\mathbfit{y} - \mathbfit{f}(\btheta))^T  (\mathbf{V_0} + \sigma_J^2 \mathbf{I})^{-1}   (\mathbfit{y} - \mathbfit{f}(\btheta)) }
\label{eq:modeljitter}
\end{align}
where $\mathbf{V_0}$ is the nominal covariance, $\mathbf{I}$ the identity matrix and $\mathbfit{f}(\btheta)$ is the signal model, containing Keplerians and possibly other features. All the following analyses are made with the model~\eqref{eq:modeljitter}.

\subsection{Numerical effects}

\label{sec:numericalerror}

It appeared in section~\ref{sec:moremodels} that uncertainties on the eccentricity estimates increase the biases. This is also valid for the uncertainties stemming from the numerical methods used to compute the orbital elements.

As noted by~\cite{eastman2013}, there is a specific error in the implementation of the Metropolis-Hastings algorithm that worsens the bias, when the true eccentricity is \REWRITE{close to zero} and when the parameter space is parametrized by $(e,\omega)$ instead of $(k,h)$ or $(\tilde{k},\tilde{h})$. This error consists in not recording the value of a proposed parameter in the chain if it is rejected.

More generally, credible intervals or Bayes factors can be unreliable if the numerical schemes have not converged. 
~\cite{hogg2017} gives several ways to check for convergence of MCMC algorithms. In the following analyses, the convergence diagnostic is the number of effective samples, as computed in Appendix A of~\cite{delisle2018} (see also~\cite{sokal1997}). This number, $N_\mathrm{eff}$ is interpretable as an equivalent number of independent samples from the posterior distribution. Then, for instance, the numerical uncertainty on the mean of the marginal posterior $p(\theta_0|\mathbfit{y})$ of a parameter $\theta_0$ scales as $\sigma_p/\sqrt{N_\mathrm{eff}}$, where $\sigma_p$ is the standard deviation of $p(\theta_0|\mathbfit{y})$.

The first result we show is that at low SNR, the convergence is slower. This is likely due to the existence of local minima at high eccentricity  (see fig.~\ref{fig:nlminsfitcf_body}, section~\ref{sec:complicatedshape_body}). We simulate a one planet system on CoRoT-9 28 measurements~\citep{bonomo2017a}, the eccentricity is generated with a Beta distribution (a=0.867, b=3.03), angles are uniformly distributed and the period is fixed to 95 days.  The nominal measurement errors are normalised to obtain a mean variance of 1. A Gaussian white noise following the normalised nominal errors plus a one meter per second white noise is added. A hundred data sets with orbital elements and noise sampled according to their distributions are created. This is done for $K=2.5$ \REWRITE{m.s\textsuperscript{-1}} and $K=5$ \REWRITE{m.s\textsuperscript{-1}},  \REWRITE{which corresponds to SNR 4.2 and 8.5 (as defined in eq.~\eqref{eq:snr})}, so to obtain $2\times 100 = 200$ data sets.  Finally, a MCMC is performed on each of them with model~\eqref{eq:modeljitter} with 1,100,000 samples and a burn-in phase of a fourth of the total number of samples. The average number of independent samples is 7300 vs 12000, and the chain do no not reach a 1000 efficient samples in 28 and 5 cases respectively.  \REWRITE{In conclusion, low SNR signals should be treated with particular attention since the number of independent samples of the MCMC is smaller, while chains are initialized with a high eccentricity due to the higher bias of least square estimates at low SNR (see eq.~\eqref{eq:bias00}).} 

As an example of the problem associated with not having enough samples, this experiment is repeated with an injected eccentricity of 0. 
We consider that the ``small eccentricity'' hypothesis is rejected if $p(e\in [0,0.05]) < 0.05$. Let us consider the \REWRITE{samples obtained for the experiment described above with $K=2.5$ \REWRITE{m.s\textsuperscript{-1}}}. Depending on whether we take the first fourth of the samples or all of them, the fraction of rejection of the small eccentricity hypothesis goes respectively from  12 to 4 $\%$.  

In the following sections, simulations are taken into consideration if they reach  an effective number of independent samples greater than 1000, so that the mean of the eccentricity posterior (of variance $\sigma_e$) is known with at least a $\approx \sigma_e/\sqrt{N_{\mathrm{eff}}} \approx \sigma_e \times 3 \%$ accuracy. Since we compare different types of noise, we do not require an extremely good precision on the posterior distribution (also the following result stay approximately identical if we take runs with at least 2000, 5000 efficient samples). In practice, it is safer to ensure that 10,000 independent samples are reached to obtain a $\sigma_e/ \sqrt{10000} = \sigma_e \times 1\%$ accuracy on the posterior mean.

\subsection{Incorrect noise model}
 \label{sec:wronglikelihood}
 
 \subsubsection{Non Gaussian noise}
 \label{sec:nongausianbody}
 
 \REWRITE{
 The first question we address is whether, when the noise is non Gaussian, using the model~\eqref{eq:modeljitter} leads to spurious inferences on eccentricity. 
 For the sake of brevity, we here only report the results of the analysis done in appendix~\ref{sec:nongauss}.  The non Gaussianity of the noise can only lead to slightly underestimating or overestimating the error bars, when the noise distribution is very heavy-tailed. 
 We found that the estimates of the eccentricity and the error bars are mostly sensitive to the covariance of the noise and the following sections are focused on this aspect.
}

\subsubsection{Incorrect covariance: simulations }
\label{sec:incorrect_sim}

To study the effect of the noise covariance on eccentricity estimates, we proceed as follows. We consider the 28 measurement times of CoRoT-9, spanning on 1700 days, and generate a circular planet at 95 days, denoted by $y_\mathrm{planet}(t)$. We then generate 100 realisations of white noise, $y_\mathrm{noise}(t)$, and for each of them, a non modelled noise is added. Six different types of such noises, described below, are considered. In total, 6$\times$100 = 600 data sets are obtained. 
The signal generated is then of the form $y(t) = y_\mathrm{planet}(t)  + y_\mathrm{noise}(t)  + y_k(t)  $ with $k=0..5$. \REWRITE{ $ y_\mathrm{noise}(t) $ is generated according to the nominal uncertainties, which are CoRoT-9 uncertainties scaled so that their mean square is equal to 1 m.s\textsuperscript{-1}.}  
Such signals are generated in three contexts: with a semi-amplitude of the planet at 95 days of $2.5$, $3.5$ or $5$ \REWRITE{m.s\textsuperscript{-1}}. 
The methodology described aims at evaluating if, when the noise model is incorrect, circular planets tend to appear as eccentric. This experiment is also done with the eccentricity of the planet drawn from the same Beta distribution as the prior, to evaluate the impact of the noise on the estimates. The simulation where the eccentricity is fixed to zero and the one where it follows a beta distribution are respectively called $S_0$ and $S_e$.

In each simulation, on each of the 3$\times$600 data sets, the posterior distribution of the orbital elements is computed using the model~\eqref{eq:modeljitter} and priors given in table~\ref{tab:priorbayes}. $1,100,000$ samples are computed, the first fourth being the burn-in phase. The algorithm is an adaptive Metropolis algorithm as in~\cite{delisle2018}, and the convergence is checked by calculating the effective number of independent samples (see~\cite{delisle2018}, Appendix A).

Since our goal is to test the effect of the noise nature, and not its level, we impose that the norm of the non modelled noise is such that $\|\mathbfit{y}_0\| = \|\mathbfit{y}_1\| = ... \|\mathbfit{y}_5\| = \gamma $. For each of the 100 realisations of white noises, we draw $\gamma $ from a $\chi^2$ law with $N=28 $ degrees of freedom. The $y_k$ are defined as follows:
 
 \begin{itemize}
\item[$y_0$]: white, Gaussian noise identically distributed.
\item[$y_1$]: also white Gaussian noise but with different variances. The variances are drawn from a Laplace law so to obtain a wide range of value.
\item[$y_2$]: A circular planet that is too small to be fitted, the period of the added planet $P$ is drawn from a log-normal distribution until a period is found such that $P$ differs from 95 days of at least 20\%. This value is chosen to avoid the lowest probability region of period ratios of planet pairs found by Kepler~\citep{steffen2015}.
\item[$y_3$]: planet in resonance with the injected planet, in 1:2 or 3:2 resonance, inner or outer with probability 1/2. 
\item[$y_4$]: a Gaussian correlated noise with covariance $\kappa$
\begin{align}
\label{eq:kernel}
\kappa_{i,j} = \alpha^2 \exp\left[-\frac{1}{2}\left\{ \frac{\sin^2[\pi(t_i-t_j)/\tau]}{\lambda_p^2} + \frac{(t_i-t_j)^2}{\lambda_e^2}\right\}\right], 
\end{align}
as in~\cite{haywood2014}. We use the values of the Evidence Challenge~\citep{nelson2018}, $\alpha = \sqrt{3}$ m/s, $\lambda_e = 50.0$ days, $\lambda_p = 0.5$ (unitless), and $\tau = 20.0$ days.
\item[$y_5$]:  same as $y_4$ but with values $\alpha = \sqrt{3}$ m/s, $\lambda_e = 50.0$ days, $\lambda_p = 0.3$ (unitless), and $\tau = 30.0$ days.
 \end{itemize}
The rationale behind taking $y_1$ as such is to emulate the effect of mild outliers, not obvious enough to be completely discarded. As the Laplace distribution has heavy tails, it generates values of the variances that are very different from each other. $y_2$ and $y_3$ are chosen as such because the strongest resonances found in Kepler data are the 3:2 and 1:2 (or close to) resonances. In first approximation, the real period ratio distribution is a combination of these two artificial distributions. Finally, $y_4$ and $y_5$ are two types of stellar noises, the second one having a slightly stronger periodic component in the covariance. 

For the simulation where the eccentricity is fixed to zero, $S_0$, on each simulation we compute the probability $\mathrm{Pr}\{e \in [0, 0.05] \}$. We report the number of simulations where this quantity is below a threshold $\alpha=0.05$ and that have at least a thousand effective samples. 
The results are shown on Fig.~\ref{fig:wrongrej} for different values of the semi-amplitude of the input planet ($K = 2.5$, 3.5 or 5 \REWRITE{m.s\textsuperscript{-1}}, respectively blue, red and yellow markers). These values of $K$ were chosen to be close to relatively low SNRs. The small eccentricity rejection rate is approximately constant for the different noises except for a resonant companion, where the non zero eccentricity is, on average, rejected in 26\% of the cases, that is 13 times more often than for the other noise models. A non resonant companion also might induce an increased rate of rejected zero eccentricity (4.2\% on average vs 1.3\% for white noise).  The power of the stellar noise is stronger around 20 and 30 days, which is not on a harmonic of the period of the planet (95 days), and therefore they lead to an even lower rejection of the low eccentricity scenario (see section~\ref{sec:correlations}).

 In the simulation $S_e$, we compute the absolute value of the difference between the estimated eccentricity and the true one for the three estimators~\eqref{eq:maxmarginal}. In all cases, we find that the median and mean of the posterior distribution of eccentricity are very similar and largely outperform the other estimates in terms of mean squared error (MSE) (eq.~\eqref{eq:defmse}) and mean absolute error (MAE) eq.~\eqref{eq:defmae}). We have seen in section~\ref{sec:bayesianestimates} that for a correct model, the mean and median have respectively minimal MSE and MAE, this is indeed the case on the simulations with the white noise. The MSE and MAE over the 6 types of noise are reported in table~\ref{tab:estchoice}.  \REWRITE{The estimator $e_{\mathrm{mean}}$ and $e_{\mathrm{med}}$ are more accurate when the noise model is incorrect both for MSE and MAE, as a consequence we deem them as the best ones overall. }
\REWRITE{An argument in favour of $e_{\text{mean}}$ is that it has minimal MSE (eq.~\eqref{eq:defmse}). The MSE penalizes the square of the difference between the estimated and true value. Multiplying by two this difference multiplies by four the cost of the error. Therefore, estimators with small MSE are less likely to produce large errors.}
 
 This result seems in contradiction with~\cite{zakamska2011}, which finds the mode of the posterior to be less biased. However, they consider cases where the eccentricity is small. We also find that for small eccentricities,  $e_{\mathrm{max}}$ is less than $e_{\mathrm{med}}$ and $e_{\mathrm{mean}}$. Due to the Beta prior in eccentricity, it happens that the posterior is bi-modal and $e_{\mathrm{max}}=0$. When $e$ follows the prior distribution, $e_{\mathrm{med}}$ and $e_{\mathrm{mean}}$ are more accurate, at least in terms of MSE and MAE.

 We now only present in Fig.~\ref{fig:estimerr} the performance in terms of MAE of the estimator $e_{\mathrm{mean}}$ (circles) and the maximum likelihood estimator $e_{\mathrm{ML}}$ (crosses), for comparison (MSE behaves similarly). 
 \begin{table}
 \begin{tabular}{p{1.2cm}p{0.85cm}p{0.85cm}p{0.85cm}p{0.85cm}p{0.85cm}}
 		Estimator & $e_{\mathrm{ML}}$ & $e_{\mathrm{MAP}}$ & $e_{\mathrm{max}}$ &$e_{\mathrm{mean}}$& $e_{\mathrm{med}}$ \\ \hline
$\sqrt{\mathrm{MSE}}$&  0.1384  &  0.1362 &   0.1613  &  0.1088 &   0.1117 \\
 	MAE &	0.1071   & 0.1062  &  0.1239  &  0.0842 &   0.0849
 \end{tabular}
\caption{Root mean squared error ($\sqrt{\mathrm{MSE}}$) and mean absolute error (MAE) averaged over all types of noises of several estimators: maximum likelihood ($e_{\mathrm{ML}}$, max. \textit{a posteriori} $e_{\mathrm{MAP}}$, mode, mean and median of the eccentricity posterior $e_{\mathrm{max}}$ ,$e_{\mathrm{mean}}$ and $e_{\mathrm{med}}$ for the $K=3.5$\REWRITE{m.s\textsuperscript{-1}} simulation. }
\label{tab:estchoice}
 \end{table}
  In all cases $e_{\mathrm{mean}}$ shows better performance. The MAE does not heavily depend on the type of noise. However, we do observe a slight increase of the error for the unseen resonant companion. 
  Let us also note that as the signal semi-amplitude increases, the difference of performance between the maximum likelihood and $e_{\mathrm{ mean}}$ becomes less clear. 
 
 \begin{figure}
 	\includegraphics[width=8.4cm]{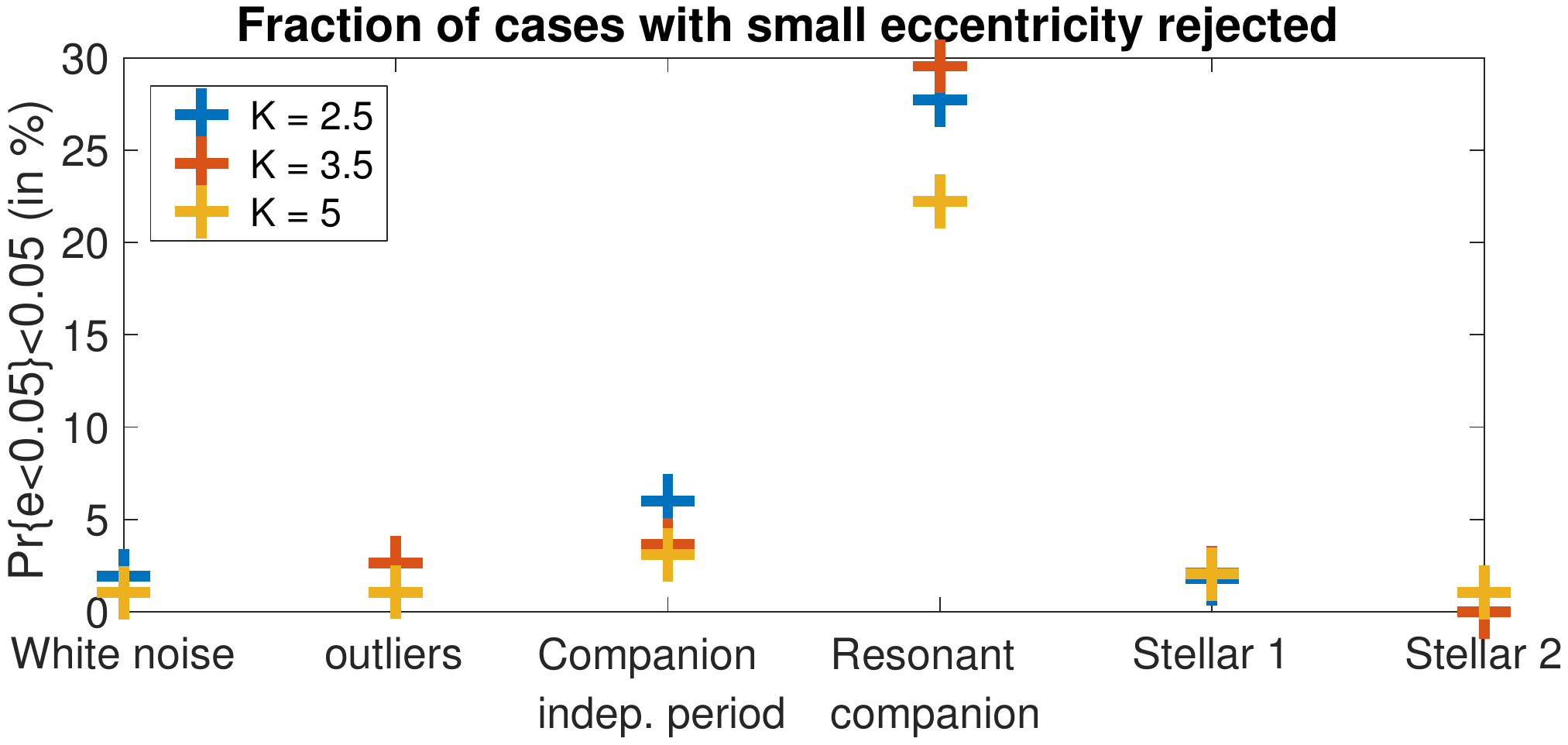}
 	\caption{Fraction of cases where the posterior probability of $e \in [0,0.05]$ is below 0.05 for an injected circular signal and different noises described in section~\ref{sec:incorrect_sim}. The blue, red and yellow points correspond to the experiment for $K$ = 2.5, 3.5 and 5 m.s\textsuperscript{-1} respectively. The posterior distribution is computed with the model~\eqref{eq:modeljitter}.}
 	\label{fig:wrongrej}
 \end{figure}
 \begin{figure}
 	\includegraphics[width=8.4cm]{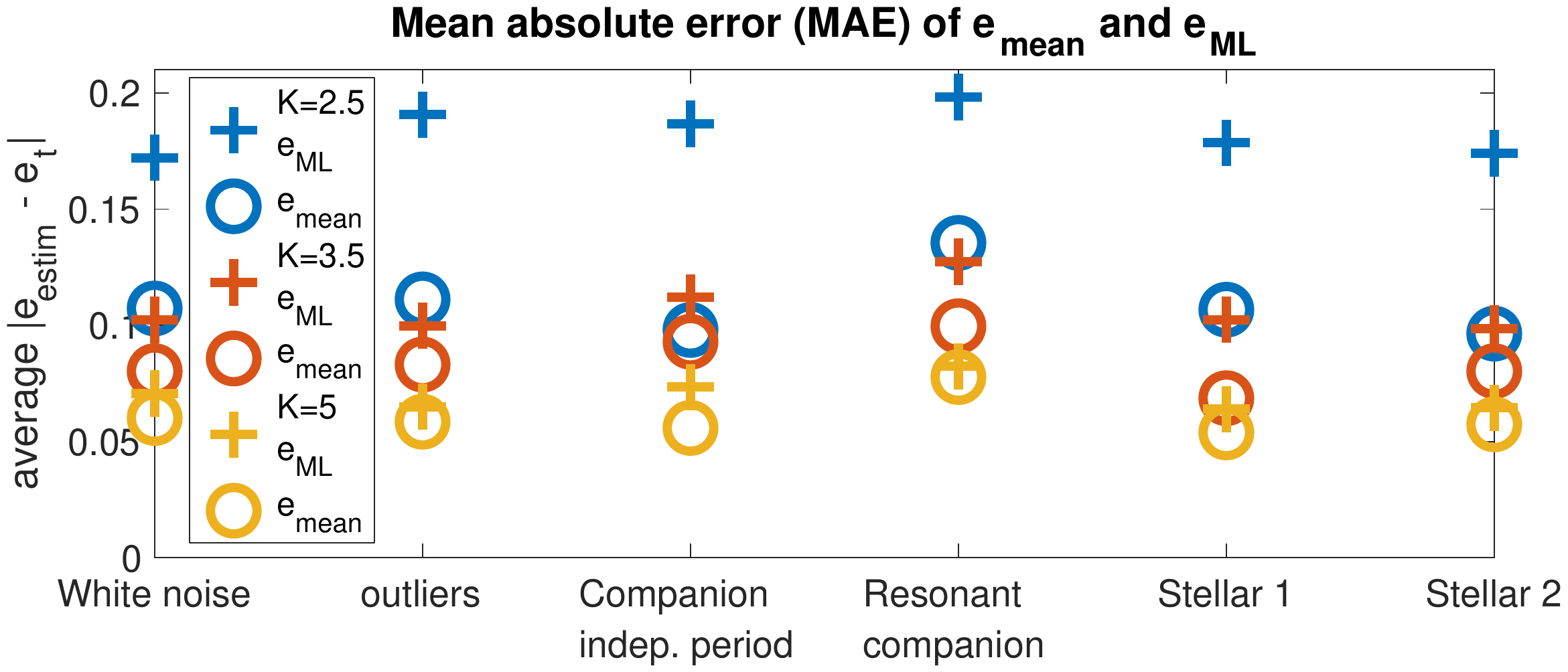}
 	\caption{Mean absolute error on eccentricity of the posterior mean ($e_{\mathrm{mean}}$, eq.~\eqref{eq:maxmarginal}) and maximum likelihood ($e_{\mathrm{ML}}$) eccentricity estimator. \REWRITE{The values of $K$ are given in m.s\textsuperscript{-1}.}}
 	\label{fig:estimerr}
 \end{figure}

\subsubsection{Interpretation: noise power at the planet semi-period}
\label{sec:correlations}

\REWRITE{	
In sections~\ref{sec:nongausianbody} we showed that the error on eccentricity is mainly determined by the true covariance of the noise. In section~\ref{sec:incorrect_sim} however, the simulated stellar noises did not yield particularly high errors on the eccentricity. We now show that the property of correlated noises most impacting eccentricity estimates is their power at the semi-period of the planet of interest.}

Let us consider a signal $\mathbfit{y} = \mathbfit{y}_0+ \bepsilon$, where $\mathbfit{y}_0$ is a circular orbit of period $P$, and $\bepsilon$ is an unknown stochastic signal, that the data analyst supposes to have covariance $\mathbfss{V}$. \REWRITE{Denoting by $\omega_0 : = 2\pi/P$, we define the $N\times 2$ matrix $ \mathbfss{M}_{k} = [\cos(k \omega_0 \mathbfit{t} ), \sin( k \omega_0 \mathbfit{t} ) ]$   and $\mathbfss{P}:= \mathbfss{M}_{k}  (\mathbfss{M}_{k}^T \mathbfss{V}^{-1} \mathbfss{M}_{k})^{-1} \mathbfss{M}_{k}^T \mathbfss{V}^{-1}$  the projection matrix onto the space spanned by the columns of $\mathbfss{M}_{k}$}. More generally, for a projection matrix $\mathbfss{P} $ onto a vector space $\mathbfss{M}$ we define
\begin{align}
R_\bepsilon(\mathbfss{M}) &=  \frac{(\mathbfss{P} \bepsilon)^T   \mathbfss{V}^{-1}   (\mathbfss{P} \bepsilon)  }{\bepsilon^T   \mathbfss{V}^{-1}   \bepsilon } \frac{N}{2} \label{eq:qepsilon} \\
Q_\bepsilon(\mathbfss{M})  &=  \mathbb{E}\left\{ Q_\bepsilon  \right\} \label{eq:repsilon}.
\end{align}
The rationale of defining these quantities is to identify if, assuming a covariance $\textbf{V}$ a noise is more correlated to a certain space than to its orthogonal. If $\bepsilon$ is a white noise, then it is not particularly correlated to any particular space, so that $R=1$. 
In the limit cases where $\bepsilon$ lies in, or is orthogonal to the space spanned by $ \mathbfss{M}_{k}$, then $R = N/2$ resp. 0. 

To test the influence of $Q_\bepsilon(\mathbfss{M}_{2})$ on the eccentricity error we proceed as follows. We consider an array of measurement times from a real system and generate a circular signal $y_0$  \REWRITE{plus a noise with a certain true covariance.  In order to obtain noises with very different spectral contents, we proceed as follows. }
 For a given frequency $\omega$, we draw thirty realizations of $\bepsilon = \cos( \omega \mathbfit{t} + \phi)$ , $\phi$ following a uniform distribution on $[0, 2 \pi]$. For each of them we compute $R_\bepsilon$, the least square estimate of the orbital elements $\widehat{\theta}$, the estimate of the noise level $\|\mathbfit{y} - \mathbfit{f}(\widehat{\theta}\|/\sqrt{N}$.
The $R_\bepsilon$ are averaged to have an estimate $\widehat{Q_\bepsilon}$ of $Q_\bepsilon$ and the average error on eccentricity $\langle |\widehat{e}- e_t|\rangle$. For each type of noise $\bepsilon$, we plot $(\widehat{Q_\bepsilon}, \langle |\widehat{e}- e_t|\rangle)$, which corresponds to a blue point in Fig.~\ref{fig:errorreps}. The point obtained with $\bepsilon$ being a white Gaussian noise model is represented with a yellow cross. We then bin the values with a constant step in $\log Q_\bepsilon$ and compute the average error, as well as its standard deviation (purple stair curve). We also compute the mean value of the estimated jitter, \REWRITER{and divide it by the value of the jitter estimated for $Q_\bepsilon = 1$. The normalized jitter so obtained is represented in green, with its scale on the right $y$ axis}. 

Fig.~\ref{fig:errorreps} is obtained with the time array of Gl 96 SOPHIE measurements (67 measurements), $e_t=0$, an assumed covariance matrix $\mathbfss{V}$ equal to identity, a period of 40 days and fixed ratio of the norm of the Keplerian signal and the input noise of 10 (\REWRITE{which corresponds to a very high SNR = 78}), to have as little influence as possible of the noise level. As the noise becomes more correlated with the $\mathbfss{M}_{2}$ space ($Q_\bepsilon$ increases) it is absorbed in the fit and the RMS of the residual decreases. Indeed in Fig.~\ref{fig:errorreps} it is apparent that as $Q_\bepsilon$ increases the error on eccentricity grows while the estimated level of the noise decreases. \REWRITER{The error on $e$ for $e_t = 0.9$ is consistently higher, which is likely due to local minima at high eccentricity (see section~\ref{sec:complicatedshape_body})}

One can then wonder how the SNR $S$ affects the bias on eccentricity for correlated noises. \REWRITER{Defining the noise ``quality factor'' as $q:=\sqrt{Q_\bepsilon(\mathbfss{M}_2)}$}, in the linear approximation used for~\eqref{eq:snr}, the uncertainty on $k$ and $h$ becomes $q\sigma_k$. With notations of equation~\eqref{eq:bias00}, the bias at $e_t=0$ is therefore approximately 
\begin{align}
	b(0,S) \approx \sqrt{\frac{\pi }{2}} \frac{q}{S}.
\end{align}
 For a given $q$, as the SNR increases, the bias decreases. 
 
To check if the power at semi period is also correlated with the error on eccentricity if the true eccentricity is high, we perform the same experiment with a value of the eccentricity equal to $0, 0.3,0.5,0.7,0.9$. Again, the ratio of the norm of the Keplerian signal and the input noise is fixed to 10. For each eccentricity the experiment is done 20 times with periods randomly drawn with a log-normal law. The results are shown in Fig.~\ref{fig:errorrepsmean}. We first remark that as $e$ increases, the bias, and therefore the error on eccentricity, decreases. Secondly, from $e = 0.5$, $Q_\bepsilon$ is less relevant to predict the effect of the noise on eccentricity estimates. This result is in accordance with~\cite{wyttenmeier2019},  who found that planets with e>0.5 are very unlikely to be mistaken for a two circular planet model. As eccentricity increases, the harmonics of order greater than two increase, so that a noise localised at the semi period cannot mimic a higher eccentricity. 

The same experiment is performed with a ratio of norm of the signal and the noise of three and an injected circular orbit. We compute $Q_\bepsilon(\mathbfss{M}_k)$  with $k=1/2, 2/3, 1, 3/2, 2, 3, 4, 5$. The results are plotted on Fig.~\ref{fig:errorharmo}, with a colour code for each $k$. It appears that a strong component of the noise on the harmonic 2, 3, 4 or 3/2 and 5 lead to an increased error on eccentricity (in decreasing order of effect). The noise level is notably underestimated in the $k=2$ case. On the contrary, a strong component of the noise on $k=0$ leads to a smaller error, which is easy to understand. Indeed, since the eccentricity of the injected signal is zero, the noise reinforces the signal.

\begin{figure*}
	\begin{minipage}[l]{0.48\textwidth} 
		\includegraphics[width=8.2cm]{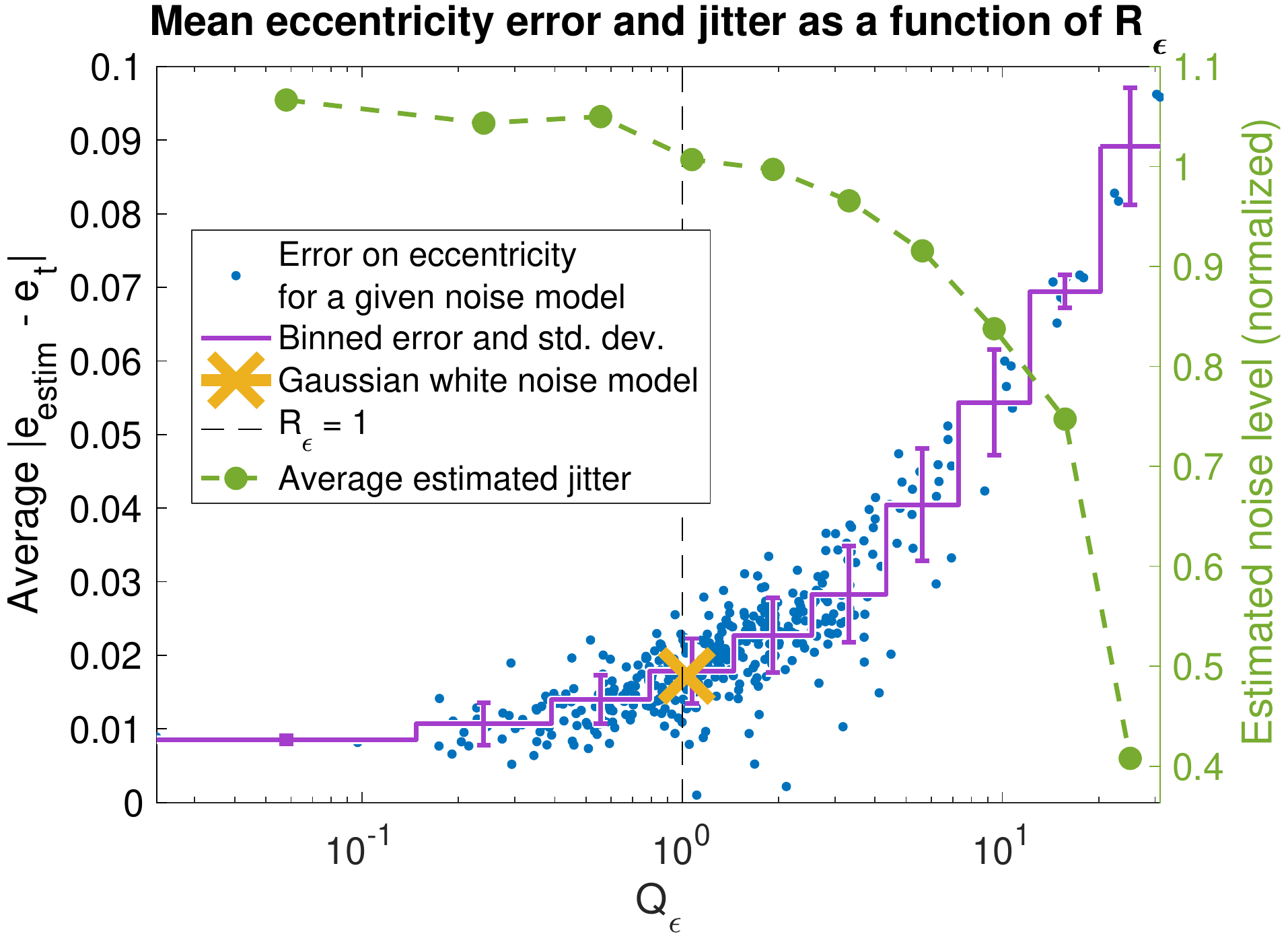}
		\caption{Error on eccentricity plotted against the estimated $Q_\bepsilon$ (see eq.~\eqref{eq:repsilon}) for different noise types (blue points) and Gaussian white noise (yellow cross). Same value averaged per interval of $Q_\bepsilon$ with standard deviations (purple stair curve). The estimates of noise level, averaged per bin, are represented in green.}
		\label{fig:errorreps}
	\end{minipage} \hfill
	\begin{minipage}[c]{0.48\textwidth}
		\includegraphics[width=8.2cm]{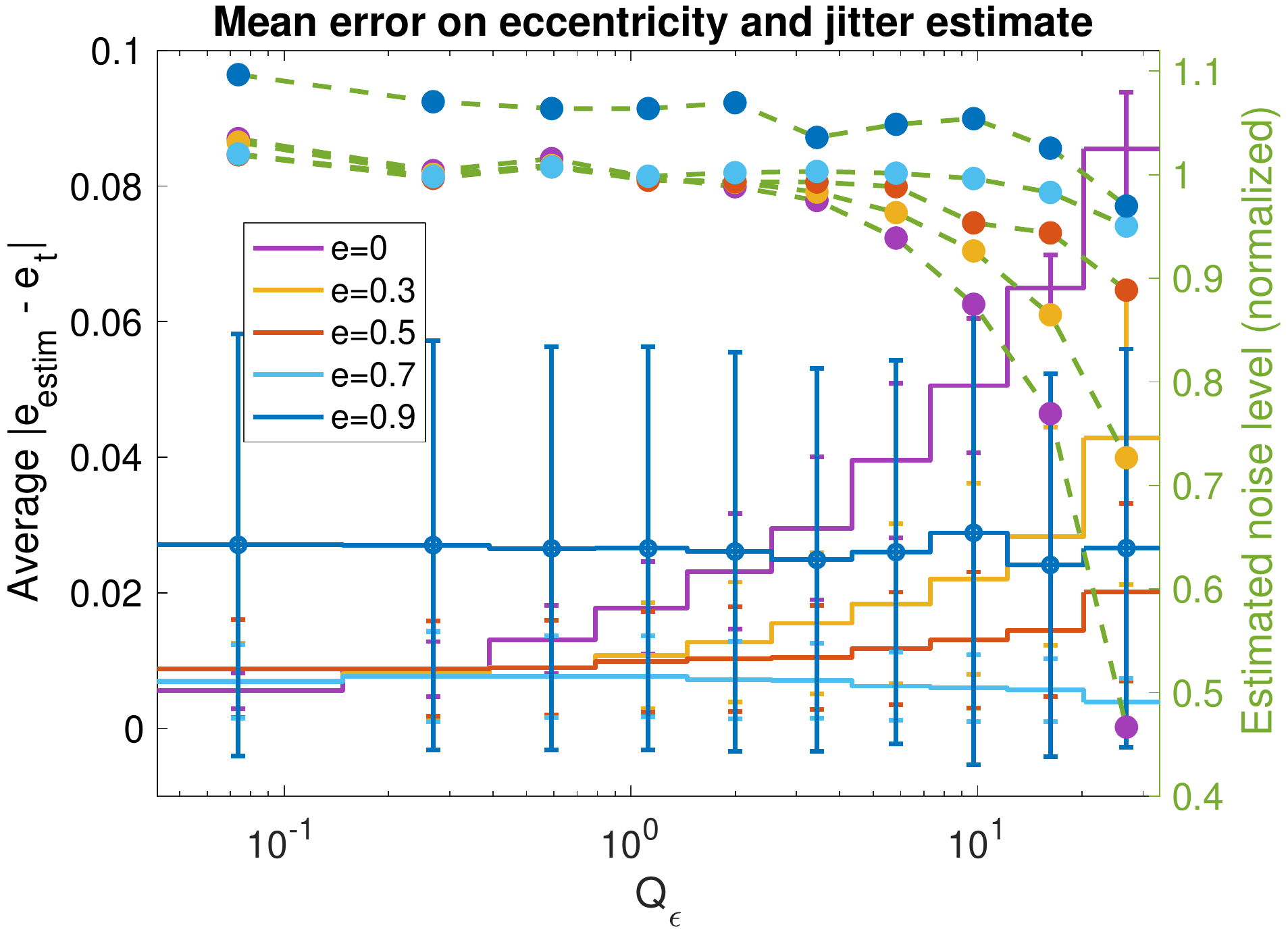}
		\caption{Error on eccentricity plotted against the estimated $Q_\bepsilon$ (see eq.~\eqref{eq:repsilon}) averaged for 10 different periods. The average error per $Q_\bepsilon$ and the estimated jitter for the different values of the true eccentricity are represented: $e=0,0.3,0.5,0.7,0.9$ (resp. purple, yellow, red, light blue, dark blue).}
		\label{fig:errorrepsmean}
	\end{minipage}\\ [0.6cm]
\end{figure*}
\begin{figure}
	\includegraphics[width=8.2cm]{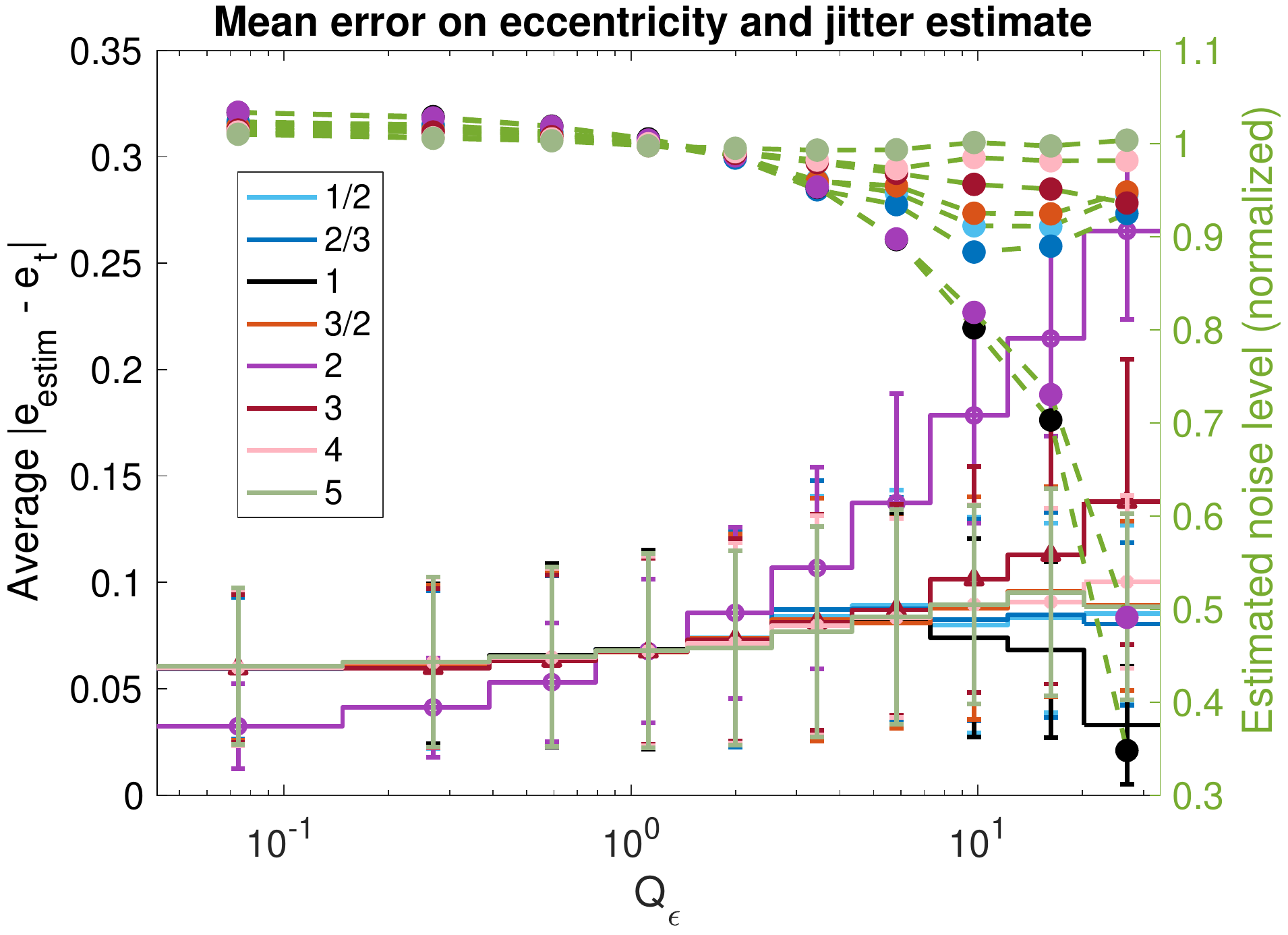}
	\caption{Error on eccentricity (solid stairs) and estimated error (round markers) plotted for a true circular orbit against the estimated $Q_\bepsilon(\mathbfss{P}_k)$ (see eq.~\eqref{eq:repsilon}) where $\mathbfss{M}_k$ are the vector spaces =$ (\cos k\omega_0 \mathbfit{t}, \sin k\omega_0 \mathbfit{t}) $ for $k$=0.5...5 (see legend for color code) and $\omega_0$ is the frequency of the input planet. }
	\label{fig:errorharmo}
\end{figure}

\REWRITE{
The results of section~\ref{sec:incorrect_sim} are interpretable with the analysis above. The metric $Q_\bepsilon$ defined in eq.~\eqref{eq:repsilon} is computed (here  $Q_\bepsilon(\mathbfss{M}_{2})$) for the noises of section~\ref{sec:incorrect_sim}, that are noises generated with nominal uncertainties plus $y_0,...y_5$.  For $K = 2.5$, 3.5 and 5 m.s\textsuperscript{-1}, we generate 10,000 realisations of these noises, and inject a circular planet at period $P$ with a random phase. We do this simulation for $P = 5$ to 100 days per step of 5 days. in each of the $ 6 \times 10,000 \times 20 $ simulations, we compute the FAP associated with $e=0$ with formula~\eqref{eq:FAPe}, and adopt as a convention that $e=0$ is rejected if the $p$-value is below 0.05. We also compute $R_\bepsilon(\mathbf{M}_2)$ as defined in eq.~\eqref{eq:qepsilon}. For the 6 noises and 20 periods, we average the values of the $R_\bepsilon$ to obtain $Q_\bepsilon$~\eqref{eq:repsilon}.  
 Fig.~\ref{fig:repsilon_norm} shows the proportion of false eccentricity rejected as a function of $Q_\bepsilon$ for $K = 5$ m.s\textsuperscript{-1}. It clearly appears that there is a strong correlation with the power of the noise at the semi period.  
}

Obviously, one can test on a given system if a specific type of noise has a particular impact on a planet with injected parameters, the goal of this section was to identify some generic properties of the noise that lead to spurious inferences.

\begin{figure}
	\hspace{-0.15cm}
		\includegraphics[width=7.5cm]{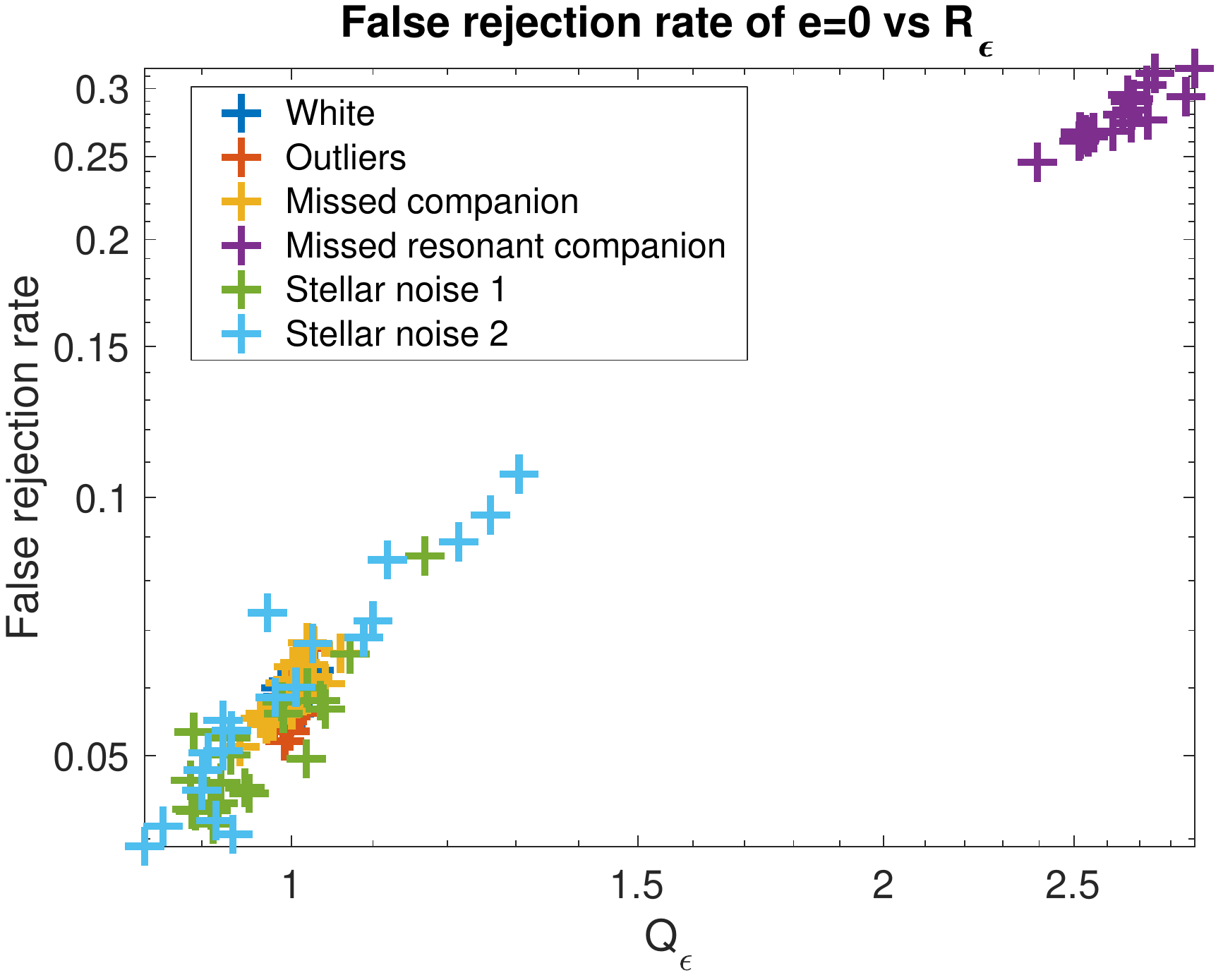}
		\caption{False rejection rate estimated with formula~\eqref{eq:lrt2_body} as a function of $Q_\bepsilon$ as defined in eq.~\eqref{eq:repsilon} for the six noises generated ($y_0,...y_5$) and period of the true planet equal to 5..5..100 days.}
	\label{fig:repsilon_norm}
\end{figure}

\subsection{Robustness to prior changes}
\label{sec:priors}
\subsubsection{Simulations}
\label{sec:sensitprior}
In the previous sections we have studied the impact of having a wrong likelihood function. We now turn to the sensitivity of the estimate on the prior probability, here with a numerical experiment.
With the formalism of section~\ref{sec:pbstate}, the data are generated with prior and likelihood $p_t(\btheta), p_t(\mathbfit{y}|\btheta)$, and the analysis is done with $p(\btheta), p_t(\mathbfit{y}|\btheta)$. The likelihood is correct, but the prior is not, which corresponds to having an incorrect idea of the population distribution. Note that the prior probability of all parameters have an effect on the eccentricity estimate, but we focus on the prior probability chosen for the eccentricity.

Two distributions are considered. We generate eccentricities according to the distribution Beta(a=0.867, b=3.03) and compute posterior probabilities with a uniform prior. Otherwise the priors are taken as in table~\ref{tab:priorbayes} and the data are generated as in section~\ref{sec:wronglikelihood}. The results are shown in table~\ref{tab:prior} as a function of the input amplitude. The errors, measured by MSE and MAE are systematically worse when using the incorrect prior, however \REWRITE{with an extra error not exceeding 15\%}. It seems like an error of $\approx 0.01$ on $e$ could be counted as uncertainty on the prior distribution.

The uncertainty on the prior seems not to be a major concern, at least for the estimation of eccentricity. However, in some cases, one might want to recompute the credible interval with another prior, which is the object of the next section.

\begin{table}
	\begin{tabular}{p{1.2cm}| p{1.2cm}|p{1.2cm}|p{1.2cm}p{1.2cm}}
		Estimator & prior	& K=2 m.s\textsuperscript{-1}& K=3.5 m.s\textsuperscript{-1} & K= 5 m.s\textsuperscript{-1} \\ \hline
		$\sqrt{\mathrm{MSE}}$& correct 	&0.1325 & 0.0994& 0.0791 \\ 
		$\sqrt{\mathrm{MSE}}$& incorrect 	&0.1530 & 0.1118& 0.0899 \\ \hline\hline
		MAE& correct 	&0.1073 & 0.0803& 0.0601 \\
		MAE&  incorrect & 0.1188&  0.0851&  0.0713 
	\end{tabular}
	\caption{Mean absolute and mean square error (MAE and MSE) of the estimate when the eccentricities are generated with a Beta distributions, and the analysis is done with the same Beta distribution as prior (correct model) or done with a uniform prior (incorrect model)}
	\label{tab:prior}
\end{table}

\subsubsection{Recomputing the posterior without new sampling}
\label{sec:recomputepriors}
The most straightforward way to explore the dependency of the posterior on the prior is to recompute it with another prior distribution. However, this might be lengthy to do it systematically on several systems.
We here propose an alternative which consists in multiplying the prior by a constant on a subset of its domain of definition, and to scale it elsewhere. In so doing, the output of the posterior sampler can be used straightforwardly without doing any sampling. In the following we illustrate the process with the prior on eccentricity.

We consider a measurable subset $D$ of $[0,1]$ and its complement  $\bar{D}$ in $[0,1]$. For instance an interval $D = [0, e_0]$ for some arbitrary $e_0 \in (0,1]$. Let us denote by $\tilde{\btheta}$ the model parameters other than eccentricity, so $ \btheta = (e,\tilde{\btheta})$.
We define a new prior $p'$ such that for $e \in D$, $p'(e, \tilde{\btheta}) = x p(e, \tilde{\btheta})$. To ensure that $\int_0^1 \int p'(e, \tilde{\btheta}) \dd \tilde{\btheta}\dd e= 1$, we take
\begin{align}
\label{eq:a}
a &= \int_D \int p(e, \tilde{\btheta})  \dd \tilde{\btheta} \dd e. \\
x &\in  [0, (1-a)/a] \\
p'(e,\tilde{\btheta}) &= (1-ax)/(1-a) p(e, \tilde{\btheta}) \; \;  \text{ for } \; \; e \in\bar{D} 
\end{align}
We now want to compute the probability that $e\in C$ for a prior distribution $p'(e,\tilde{\btheta})$. This one is given by replacing $p$ by $p'$ in equation~\eqref{eq:credible}. With the notations  
\begin{align}
z &:= (1-ax)/(1-a)\\
I_E &:= \int_{e \in E} \int_{\tilde{\btheta}}  \frac{p(\mathbfit{y}|e, \tilde{\btheta}) p(e, \tilde{\btheta})}{p(y)} \dd \tilde{\btheta} \dd e
\end{align}
for $E \subset [0,1]$, we compute
\begin{align}
\begin{split}
\mathrm{Pr}\{ e \in C | \mathbfit{y}, p' \}  &=  \int_{C} \int_{\tilde{\btheta}} p'(e|\mathbfit{y})    \dd \tilde{\btheta}\dd e\\
& =  \frac{ x \; I_{C\cap D} + z \; I_{C\cap \bar{D}} }{ x \; I_D + z \;  I_{\bar{D}}} 
\label{eq:credible2}
\end{split}
\end{align}
where $p'(e|\mathbfit{y})$ is the posterior distribution when the prior is $p'(e)$ and $p(\mathbfit{y} |e)$ is the likelihood marginalized on all parameters but eccentricity.
When $x=0$, all the prior probability goes to the complementary of $D$ and $\mathrm{Pr}\{ e \in D | \mathbfit{y} \}$ goes to 0. If $x = (1-a)/a$, $e$ is certainly in $D$ therefore $\mathrm{Pr}\{ e \in D | \mathbfit{y} \}=1$. 

The advantage of this calculation is that the integrals can be computed from the posterior samples. Denoting, $| E |$ the number of the MCMC samples that are such that $e\in E \subset [0,1]$ and $N_{\mathrm{tot}}$ the total number of samples, an estimate of $I_E$ is $\widehat{I_E} = |E|/N_{\mathrm{tot}}$, so that an estimate of~\eqref{eq:credible2} is
\begin{align}
\widehat{\mathrm{Pr}}\{ e \in C | \mathbfit{y} , p' \} = \frac{x \; |C \cap D |+ z \; | C \cap \bar{D} | }{x \; |D  |+ z | \bar{D}|}.
\label{eq:credible3}
\end{align}
The reasoning can be extended straightforwardly to credible regions $D$ and $C$ in the parameter space, and to prior region subdivisions in $D_1...D_q$ with disjoint $(D_i)_{i=1..q}$ whose union is the whole parameter space. 

Apart from the Markov chain samples, the only quantity needed to use~\eqref{eq:credible3} is $a$ as given by~\eqref{eq:a}. This expression might be difficult to compute in general, but in the case where $p(e,\tilde{\btheta}) = p(e)p(\tilde{\btheta}) $, $a = \int_D p(e) \dd e $, which is a one-dimensional integral. Analytic expressions might exist and a Riemann integration is always possible.

Since the integrals to be evaluated from posterior samples are random variables, it must be ensured that they have a controlled uncertainty. When breaking the posterior in many domains $D_1...D_q$, the procedure outlined may become unreliable if there are not enough independent samples in each $D_k$. One can easily compute the effective number of samples in each $D_k$, $N_{\mathrm{eff},k}$ by counting how many samples are in that region and dividing by the correlation time-scale. A number of effective samples greater than $n$ gives an accuracy of $\approx 1/\sqrt{n} \times 100 \%$ on the probability $\mathrm{Pr}\{e \in D_k | \mathbfit{y}\}$.  Further investigation is left for future work.

\subsection{Model comparison: one eccentric planet or 2:1 mean motion resonance}
\label{sec:worstcase}
 
\REWRITE{A system of two planets in 2:1 mean motion resonance can be mistaken for one eccentric planet, and vice versa.}
We here study the possibility to disentangle the two  cases via Bayes factor as a function of the SNR.  Two models are considered, $\mathcal{M}_e$ and $\mathcal{M}_{1:2}$, respectively an eccentric planet and two circular planets in mean motion resonance:
\begin{align}
\mathcal{M}_e: \;\;\;   y(t) &= K(\cos (\nu + \omega) + e\cos(\omega) )+ g(\tilde{\btheta}) + \  \epsilon \\ 
\mathcal{M}_{1:2}: \;\;\;    y(t) &= K_1\cos \left(\frac{2\pi}{P_1} t + \phi_{01}\right)  + K_2\cos \left(\frac{2\pi}{P_2} t + \phi_{02}\right) \\ & \; \; \; \; + g(\tilde{\btheta}) + \  \epsilon \nonumber
\end{align}
Denoting by $g(\tilde{\btheta})$ a deterministic model encapsulating other planets, offsets, trends etc.. We also let vary a jitter term $\sigma_J$ as in eq.~\eqref{eq:likelihood}. The Bayes factor of the two models is defined as
\begin{align}
B = \frac{p(\mathbfit{y} | \mathcal{M}_e)}{p(\mathbfit{y}  | \mathcal{M}_{2:1})} .
\end{align}
where 
\begin{align}
p(\mathbfit{y}  | \mathcal{M})   = \int  p(\mathbfit{y}| \btheta)     p(\theta) \dd \btheta
\end{align}
with $\btheta = (K,k,h,P,M_0,\tilde{\btheta}, \sigma_J)$ or $\btheta = (K_1,P_1,M_{01}, K_2,P_2,M_{02},\tilde{\btheta}, \sigma_J)$ for $\mathcal{M} = \mathcal{M}_e$ and $\mathcal{M} = \mathcal{M}_{2:1}$ respectively.
\REWRITE{We expect the two models to be distinguishable if the amplitude of the second harmonic of the signal can be resolved~\citep[][see eq. 5]{angladaescude2010}.}

In order to determine at which SNR the Bayes factor allows to disentangle resonant planets and eccentric ones, we perform a numerical experiment. We select a semi-amplitude $K$ and a period $P$, then generate a Keplerian signal with $e=0.25$, random $M_0$ and $\omega$. $K$ is chosen on a grid (2,5,8,11,14 m.s\textsuperscript{-1}), to see how the ability to disentangle scenarios evolves with the true SNR.

On the other hand, we generate a two planet system. The outer planet has period $P$ with random phase and semi-amplitude $K_1 = K$. The inner planet is circular with $K_2 = K_1/4$, so that $K_2/K_1 = e$ of the single planet.  The phase of the second planet is chosen uniformly. The simulation is performed in two different settings. In the first one, the period ratio is fixed to $P_2 = P_1/2$. In the second,
 $1/P_2$ is chosen uniformly between $[(1-\alpha)2/P_1 , (1+\alpha)2/P_1 ]$ with $\alpha = 0.1$. The rationale behind the choice of $P_2$ is that the period ratios of Kepler planets are located within a neighbourhood of 2:1~\citep{steffen2015}. 

	The priors chosen to compute the Bayes factor are summarized in table~\ref{tab:priorbayesf}.
The Bayes factor is computed with the nested sampling algorithm PolyChord~\citep{handley2015b, handley2015}. The performance of the algorithm was checked on the data sets of the Evidence Challenge~\citep{nelson2018}. 
\REWRITE{For each data set, the algorithm is ran at least five times. The $\ln \mathcal{Z}$ estimate is taken as the median of the different runs and the error bars are given by the variance of the empirical median as provided by~\cite{kenney1962}. }
The error on the log Bayes factor $\ln$BF = $\ln \mathcal{Z}_2 - \ln \mathcal{Z}_1 $ is taken as $(\sigma_{\mathcal{Z}_1}^2 + \sigma_{\mathcal{Z}_2}^2)^{1/2}$.

Fig.~\ref{fig:bayesfact} and ~\ref{fig:bayesfact2} show the results of the simulation, where the measurement time arrays are those of CoRoT-9~\citep{bonomo2017a} and Gl 96~\citep{hobson2018}. The $\ln$ of the Bayes factor is represented, such that the correct model is always at the numerator. The black dashed lines indicate a Bayes factor equal to 150 and 1/150. Bayes factors above 150 and below 1/150 correspond respectively to very strong evidence in favour or against the correct model. 

Fig.~\ref{fig:bayesfact} and ~\ref{fig:bayesfact2}  show that distinguishing the resonant and eccentric models is possible, especially if the period of the inner planet varies, which is the case in Kepler data. \REWRITE{Nonetheless, at low SNR, there are cases of decisive Bayes factor against the correct model. In the low SNR regime, the evidence is dominated by the prior, and the parameter space is hard to explore. Both effects can account for these spurious results. }

We point out that the values reported in Fig.~\ref{fig:bayesfact} and ~\ref{fig:bayesfact2} are very sensitive to the prior on $K$. The narrower it is, the ``cheaper'' it is to add a planet, such that the two planets model is be favoured. The same experiment as above is done with a flat prior on $K$ on $[0, 40]$ m/s. In that case, the two planet model is systematically favoured, \REWRITE{as shown on Fig.~\ref{fig:bayesfact3}. }
	
	\REWRITE{
	As a conclusion, it seems good practice to check the influence of the prior on $K$. Secondly, since we expect that as the period ratio of resonant planets not to be exactly one half, if the two planet hypothesis is true then it will be strongly favoured by the Bayes factor. It therefore seems reasonable to consider the eccentric planet as the null hypothesis and not to reject it if there is no strong evidence for the two planet model.}

\begin{table}
	\caption{Priors used for the numerical experiments.}
	\begin{tabular}{p{1.1cm}|p{6cm}}
		Parameter & prior \\ \hline
		$K$ & Uniform on $[0,10000]$\\
		$P$ &  $1/P$ uniform on $[0, 20] $ \\
		$e$ & Beta(a=0.867,b=3.03) as in~\cite{kipping2014}\\
		$\omega$ & uniform on $[0,2\pi]$\\
		 $M_0$ &  uniform on $[0,2\pi]$ \\
	\end{tabular}
	\label{tab:priorbayes}
\end{table}
\begin{table}
	\caption{Priors used for the Bayes factor of eccentric model vs 2:1 mean motion resonance. The symbol $T_{\mathrm{obs}}$ denotes the observation time, $P_0 = 60$ days and $\alpha =0.1$.}
	\begin{tabular}{p{1.1cm}|p{6cm}}
 Parameter & prior \\ \hline
 $\sigma_J^2 $ & Uniform on $[0, 100]$ m/s \\ 
 offset & Uniform on $[-100, 100]$ \\ \hline
 $K,K_1$ & Uniform in $\ln K$ on $[-1,9]$\\
 $P,P_1$ &  $1/P$ uniform on $[1/P_0- 1/T_{\mathrm{obs}},1/P_0+ 1/T_{\mathrm{obs}}] $ \\
 $e$ & Uniform on $[0,1]$\\
 $\omega$ & uniform on $[0,2\pi]$\\
 $M_0$ &  uniform on $[0,2\pi]$ \\ \hline 
 $K_2$ & Uniform in $\ln K$ on $[-1,9]$\\
 $P_2$ & $1/P$ uniform on $[(1-\alpha)\frac{2}{P_0}, (1+\alpha)\frac{2}{P_0}    ]$ \\
 $\phi_1$ & uniform on $[0,2\pi]$\\
 $\phi_2$ & uniform on $[0,2\pi]$\\
	\end{tabular}
\label{tab:priorbayesf}
\end{table}

\begin{figure*}
\begin{minipage}[c]{0.49\textwidth}
	\includegraphics[width=8.2cm]{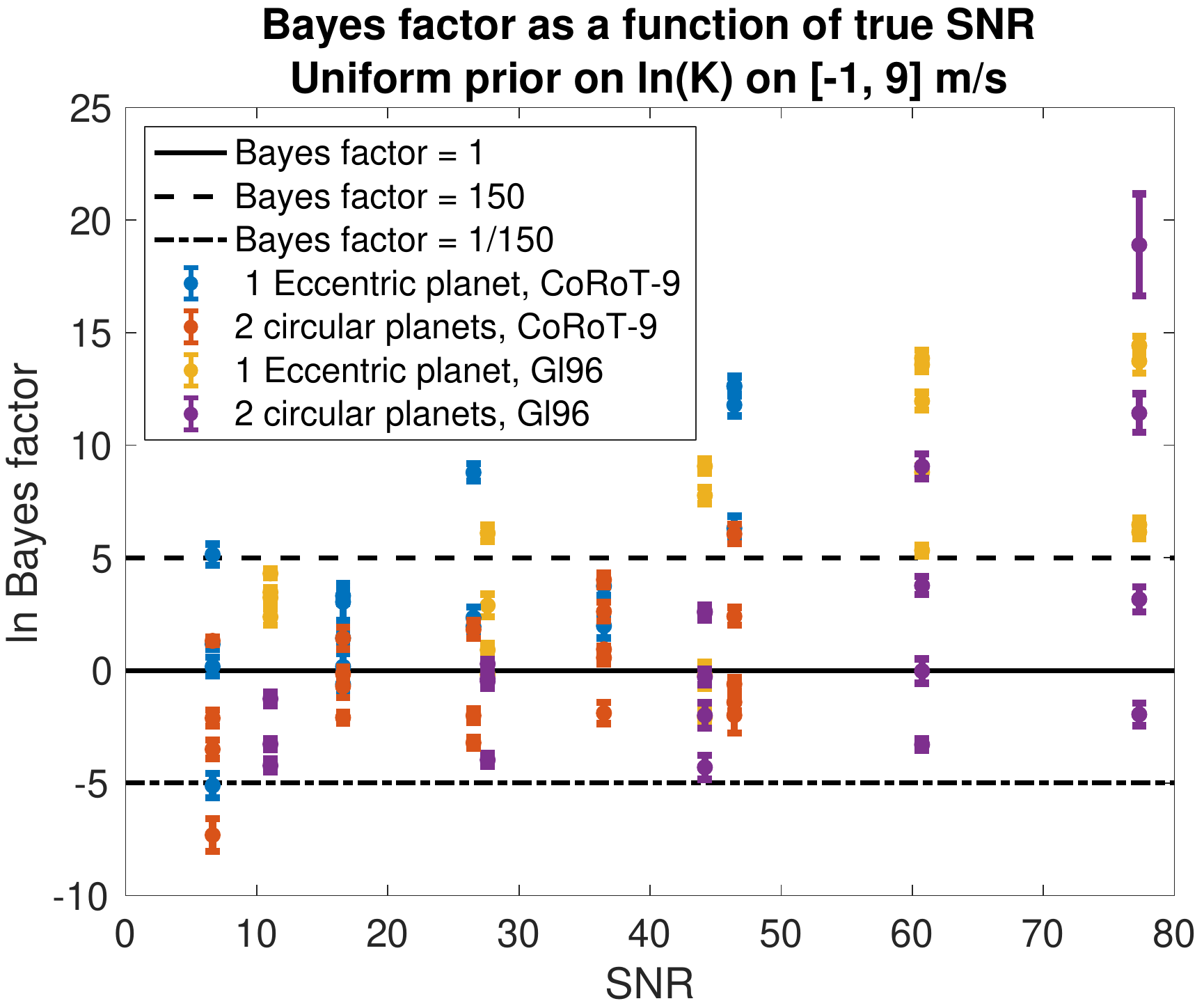}
	\caption{Difference of the log evidences of the correct and incorrect models. In the two planet case, the inner planet is generated with period exactly half of the outer planet. The color code corresponds to the description of the true data set (1 eccentric or 2 circular planets) and the array of measurement times used (CoRoT-9 or Gl96). }
	\label{fig:bayesfact}
\end{minipage}\hspace{0.2cm}
\begin{minipage}[c]{0.49\textwidth}
	\includegraphics[width=8.2cm]{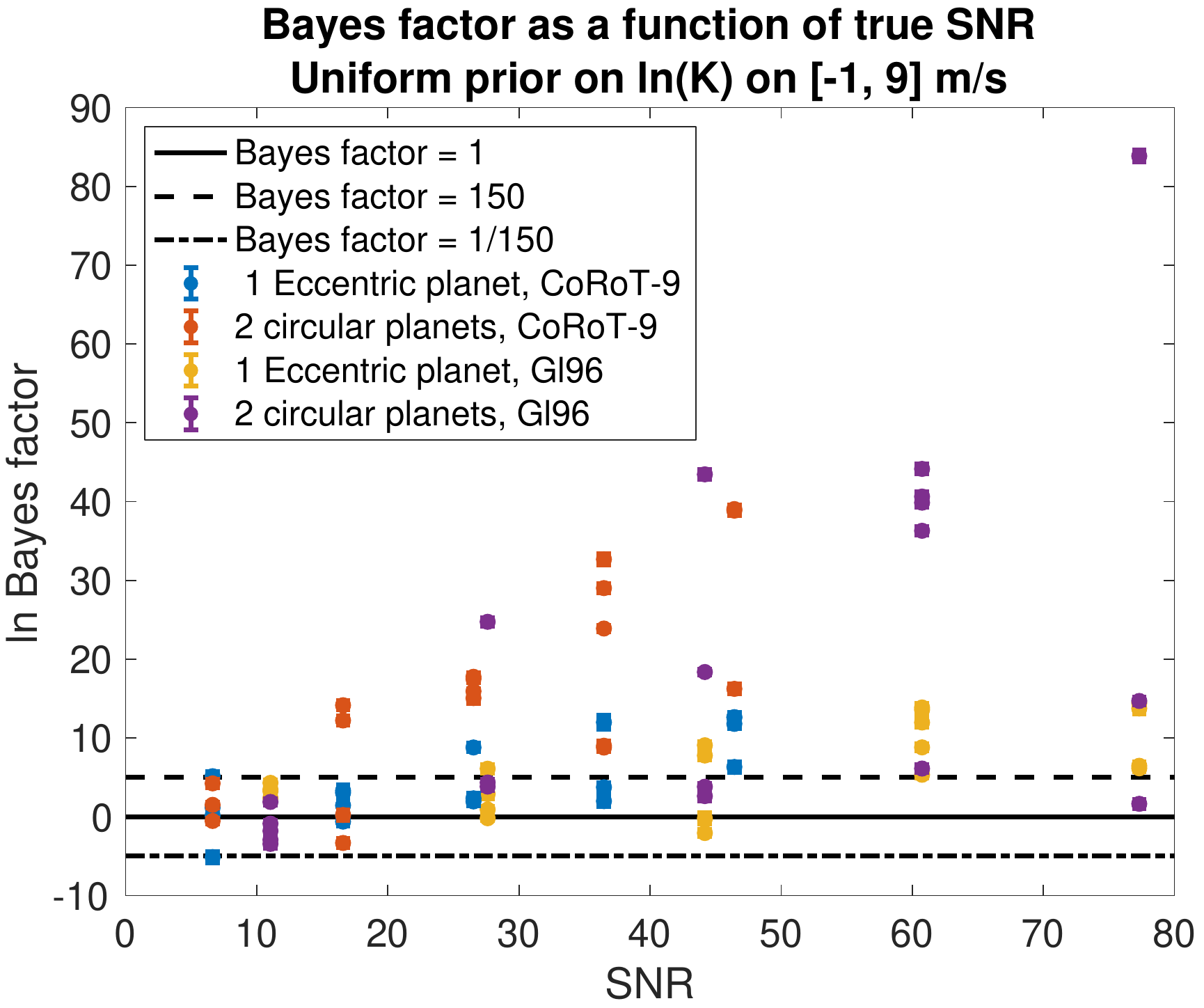}
	\caption{Difference of the log evidences of the correct and incorrect models, when the inner planet is generated with a frequency uniformly drawn on the same interval as the prior (table~\ref{tab:priorbayesf})). }
	\label{fig:bayesfact2}
\end{minipage}\hfill
\end{figure*}
\begin{figure}
	\includegraphics[width=8.2cm]{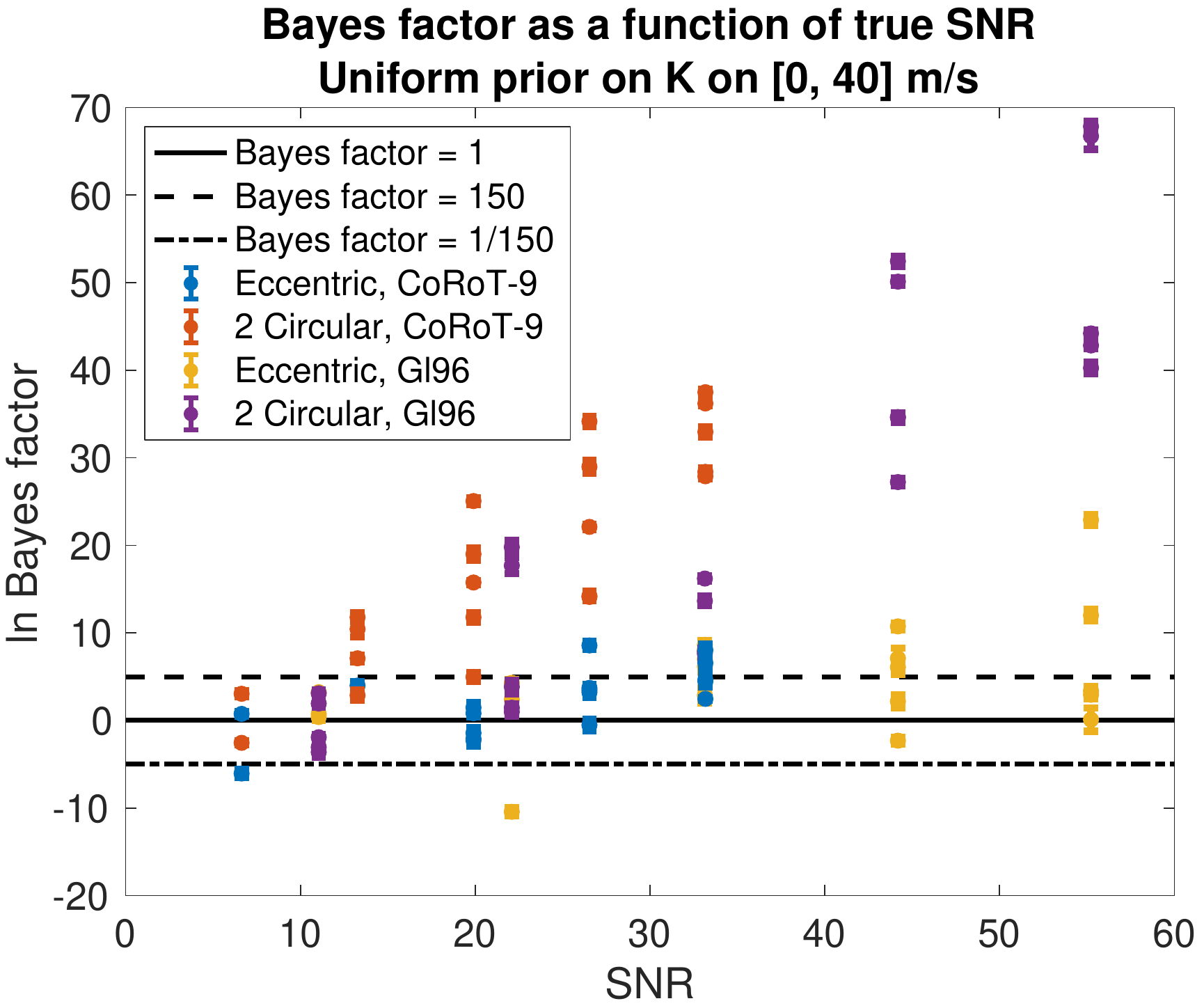}
	\caption{Difference of the log evidences of the correct and incorrect models, when the inner planet is generated with a frequency uniformly drawn on the same interval as the prior (table~\ref{tab:priorbayesf})) except for the semi amplitude, where the prior is uniform on [0, 40] m.\textsuperscript{-1}. }
	\label{fig:bayesfact3}
\end{figure}

\section{Detecting model errors: residual analysis}
\label{sec:resanalysis}
\subsection{Objective}

 In the previous section, we have seen that adjusting a jitter term is satisfactory in most cases, however we did not envision all possible errors. It is good practice to check if the models considered are plausible at all. One approach to take into account model uncertainty is to consider many models~\citep{jones2017} and rank them via cross validation, Akaike Information Criterion (AIC) or Bayesian Information Criterion (BIC), or even Bayes factor if possible. Alternately, we can test the hypothesis that the residuals are consistent with the model in an absolute sense. This problem is sometimes referred to as the goodness-of-fit problem, and is in general difficult~\citep[see][chap. 14]{lehmannromano2005}.

We reason as follows: if the set of models is appropriate to describe the data, then the residuals of the best fit must verify certain properties. If they do not, then we reject the hypothesis that there is one acceptable model that explains entirely the data set, among the set of models considered. Furthermore, we would like to obtain hints on the origin of a model misspecification. We expect outliers to change the  distribution of the residuals, and astrophysical or instrumental noise to introduce correlations.
 As a consequence, we consider two types of checks: is the distribution of the residuals approximately Gaussian? and: is there remaining time correlations in the residuals? 

\subsection{Distribution of the residuals}

As in section~\ref{sec:moremodels}, we first examine the linear case and show that the results are still helpful in the non linear setting. Let us suppose that we have a linear model $\mathbfit{y} = \mathbfss{A} \btheta + \bepsilon$ where $\mathbfss{A}$ is a $N\times p$ matrix and $\bepsilon$ is a Gaussian noise of covariance matrix $\mathbfss{V} =: \mathbfss{W}^{-1}$.  Let us denote by $\widehat{\mathbfit{y}}$ the least square fit model, and suppose the model ($\mathbfss{A},\mathbfss{V}$) is known. Then the weighted residual 
\begin{align}
\label{eq:rw}
\mathbfit{r}_{W} := \mathbfss{W}^{1/2} (\mathbfit{y} - \widehat{\mathbfit{y}})
\end{align}
is a vector of $N$ random variables that are approximately independent, Gaussian of null mean and variance one. To obtain a weighted residual that is a vector of independent Gaussian variables, let us define $\mathbfss{Q}$, the matrix such that $\mathbfss{J} = \mathbfss{Q}^T (\mathbfss{I}_N - \mathbfss{W}^{1/2}\mathbfss{A}^T(\mathbfss{A}^T \mathbfss{W}\mathbfss{A})^{-1}\mathbfss{A}^T \mathbfss{W}^{1/2}) \mathbfss{Q}$ is diagonal (it exists). Then the re-weighted residual $\mathbfit{r}_{QW}' = \mathbfss{Q}^T\mathbfss{W}^{1/2} (\mathbfit{y} - \widehat{\mathbfit{y}})$ has $p$ null components. The $N-p$ others are Gaussian variables of mean 0 and variance 1. In what follows, we denote by $\mathbfit{r}_{QW}$ the vector made of the $N-p$ components of $\mathbfit{r}_{QW}'$.  These two results are proven in Appendix~\ref{appendixresiduals}.

In practice, $\mathbfss{A}$ and $\mathbfss{V}$ are unknown, and we choose models $\mathbfss{A}'$ and $\mathbfss{V}'$. The two above properties can be used to test if $(\mathbfss{A},\mathbfss{V}) = (\mathbfss{A}',\mathbfss{V}')$ because if so, then the weighted residuals $\mathbfit{r}_W$ and $\mathbfit{r}_{QW}$ have a known distribution. 

We  compute an experimental cumulative distribution function (CDF) of $\mathbfit{r}_W$ and $\mathbfit{r}_{QW}$. If our model is correct, then it should be close to the CDF of a Gaussian variable of mean zero and variance one. 

\subsection{Correlations in the residuals}
\label{sec:corrres}

\REWRITE{The test suggested in the previous section is relevant to check the distribution of the residuals without temporal information, and thus is not most adapted to spot correlations.
We here adapt the variogram~\citep{matheron1963} for unevenly sampled time series, \REWRITER{similarly to ~\cite{baluev2013_gj581}}}. 
The quantity $d(t_i,t_j) = \mathbfit{r}_W(t_i) - \mathbfit{r}_W(t_j)$ is plotted as a function of $t_i-t_j$ for $t_i > t_j$. If $\mathbfit{r}_W$ is indeed independent and Gaussian, $d(t_i,t_j) $ should not depend on the time interval. 

\REWRITE{Secondly, we consider $n$ time bins with constant spacing in $\log t$. For each bin $B$, we compute the sample variance of the $d(t_i,t_j)$ such that $t_i - t_j \in B$. We expect that if there are correlations, these variances should grow as $t_i-t_j$ increases.
	To obtain an error bar on the variances, we add an independent Gaussian noise of mean 0 and variance 1 to $r_W$ and re-compute the variances for the same time bins. One can alternately add a Gaussian noise of covariance $\mathbf{V}$ to the data, and re-compute the residuals. The error bars are taken as $\pm \sigma$ where $\sigma$ is the standard deviation of the variances estimates per bin.   }

\subsection{Example}

Let us now show how it can be used in practice. We take the 214 measurement times of Proxima b~\cite{angladaescude2016}. $\mathbfss{A}$ is made of six columns as defined in Appendix~\ref{appendix_realformula} and fix $\mathbfit{x}_t$. We then generate three series of a thousand realisation of $\mathbfit{y} = \mathbfss{A}\mathbfit{x}_t + \bepsilon$. The covariance matrix of the noise has a kernel $\e^{-|\Delta t|/\tau}$ where $\Delta t$ is the duration between two samples. The three series are generated with a noise time-scale $\tau = 0$, 10 and 100 days. For each of the $3 \times 1000$ signals generated, we compute the least square fit with the correct matrix $\mathbfss{A}$, but with a weight matrix $\mathbfss{W}$ equal to identity, so our model is entirely correct only in the first case. First, we pick randomly one realization among the 1000 available in each series, and perform the first test whose result is plotted Fig.~\ref{testres1}. One clearly sees a pattern: the higher the correlation, the smaller is the difference between residuals. \REWRITE{We then consider 6 time bins, and compute the variances of the data and their uncertainties within each bin.  The results of these calculations are represented by the stair curves. For correlated noises, the variance increases with the time interval, while it stays compatible with a constant for the white noise. }
Fig.~\ref{testres2} shows the 1000 empirical CDFs in the three cases. 

The plots~\ref{testres1} and~\ref{testres2} are useful indicators of remnant correlations in the residuals and non Gaussianity. However, they do not constitute metrics with known statistical properties. One can potentially test the hypothesis that $\mathbfss{r}_{QW}$ is a realization of such a law with a Kolmogorov-Smirnov test or other metrics such as~\cite{andersondarling1954},~\cite{shapirowilk1965}, etc. We have tested the Anderson-Darling metric which did not show a high statistical power (see Hara 2017, thesis) and that is not discussed further.
The visual inspection, though less quantifiable, seems more accurate. 

Study of correlations in RV residuals have already been undertaken for instance by~\cite{baluev2011,baluev2013} with a smoothed residual periodogram~\citep{baluev2009}. The systematic comparison of these statistics is left for future work.

\begin{figure}
\centering
\begin{minipage}[c]{0.45\textwidth}
\includegraphics[width=8cm]{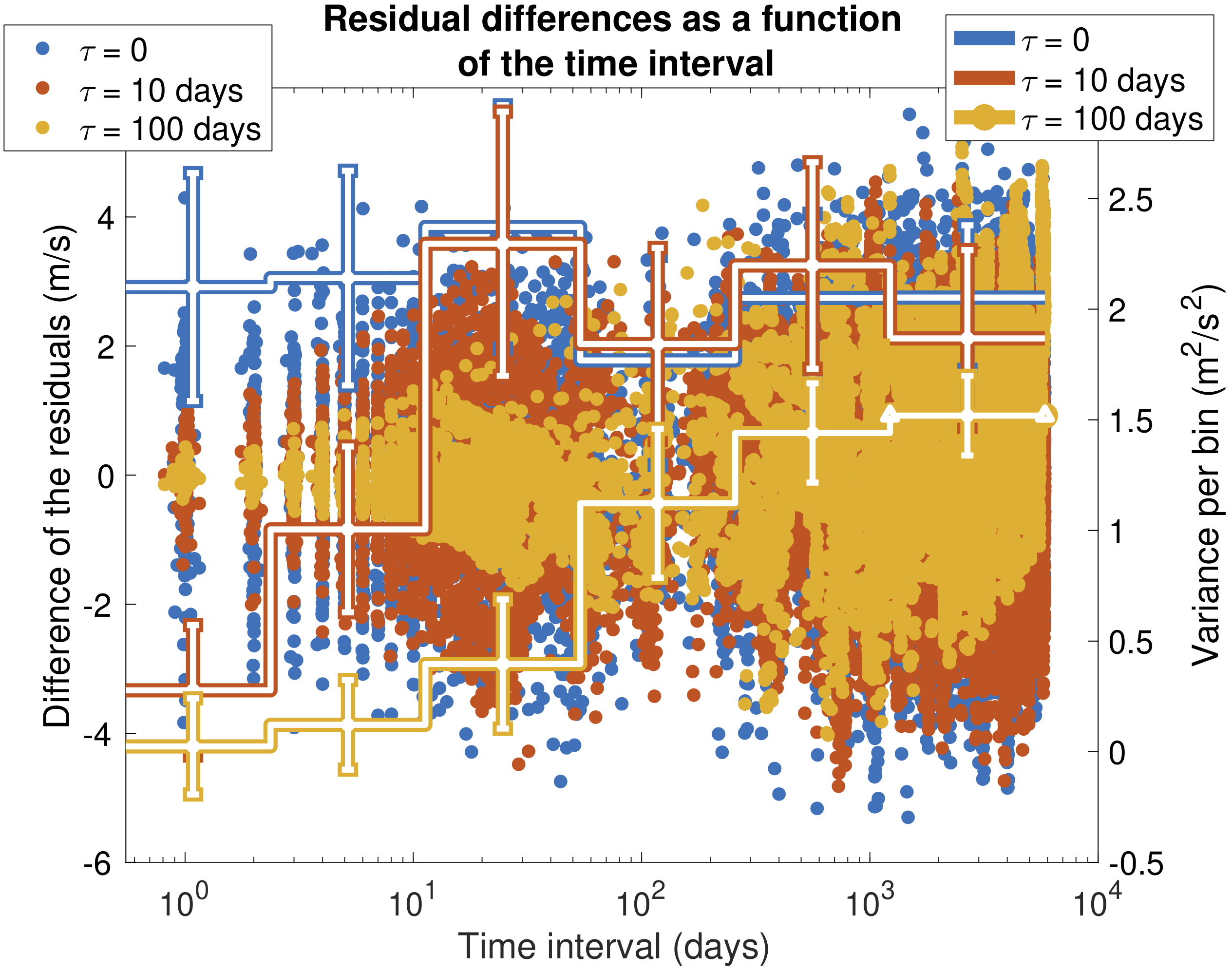}
\caption{Dots: difference between the residuals at two different time as a function of the time interval between them in three cases: when the noise has a time-scale of 0, 10 and 100 days, units on the left $y$-axis. \REWRITE{Stair curves: variance of the data points per time bin, with uncertainty, units on the right $y$-axis.}}
\label{testres1}
\end{minipage}
\hspace{0.5cm}
\begin{minipage}[l]{0.45\textwidth}
\includegraphics[width=7.5cm]{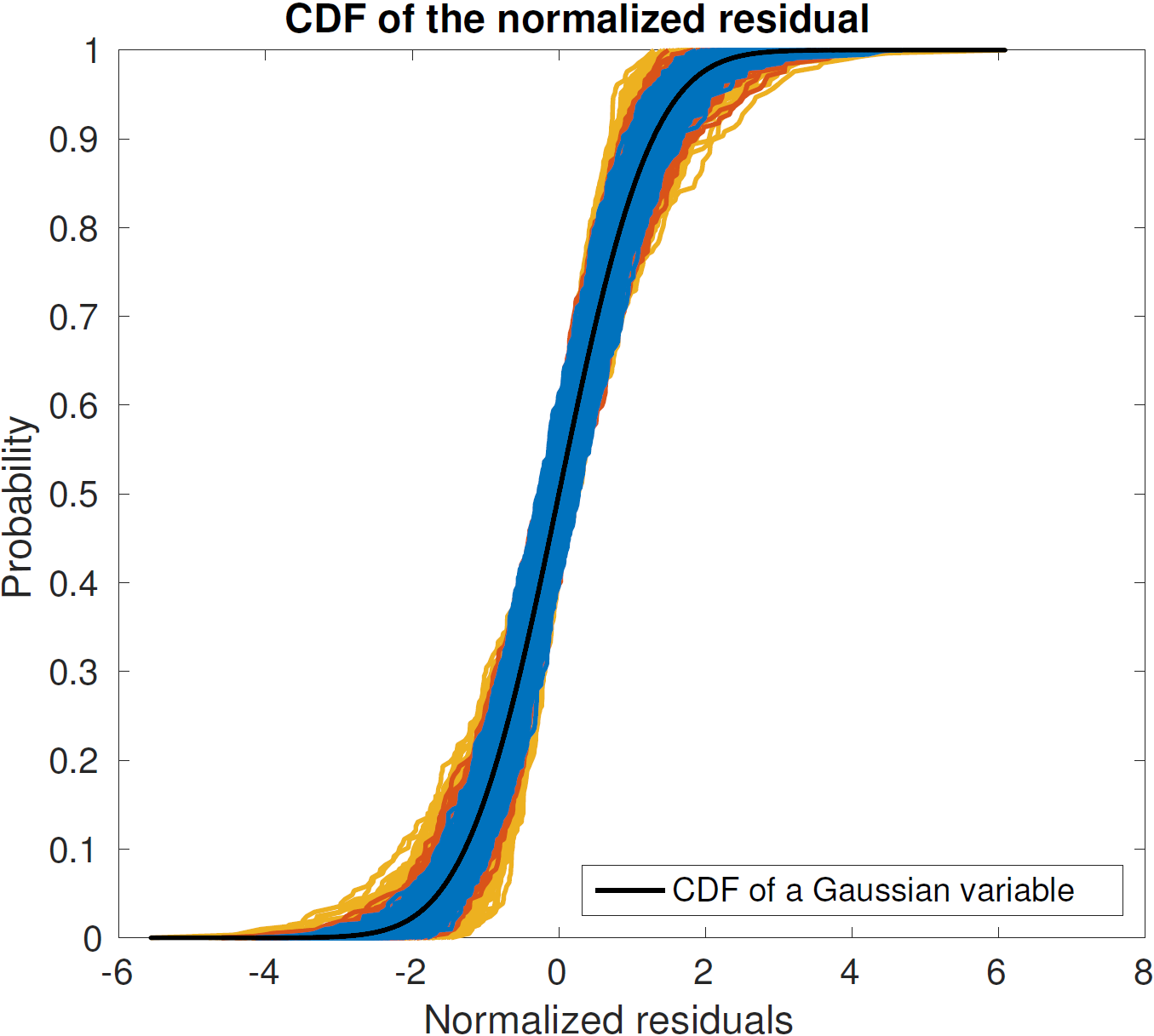}
\caption{A thousand realizations of the cumulative distribution functions of the normalized residual in three cases: when the noise has a time-scale of 0, 10 and 100 days. }
\label{testres2}
\end{minipage} \hfill
\label{testres}
\end{figure}

\section{Application: 55 Cancri}
\label{sec:application}

\begin{figure}
	\vspace{0.33cm}
	\includegraphics[width=7.8cm]{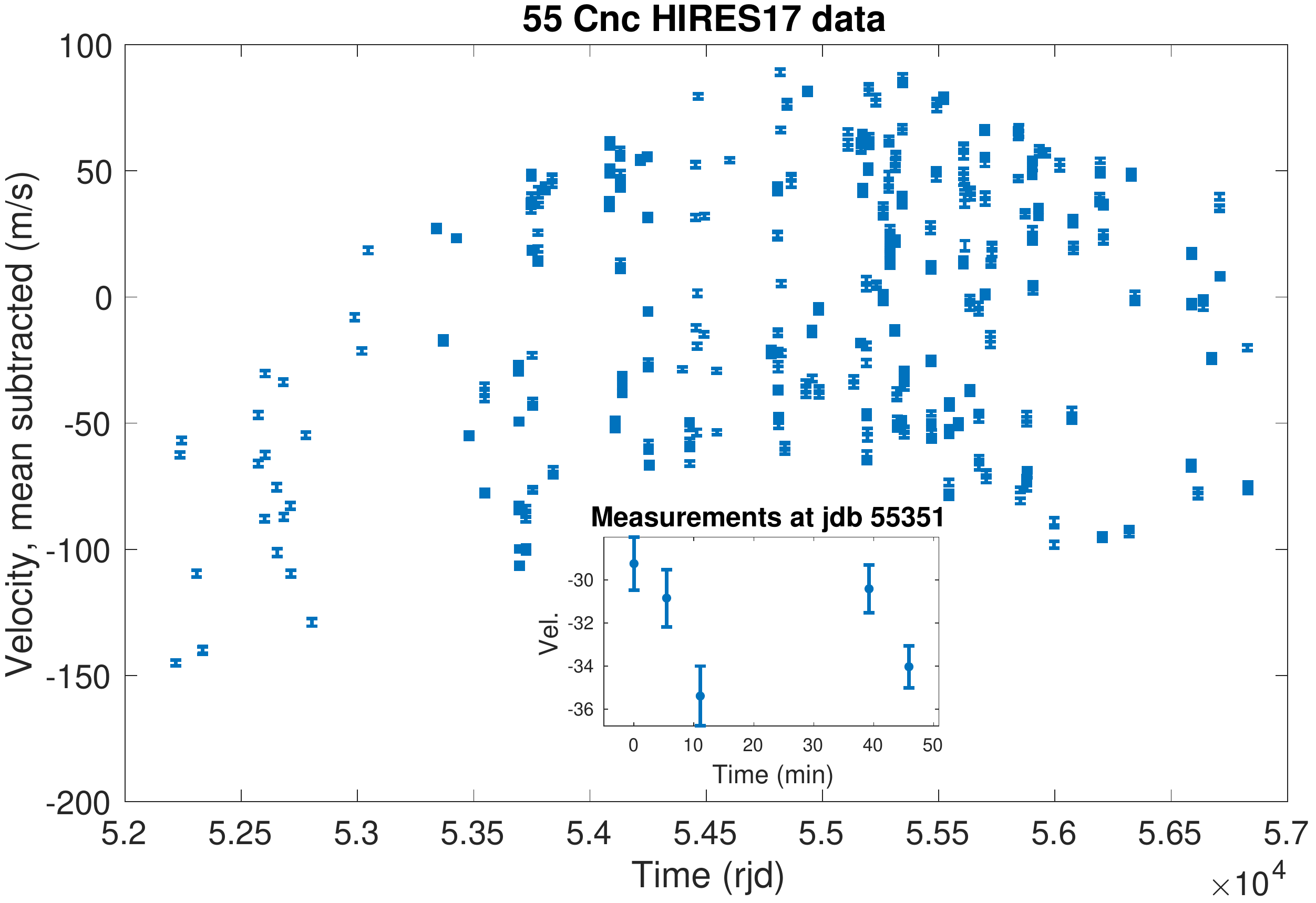}
	\caption{55 Cnc HIRES17 data with nominal error bars. The box on the bottom is a zoom on measurements taken at rjd 55351. }
	\label{fig:hires17data}
\end{figure} 
\begin{figure}
	\includegraphics[width=8.6cm]{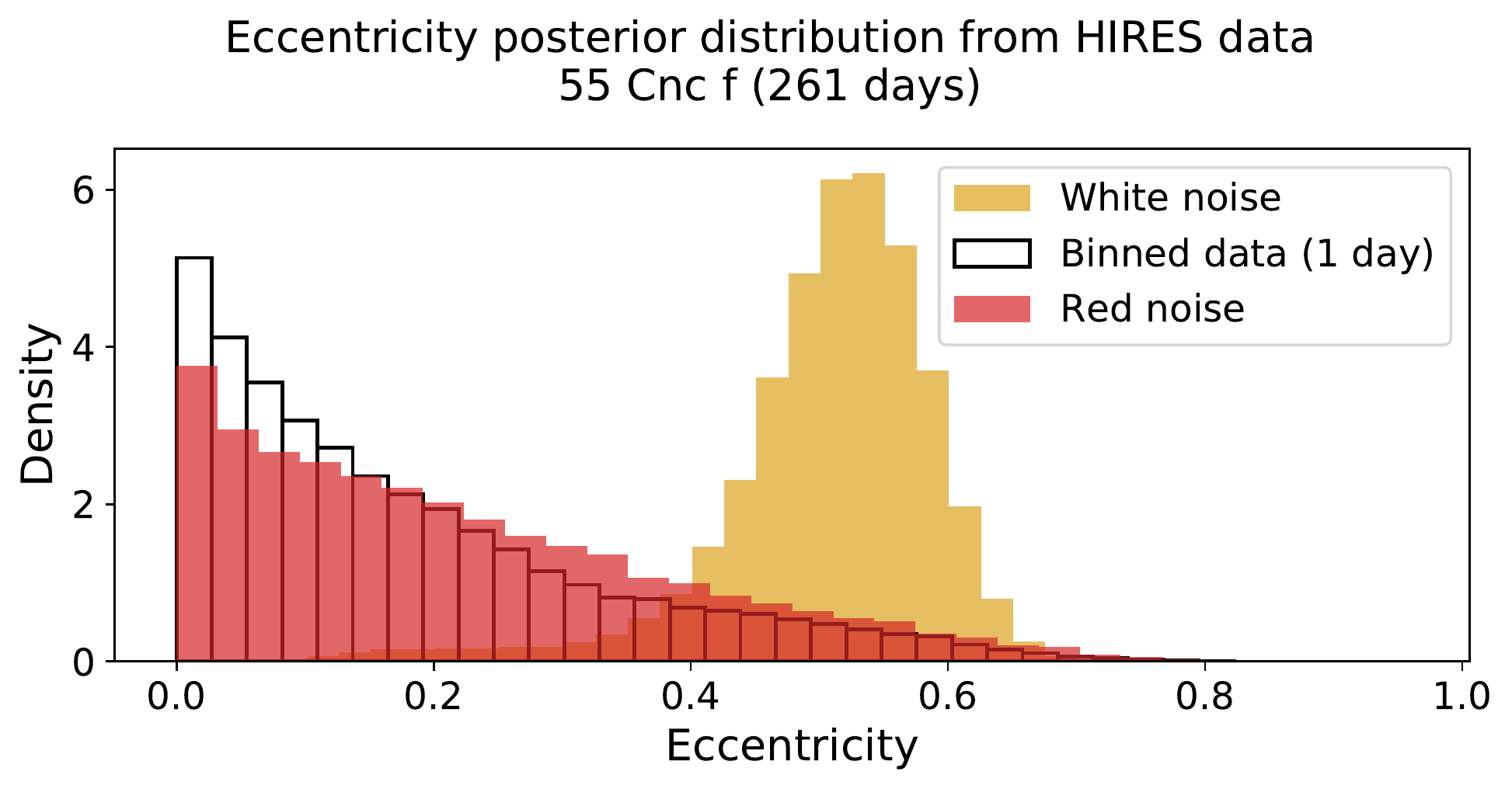}
	\caption{Posterior distributions and confidence intervals of the eccentricity of 55 Cnc f for different models of noise, raw and binned data, HIRES17 data. The yellow, red and white histograms represent the posterior distributions for a white noise and red + white noise model for the raw data, and white noise model on the binned data. }
	\label{fig:posterior55cnc}
\end{figure}
\begin{figure}
	\includegraphics[width=8.4cm]{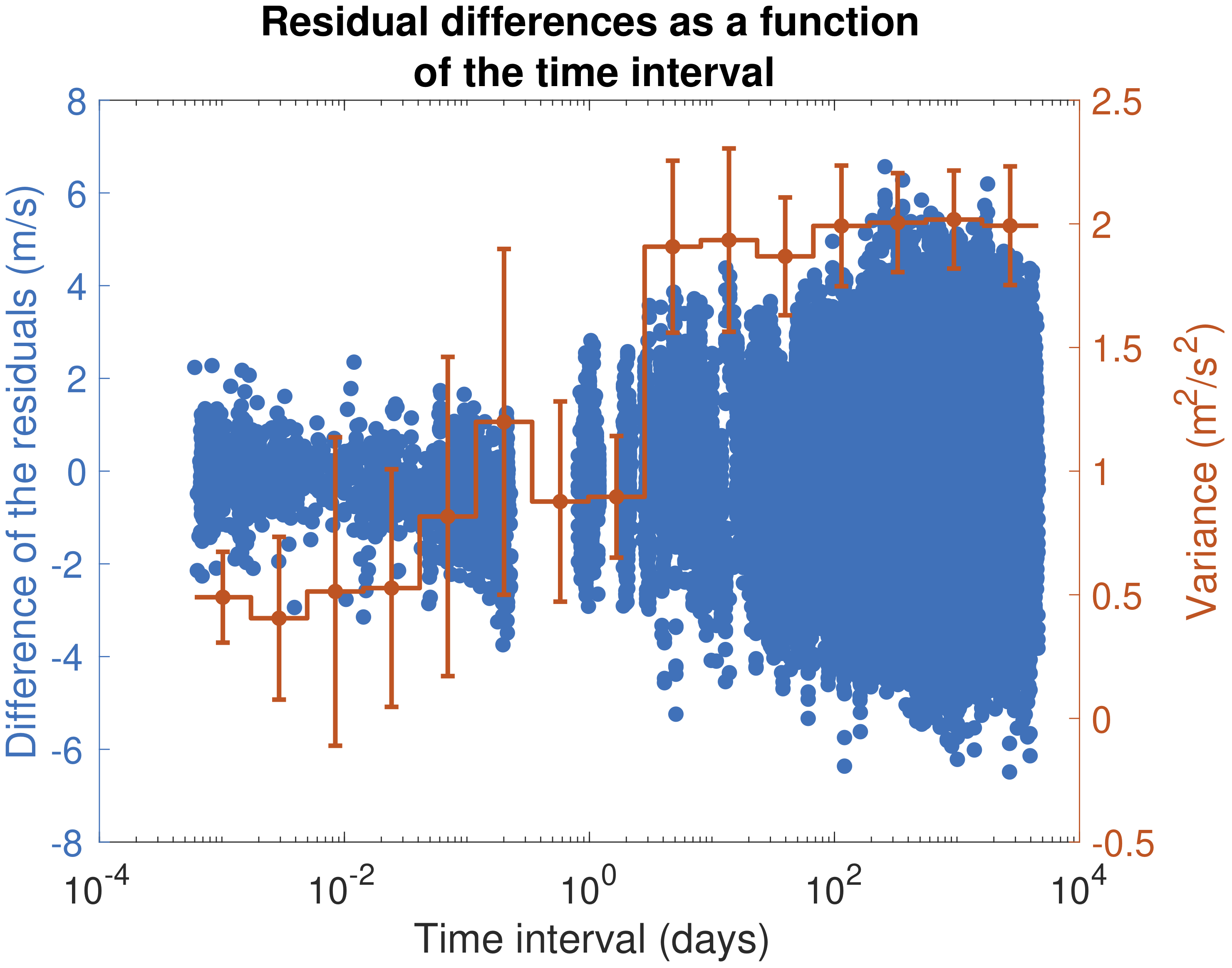}
	\caption{Difference between all couples of residuals with a 5 planets + offset and white noise fit on the HIRES17 raw data as a function of the time difference between residuals (see sec.~\ref{sec:corrres}). The red stair curves represents the variance of the residuals on a constant step in $\log$ time with an estimate of the error on this variance.  }
	\label{fig:res55cnc}
\end{figure}
\begin{figure}
	\includegraphics[width=8cm]{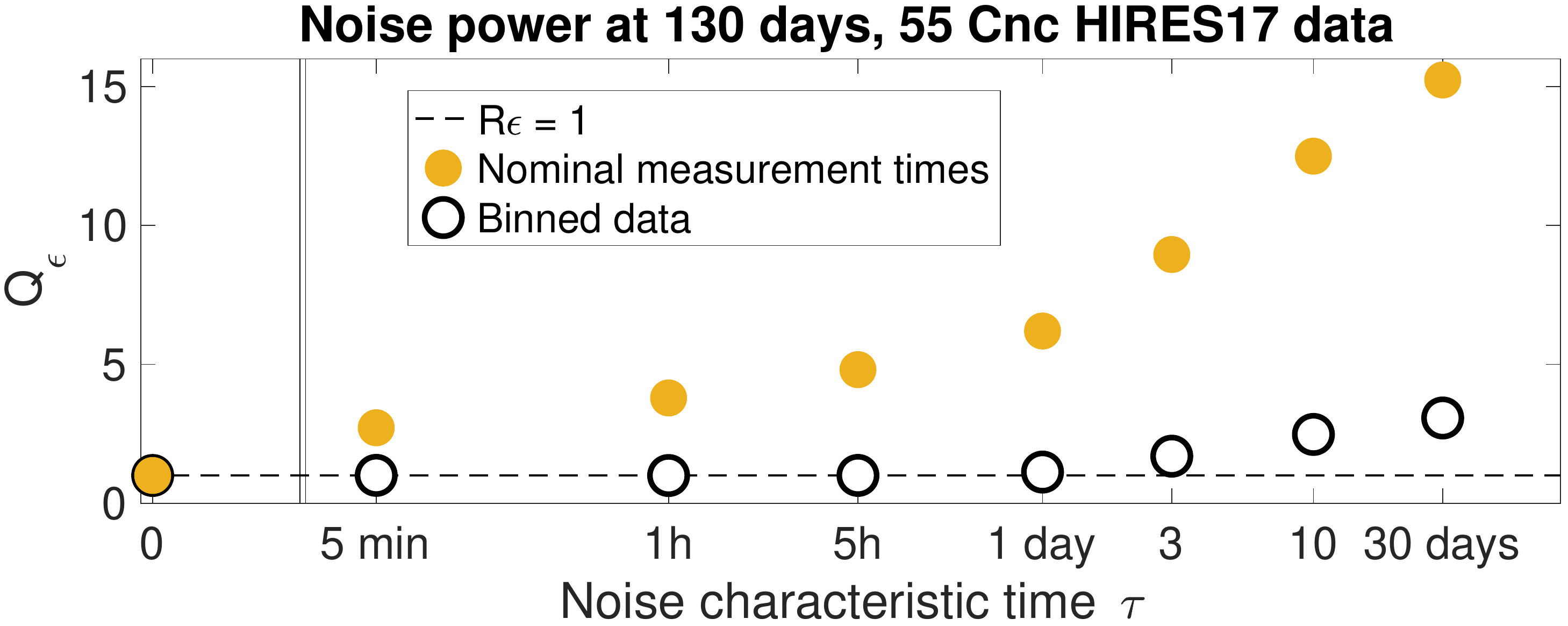}
	\caption{$Q_\bepsilon$ as defined in eq.~\eqref{eq:repsilon} at 130 days for the measurement times of HIRES17 data, as a function of the assumed noise characteristic time. The yellow dots correspond  the nominal measurements, and the white dots correspond to the data binned per day. }
	\label{fig:55cncrepsilon}
\end{figure}

\REWRITE{
To illustrate the methods above, the estimation of the eccentricities of the 55 Cnc system is discussed. This system has been extensively studied and many measurements, from several spectrographs, are publicly available: Hamilton~\citep{marcy2002, fischer2008}, ELODIE~\citep{naef2004}, HRS~\citep{mcarthur2004,endl2012}, HIRES~\citep{butler2017}.  We here focus on the HIRES~\citep{butler2017} data set, from now on denoted by HIRES17, which illustrates the importance of the noise model choice for reliable eccentricities. 
}

The HIRES17 data set is made of 607 velocity measurements, with on average 4 data points per night. Fig.~\ref{fig:hires17data} shows these data with nominal error bars. The inset is a zoom on five points taken within a time interval of 50 minutes at reduced Julian day (RJD) 55351. We first consider the raw data. The posterior distributions of the orbital elements is computed with a model including the five known planets plus a free jitter term as defined in section~\ref{sec:noiselevel}, and priors defined as in table~\ref{tab:priorbayes}. \REWRITER{We now focus on 55 Cnc f, orbiting at 260.9 days  with a minimum mass of $0.1503^{+0.0076}_{-0.0076}$ $M_J$ and a period of $259.88^{+0.29}_{-0.29}$ days~\citep{bourrier2018}}. In Fig.~\ref{fig:posterior55cnc}, the posterior distribution of the eccentricity is represented in yellow. The estimate is 0.5 with a 68\% credible interval equal to $[0.44, 0.58]$. 
This estimate is very different from the one in~\cite{bourrier2018}, \REWRITER{$e = 0.08^{+0.05}_{-0.0.04}$}, where orbital elements estimates are based on all the available radial velocity data.

\REWRITE{
To check for correlations in the residuals, these are studied with the method of section~\ref{sec:corrres}. We define $\mathbf{W} = (\mathbf{V} + \widehat{\sigma}_{ML}^2)^{-1}$, where $\mathbf{V}$ is the nominal covariance matrix, and $\widehat{\sigma}_{ML}^2$ is the maximum likelihood estimate of the jitter. The weighted residuals $r_W$ as defined in eq.~\eqref{eq:rw} are computed. For all combinations of measurement times $t_i>t_j$ we represent $d_{ij} := r_W(t_j) - r_W(t_i)$ as a function of $t_j - t_i$. We then compute the variance of the $d_{ij}$ such that $t_j - t_i$ is in a certain time bin. Fifteen such intervals are considered, with a constant length in $\log$ scale. The uncertainty on this variance $\sigma^2$ is estimated by bootstrap. The results are shown  in Fig.~\ref{fig:res55cnc}, where the blue points represent a couple $(t_j - t_i, d_{ij})$, and the red stair curve represents the variance per bin and its uncertainty. As the time difference grows, the residuals are more and more dispersed, which is indicative of correlated noise . 
}

\REWRITE{
To account for this correlated noise, we adopt two strategies. First,  we include a red noise term in the model of section~\ref{sec:noiselevel}. This one is assumed Gaussian with a covariance kernel 
\begin{align}
k(t,t') = \sigma_R^2 \e^{-  \frac{(t-t')^2 }{2 \tau^2} }, 
\label{eq:covarkern}
\end{align}
and $\sigma_R$ and $\tau$ are included in the posterior sampling \REWRITER{with a normal prior on $\ln \tau$ of mean and variance 1}. The resulting eccentricity posterior for 55 Cnc f is shown in red on Fig.~\ref{fig:posterior55cnc}. \REWRITER{The posterior median of $\tau$ is 2.3 days, with a 68\% confidence interval of $[0.5, 9.6]$ days.} The second strategy we adopt is to bin the velocity measurements per day. Such a strategy is actually what likely lead to the measurement pattern chosen, as binning data per night averages out stellar oscillations~\citep{dumusque2011i}. The posterior so obtained is represented by the white histogram on Fig.~\ref{fig:posterior55cnc}. In both cases it appears that that the eccentricity is in fact much less constrained. Note that the posterior distribution favours a null eccentricity.  We attribute this to the choice of the beta prior, which favours small eccentricities, but also to the fact described in section~\ref{sec:bayesianestimates}: high eccentricity models represent a larger volume of models, and are therefore be penalized by Bayesian estimates.
}

\REWRITE{
In section~\ref{sec:correlations}, we argued that the main characteristic of a noise that influences eccentricity estimation is, for a particular time sampling,  its power at the semi period of the planet. The results above are indeed interpretable in this term. The quantity $Q_\bepsilon(\mathbf{M_2})$ defined in eq.~\eqref{eq:repsilon} is computed for noises with a kernel defined in eq.~\eqref{eq:covarkern}. The calculation is made with the nominal time sampling and the one corresponding to a one day binning. Fig.~\ref{fig:55cncrepsilon} shows the result of these calculations. With the nominal times, the power at 130 days is much more sensitive to correlated noise, even at time-scales as short as 5 minutes,  which results in a clear eccentricity overestimation. This example illustrates that due to the sampling, some noises might create spurious eccentricities if not properly accounted for. }

\REWRITE{
As a remark, some measurement dates correspond to the observation of the 55 Cnc e transit. As noted in~\cite{bourrier2018}, the effect of the Rossiter-McLaughlin effect has been put in question in~\cite{lopezmorales2014}, and should not exceed 0.5 m.s\textsuperscript{-1}, which is well below HIRES precision and does not impact the analysis above. %
}

\section{Conclusion}
\label{sec:conclusion}
\subsection{Summary and step-by-step method}
\label{sec:discussion}

In section~\ref{sec:originecc}, we have seen that at low eccentricity, the bias of the least square estimate can be approximated analytically (eq.~\eqref{eq:bias1}). This equation as well as numerical simulations show that the bias is approximately proportional to the uncertainty on eccentricity. We show that an incorrect estimate of the noise level can create spurious global minima at high eccentricity through the Proxima b example, and that at high SNR it is less likely, though not impossible, to miss the global minimum by doing a least square fit initialized at $e=0$.
In section~\ref{sec:bayesianestimates}, it appeared that the maximum of the marginal distribution of eccentricity (eq.~\eqref{eq:maxmarginal}) is less biased than the maximum likelihood.

In section~\ref{sec:robustness}, we explored the reliability of the maximum likelihood and posterior distributions when the noise model is incorrect, along with the sensitivity of the inference to numerical methods. 
In summary, we recommend the following data analysis method for robust inference of eccentricities.
\begin{itemize}
\item Computing the posterior distribution of eccentricity leads to correct inference in general (see section~\ref{sec:bayessection}), if one includes in the noise model at least a free jitter term (eq.~\eqref{eq:modeljitter}). Without such a term the inference is very likely to be spurious. \REWRITE{The mean and median of the marginal posterior distribution of eccentricity}, defined in section~\ref{sec:bayesianestimates}, constitute good point estimates (see section~\ref{sec:incorrect_sim}). 
 Credible intervals (eq.~\eqref{eq:credible}) allow reliable hypothesis testing, confidence intervals~(see section~\ref{sec:intervalest}) are an alternative for probing the low eccentricity region especially. 
\item \REWRITE{Using a white, Gaussian noise model might however lead to spurious inferences in some cases. Eccentricity estimates are weakly sensitive to non Gaussianity of the noise, but are impacted by its true covariance, especially by the noise power at \REWRITER{$P/2$ and to a lesser extent at $P/3$ where $P$ is the period of the planet of interest} (see section~\ref{sec:correlations}). Stellar signals might have a modulation at the rotation period, but as shown in section~\ref{sec:application}, correlated noise combined with the sampling can also create spurious eccentricities. Inappropriate models leave signatures in the residuals, these ones can be studied with techniques presented in section~\ref{sec:resanalysis}.}
\item It is key to check the convergence of the MCMC used for posterior calculation, for instance with the effective number of independent samples (we recommend at least 10,000). An unreliable numerical method increases the bias (see section~\ref{sec:numericalerror}).
\item \REWRITE{At low SNR, posteriors are dominated by the prior and the parameter space is harder to explore. This might lead to spurious inferences (e.g. in section~\ref{sec:worstcase}). The influence of the prior can be assessed with the method of section~\ref{sec:recomputepriors}. }
\item The most degenerate case is an inner planet in 2:1 resonance. As discussed in section~\ref{sec:worstcase}, Bayes factors can disentangle those cases.  \REWRITE{However, the value of the Bayes factor strongly depends on the prior used for the semi amplitude. We suggest to select a $\log$ prior on semi amplitude and to consider the eccentric scenario as the null hypothesis. }
\item In multiplanetary systems, checking the system stability might also help ruling out some values of the eccentricity~\cite[e.g.][]{hebrard2016, delisle2018}.
\item As more measurements are obtained, least square and Bayesian estimates get closer, so that eq.~\eqref{eq:bias00} can also be used to approximate the number of measurement $N(\epsilon_e)$ needed to have an average bias on eccentricity $\epsilon_e$ for a planet with semi-amplitude $K$, $p$ parameters fitted in total and a measurement root mean square RMS,
\begin{align}
N(\epsilon_e) = p + \pi  \frac{\mathrm{RMS}^2}{K^2 \epsilon_e^2}.
\label{eq:nepsilone}
\end{align}
This formula is for $\epsilon_e\leqslant 0.05$ and subtends that there is a good phase coverage.
\end{itemize}
This procedure should be familiar to observers, since it is mainly a formalized and tested version of common practices. It is applicable to other purposes, especially the methods of section~\ref{sec:resanalysis}, since checking the model validity improves the inference robustness in general.

\subsection{Perspectives}

The main point of the present work is that modelling errors might have systematic impacts on the estimates of orbital elements, and thus on our understanding of planetary systems. It is yet to be determined in which extent this has been the case in past studies. 

In section~\ref{sec:resanalysis}, we presented tools to measure the absolute adequacy between a model and the data, which is complementary to comparing models to one another. There remain many such adequacy metrics to explore. These ones could prove useful in the context of exoplanets characterization but also exoplanets detection.

\section*{Acknowledgements}

N. Hara and J.-B. Delisle acknowledge the financial support of the National Centre for Competence in Research PlanetS of the Swiss National Science Foundation (SNSF). N. Hara thanks  J. J. Zanazzi for his interesting inputs. We thank the anonymous referee for his/her insightful suggestions.

\bibliography{biblio.bib} 
\bibliographystyle{mnras}

\appendix

\section{First order approximation}
\label{appendix_realformula}

In this section, we developp the Keplerian model to first order in eccentricity to obtain an analytical expression of the bias. Within this approximation, the distribution of the least square fit knowing that $e=0$ is given in \cite{lucy1971}. This section extends their formula to small $e$, and takes into account the number of fitted parameters. First, we develop~\eqref{eq:vexprbis} to order one in $e$, obtaining
\begin{align}
y(\lambda,K,P,e,\omega) = K(\cos(\lambda) + e \cos(2 \lambda - \omega))
\end{align}
where $\lambda = nt + \omega + M_0= \lambda_0 + 2\pi t/P$ is the mean longitude, $\lambda_0$ being its value at $t=0$. Denoting by $n = 2 \pi/P$ the mean motion, the above expression can be re-written
\begin{align}
\label{eq:v1}
y^{(1)}(t,A,B,C,D,n) = A \cos nt + B \sin nt + C \cos 2nt + D \sin 2nt 
\end{align} 
where $A = K \cos \lambda_0$, $B = -K \sin \lambda_0$, $C = K e \cos(2\lambda_0 - \omega)$, $D = -K e \sin(2\lambda_0 - \omega)$.
When other parameters are fitted, the uncertainties on $A,B,C,D$ increases as well. To quantify this effect, we consider the problem of fitting the period and a constant. 
\begin{align}
\begin{split}
&\mathbfit{y}^{(2)}(t,A,B,C,D,E,F) = \\
&A \cos nt + B \sin nt + C \cos 2nt + D \sin 2nt + E \frac{\partial y}{\partial n}(t) + F .
\label{eq:model2}
\end{split}
\end{align}
which in a matrix form gives
\begin{align}
\label{eq:model2b}
\mathbfit{y}^{(2)}(\mathbfit{t},A,B,C,D,E,F) = \mathbfss{M}(P) \mathbfit{x}.
\end{align}
Let us assume that the observations are $\mathbfit{y}(\vec t) = \mathbfss{M}(P) \vec x_t + \vec \epsilon$, where $\vec \epsilon$ is a Gaussian noise, independent and identically distributed with variance $\sigma^2$. The least square estimate of $\mathbfit{x}$ is $\widehat{\mathbfit{x}} = (\mathbfss{M}^T  \mathbfss{M})^{-1} \mathbfss{M}^T  \mathbfit{y}$, and the estimate of eccentricity is 
\begin{align}
\label{eq:abcd}
\widehat{e} = \sqrt{\frac{\widehat{C}^2 + \widehat{D}^2}{\widehat{A}^2 + \widehat{B}^2}} = \frac{\sqrt{\widehat{C}^2 + \widehat{D}^2}}{K_t} \left( \frac{\sqrt{\widehat{A}^2 + \widehat{B}^2}}{K_t}  \right)^{-1}  ,
\end{align}
where $K_t$ is the true semi-amplitude. By change of random variable we can obtain the law followed by $\widehat{e}$. If we assume that $N$ is large enough then the columns of $\mathbfss{M}(P)$ are approximately orthogonal,  the components of $\widehat{x}$ are independent Gaussian variables.
Since the modulus of a sum of independent Gaussian variables follows a Rice distribution,
\begin{align}
U \equiv \frac{\sqrt{\widehat{C}^2 + \widehat{D}^2}}{K_t} \sim g(u) =  S^2 u \e^{-\frac{S^2}{2} (u^2+e_t^2)} I_0(S^2 e u) \\
W \equiv \frac{\sqrt{\widehat{A}^2 + \widehat{B}^2}}{K_t} \sim h(w) =S^2 w \e^{-\frac{S^2}{2} (w^2+1)} I_0(S^2 w)
\end{align}
where $I_0$ is a modified Bessel function of first kind, $S = K_t/\sigma$ is the SNR, where $\sigma$ is the standard deviation of $\widehat{A},\widehat{B},\widehat{C}$ and $\widehat{D}$. If $K$ is sufficiently large, $W$ is close to 1 and $g(u)$ gives a good approximation of the law followed by the eccentricity fitted. Within this approximation, one can obtain analytical formula for the bias $b$ of the eccentricity that only depends on the true eccentricity and the SNR,
\begin{align}
b(e_t,S,n) = \frac{1}{S}\sqrt{\frac{\pi}{2}} L_{1/2}\left( \frac{S^2 e_t^2}{2}\right) - e_t .
\end{align}
where $L_{1/2}$ is the Laguerre polynomial of order $1/2$.
 In case $K_t$ is small, one must use the formula for the law followed by the quotient of two random variables:
\begin{equation}
\label{bias2}
\widehat{e} = \frac{U}{W} \sim f(e) = \int_{-\infty}^{+\infty} g(u) h(ue) |u| du
\end{equation}
but no simple analytical expression was found.

When fitting model~\eqref{eq:model2b} to $\mathbfit{y}(\vec t)$, the estimate $\widehat{\vec \theta}$ have a covariance matrix $\boldsymbol{\Sigma}^{-1}$ where $\boldsymbol{\Sigma} = \sigma^2 (\mathbfss{M}(P)^T \mathbfss{M}(P))$ (this is a classical statistical result, see for example~\cite{pelat}). The variances of the components of $\widehat{ \vec x}$ are given by the diagonal elements of $\boldsymbol{\Sigma}^{-1} $. Their approximate calculation is the object of the next section.

\subsection{Average error}

First we consider the estimation of the error on $A,B,C,D$ when averaging over the mean motion $n$. At little cost, we can generalize our claim to the fitting of model~\eqref{eq:model2} plus fitting other linearised Keplerian model. This approximately corresponds to fitting a multi-planetary system starting closely from the correct local minimum of $\chi^2$. Again, the model can be written as a linear one, $\mathbfit{y} = \mathbfss{M} \mathbfit{x}$ but where $\mathbfss{M}$ has $p = 6 + 5k$ columns, $k$ being the number of additional planets. To facilitate the discussion, we normalize the columns of $\mat M$. To have an expression of the model of the form~\eqref{eq:model2b}, we have multiplied the  $k$\textsuperscript{th} component of $\theta$ by the norm of the $k$\textsuperscript{th} column of $M$. The variances of these new model parameters are still given by the diagonal elements of $\sigma^2\boldsymbol{\Sigma}^{-1}$ where $\boldsymbol{\Sigma} = (\mathbfss{M}^T \mathbfss{M})$, but now $\boldsymbol{\Sigma}$ has only ones on its diagonal.

Calculating precisely the uncertainty on $A,B,C,D$ averaged over $n$ and the phase of the signal as a function of the instant of observations $t$ is complex since it requires the inversion of $\sigma$ which is a $6 + 5k \times 6 + 5k$ matrix. Instead, we use an approximation that grasps the effect we want to estimate: how the uncertainty worsens as more parameters are added to the model. We consider that the elements of $\mathbfss{M}$ are drawn from independent Gaussian laws that have a variance $1/N$. To avoid confusion with the true model, the so defined random matrix is denoted by $\tilde{\mathbfss{M}}$ and its covariance matrix by $\tilde{\boldsymbol{\Sigma}}$.

This approximation seems to be rough at first but  turns out to be surprisingly accurate as a lower bound in practice. A few arguments to justify that it is a reasonable guess are listed below.
\begin{itemize}
\item The variances of the entries were chosen such that the expectancy of a squared norm of a column is one, which is the value of $\boldsymbol{\Sigma}$ diagonal elements.
\item The columns are cosines and sines, which are approximately orthogonal, and in the Gaussian case decorrelation implies independence. Furthermore, the average of the spectral window is equal to the expected value of a correlation between two random Gaussian variables. 
\item The normed vectors $\cos \nu \vec t$ and $\sin \nu \vec t$ are approximately distributed uniformly on the sphere of $\mathbb{R}^N $ when $ \nu$ is distributed uniformly between 0 and $2\pi/T_{\mathrm{obs}}$. 
\end{itemize}

The expected value of the variance of any parameter is the expected value of any diagonal element of $\boldsymbol{\Sigma}^{-1}$, since all the columns of $\tilde{\mathbfss{M}}$ follow the same law. 
To tackle that problem, we rewrite $\tilde{\boldsymbol{\Sigma}}$ as
\begin{align*}
\tilde{\boldsymbol{\Sigma}} = \sigma^2 \left(  
\begin{array}{cc}
\boldsymbol{\Sigma}_{11} & \boldsymbol{\Sigma}_1^T \\
\boldsymbol{\Sigma}_1 & \boldsymbol{\Sigma}_c
\end{array}
\right)
\end{align*}
Where $\boldsymbol{\Sigma}_{11}$ is $\tilde{\boldsymbol{\Sigma}}$ element at first row and first column and $\boldsymbol{\Sigma}_1$ is a column vector with $N-1$ entries. We now have 
\begin{align*}
\mathbb{E}\{\tilde{\boldsymbol{\Sigma}}^{-1}_{11}\} = \frac{1}{\sigma^2} \mathbb{E}\left\{ \frac{1}{ \boldsymbol{\Sigma}_{11} - \boldsymbol{\Sigma}_1^T \boldsymbol{\Sigma}_c^{-1} \boldsymbol{\Sigma}_1}  \right\}
\end{align*} 
By Jensen inequality~\citep{jensen1906}, since $x \rightarrow 1/x$ is convex,
\begin{align*}
\mathbb{E}\{\tilde{\boldsymbol{\Sigma}}^{-1}_{11}\} \leqslant \frac{1}{\sigma^2} \frac{1}{ \mathbb{E}\left\{\boldsymbol{\Sigma}_{11} - \boldsymbol{\Sigma}_1^T \boldsymbol{\Sigma}_c^{-1} \boldsymbol{\Sigma}_1\right\}}  .
\end{align*} 
Now since for two independent variables $X$ and $Y$, $\mathbb{E}\{XY\} = \mathbb{E}\{X\}\mathbb{E}\{Y\} $,
\begin{align*}
\boldsymbol{\Sigma}_1^T \boldsymbol{\Sigma}_c^{-1} \boldsymbol{\Sigma}_1 = \sum_{k=2}^p \mathbb{E}\{\boldsymbol{\Sigma}_{1k}^2\} \mathbb{E}\{\boldsymbol{\Sigma}_{c,kk}^{-1}\} \leqslant \sum_{k=2}^p \mathbb{E}\{\boldsymbol{\Sigma}_{1k}^2\} = 1 - \frac{p-1}{N}
\end{align*} 
As by construction $\mathbb{E}\left\{\boldsymbol{\Sigma}_{11}\right\}=1$, we finally obtain 
\begin{align*}
\mathbb{E}\{\tilde{\boldsymbol{\Sigma}}^{-1}_{11}\} \leqslant \frac{1}{\sigma^2} \frac{1}{1-\frac{p-1}{N}}
\end{align*} 
where the inequality follows again from Jensen's inequality applied to matrix inversion. 
Finally, the standard deviation on $I = A,B,C,D$ is
\begin{align}
\sigma_I \geqslant \sigma \sqrt{\frac{1}{1-\frac{p-1}{N}}}
\end{align}
With the approximation $ \|\cos \nu \mathbfit{t} \| \approx \|\sin \nu \mathbfit{t} \| \approx \sqrt{N/2}$, the errors on $k = C/\sqrt{A^2+B^2}$ and $k = D/\sqrt{A^2+B^2}$ then verify
\begin{align}
\sigma_k \gtrsim \frac{\sigma}{K_t} \sqrt{\frac{2}{N}} \sigma_I = \frac{\sigma}{K_t} \sqrt{\frac{2}{N-p+1}} \approx \frac{\sigma}{K_t} \sqrt{\frac{2}{N-p}} =: 1/S.
\label{eq:sigmak}
\end{align}
As $k$ and $h$ approximately follow a Gaussian law, $e = \sqrt{k^2+h^2}$ follows a Rice distribution, whose mean is given by 
\begin{align}
\label{eq:formule}
\mathbb{E}\{\widehat{e}\} & = \frac{1}{S}\sqrt{\frac{\pi}{2}} L_{1/2}\left( -\frac{S^2 e_t^2}{2}\right) . \\
\label{eq:bias00_app}
\mathbb{E}\{\widehat{e}|e_t=0\} & =  \frac{\sigma}{K_t}\sqrt{\frac{\pi}{ N-p}} \; \; \; \; 
\end{align}
$L_{1/2}$ being the Laguerre polynomial of degree 1/2.
The relevance of formula~\eqref{eq:formule} is checked on numerical examples next section. As we shall see, the lower bound is tight when $p$ does not exceeds $\approx N/2$.

Let us finally stress that formula~\eqref{eq:formule} approximates the bias averaged on the mean motion, that is the frequency of the orbit, not the period. Averaging on the period would give more weight to the bias at low frequencies, which is high, and would therefore lead to a greater average value of the bias.

\subsection{Precision and accuracy}
\label{app:lindependency}

\REWRITE{
Eq.~\eqref{eq:defmse} expresses the mean squared error (MSE) as a function of the bias and the standard deviation. Assuming the model is correct, the MSE is an accuracy metric (dispersion about the true value), while the variance captures the precision of the estimate (dispersion of the estimate about its mean value). \cite{shenturner2008, zakamska2011} have shown that the estimates accuracy degrades -- all other parameters being fixed -- as the SNR decreases, as the period of the planet is longer, and as phase coverage degrades. }

\REWRITE{
In this appendix, we argue that theses effects can be seen as degrading the precision of the eccentricity estimates. More precisely, we show that the bias is proportional to the standard deviation of the estimate, so that the MSE is also proportional to the standard deviation. }

\REWRITE{
	We proceed with a numerical simulation. We consider the 74 measurement times of HD 69830~\citep{lovis2006}, spanning on 800 days, as they are spaced in a typical manner. We then inject a simulated planet in circular orbit and a white, Gaussian noise of standard deviation 1 m.s\textsuperscript{-1}. In all the following simulations, the phase is uniformly random. By default, the semi-amplitude of the the planet is $K$ = 3 m.s\textsuperscript{-1} and the period is 31.56 days period (like HD 69830 c).}
		\begin{enumerate}
	\item \REWRITE{Data points are taken off two by two, from 74 to 14.}
	\item \REWRITE{The semi amplitude of the planet varies from 0.5 to 6.3 m.s\textsuperscript{-1} by a step of 0.2 m.s\textsuperscript{-1}.}
	\item \REWRITE{The period is drawn from a log-normal law, where $\log_{10} P \sim G(1,1)$. Thirty different periods are drawn.}
	\item \REWRITE{The phase coverage is degraded. We consider the 25 times $t_k = k \times P$ where $P = 31.56$ is the planet period. For each $t_k$, three epochs of measurements are drawn uniformly between $t_k$ and $t_k + \Delta t$. We choose thirty different lengths for $\Delta t$, equispaced from $P/4$ to $P$. The rationale is to generate observations more or less localised around the same time when folded in phase at $P$.}
	\end{enumerate}

	\REWRITE{
	Each of the four simulations is made with thirty different parametrizations. For each of the $30 \times 4$ of these, we generate 500 realisations of white noise and report the average value of eccentricity (the bias) and the standard deviation of the estimate. The results are reported in Fig.~\ref{fig:mseproptosigma}, where it appears that the bias is proportional to the standard deviation. The points obtained with simulations 1 to 4 described above correspond to the blue, red, yellow and purple points on the  graph. 	The analytical approximation~\eqref{eq:bias00} suggests that the bias should be proportional to the standard deviation of the estimate with a factor $\sqrt{\pi/(4 - \pi)}$. The black line, which represents $y = \sqrt{\pi/(4 - \pi)} x$ is in close agreement with the scatter observed. }

	\REWRITE{
The $\sigma_e$ reported in Fig.~\ref{fig:mseproptosigma} is computed as the standard deviation of the estimate. Note that the bias is also proportional to the uncertainty on eccentricity computed from the correlation matrix, obtained from the least square fit.}

\begin{figure}
	\centering
	\begin{minipage}[l]{0.45\textwidth}
		\centering
		\includegraphics[width=8.1cm]{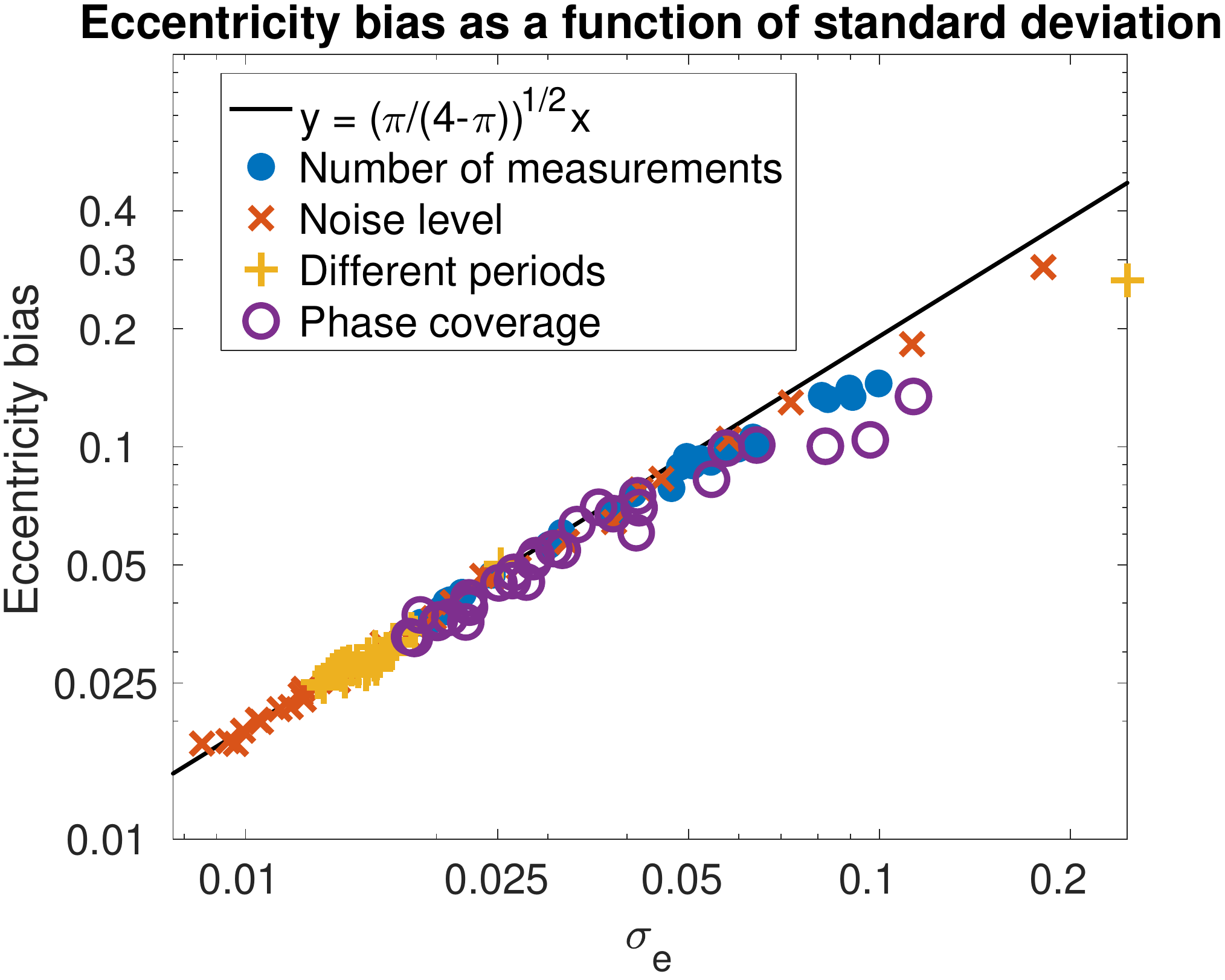}
		\caption{Bias as a function of the standard deviation of the eccentricity estimates in different configurations.} 
		\label{fig:mseproptosigma}
	\end{minipage}
	\hspace{5mm}
\end{figure}

\section{Local $\chi^2$ minima at high eccentricities}
\label{app:complicatedshape}

As shown in~\cite{baluev2015}, the number of local minima increases significantly in the high eccentricities region. These minima might lead a local minimisation algorithm or a Monte Carlo Markov Chain (MCMC) to be stuck in the wrong region of the parameter space. We here aim at quantifying  and understanding this feature.

In section~\ref{sec:moremodels}, we defined  a notion of SNR (eq.~\eqref{eq:snr}).  Interestingly enough, one can adapt this notion to predict the number of local minima at high eccentricity, provided the value of the SNR is checked to be reliable.

\begin{figure*}
	
	\centering
	\hspace{-0.3cm}
	\begin{minipage}[l]{0.49\textwidth}
		\includegraphics[width=8.7cm]{nlmin_strue.pdf}
		\caption{Blue bins: binned values of the number of systems with 1,2,3,4,5,6 or 7 local minima, with a bin size in true SNR $S_\mathrm{t}$ of 5. Red curve: fraction of the binned systems where the global minimum is not attained at the one obtained with a linear fit.}
		\label{fig:nlminst}
	\end{minipage} \hfill
	\begin{minipage}[c]{0.49\textwidth}
		\includegraphics[width=8.7cm]{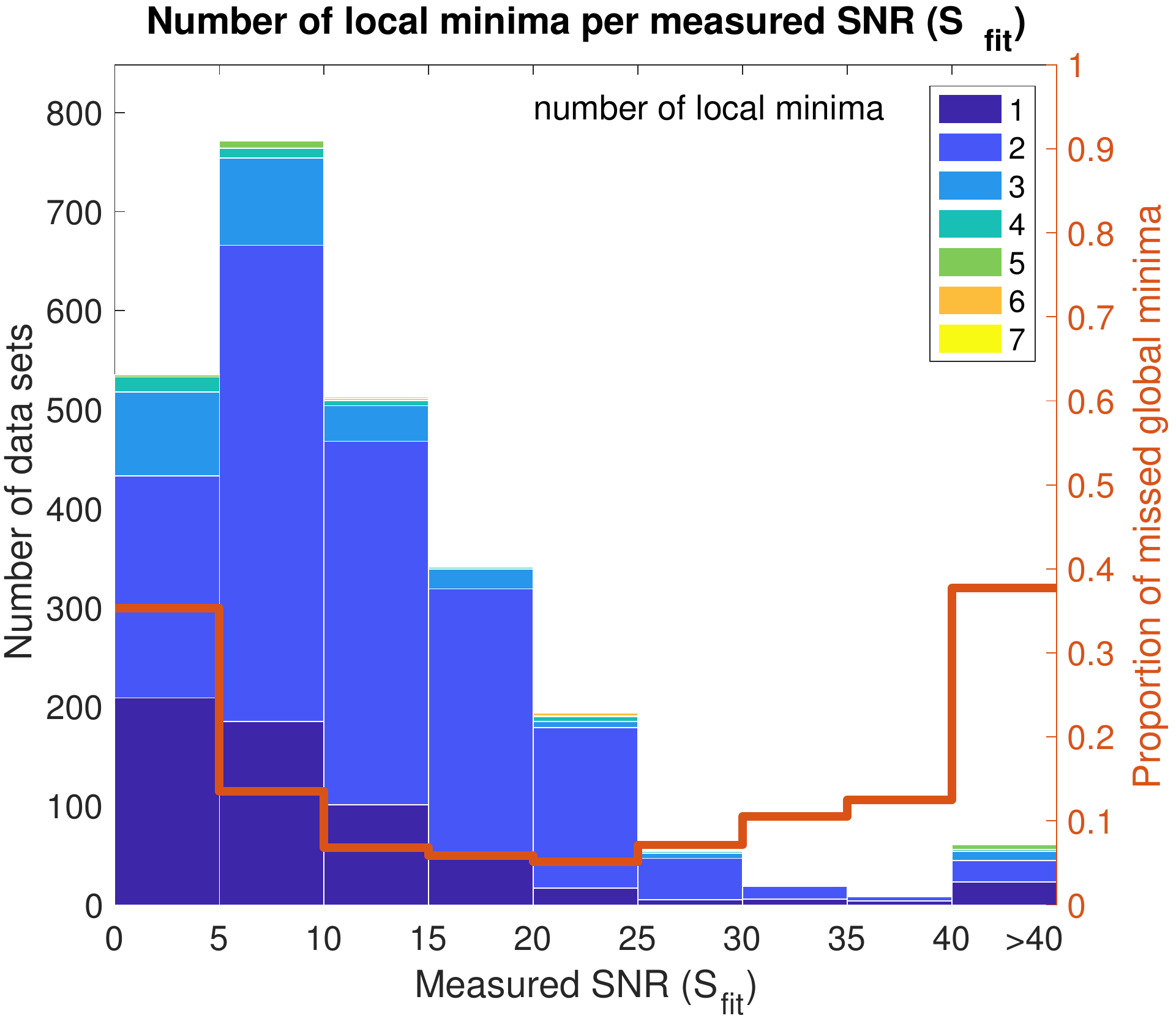}
		\caption{Blue bins: binned values of the number of systems with 1,2,3,4,5,6 or 7 local minima, with a bin size in fitted SNR $S_\mathrm{fit}$ of 5. Red curve: fraction of the binned systems where the global minimum is not attained at the one obtained with a linear fit.}
		\label{fig:nlminsfit}
	\end{minipage}\\ [0.6cm]
\end{figure*}
\begin{figure}
 	\includegraphics[width=8.7cm]{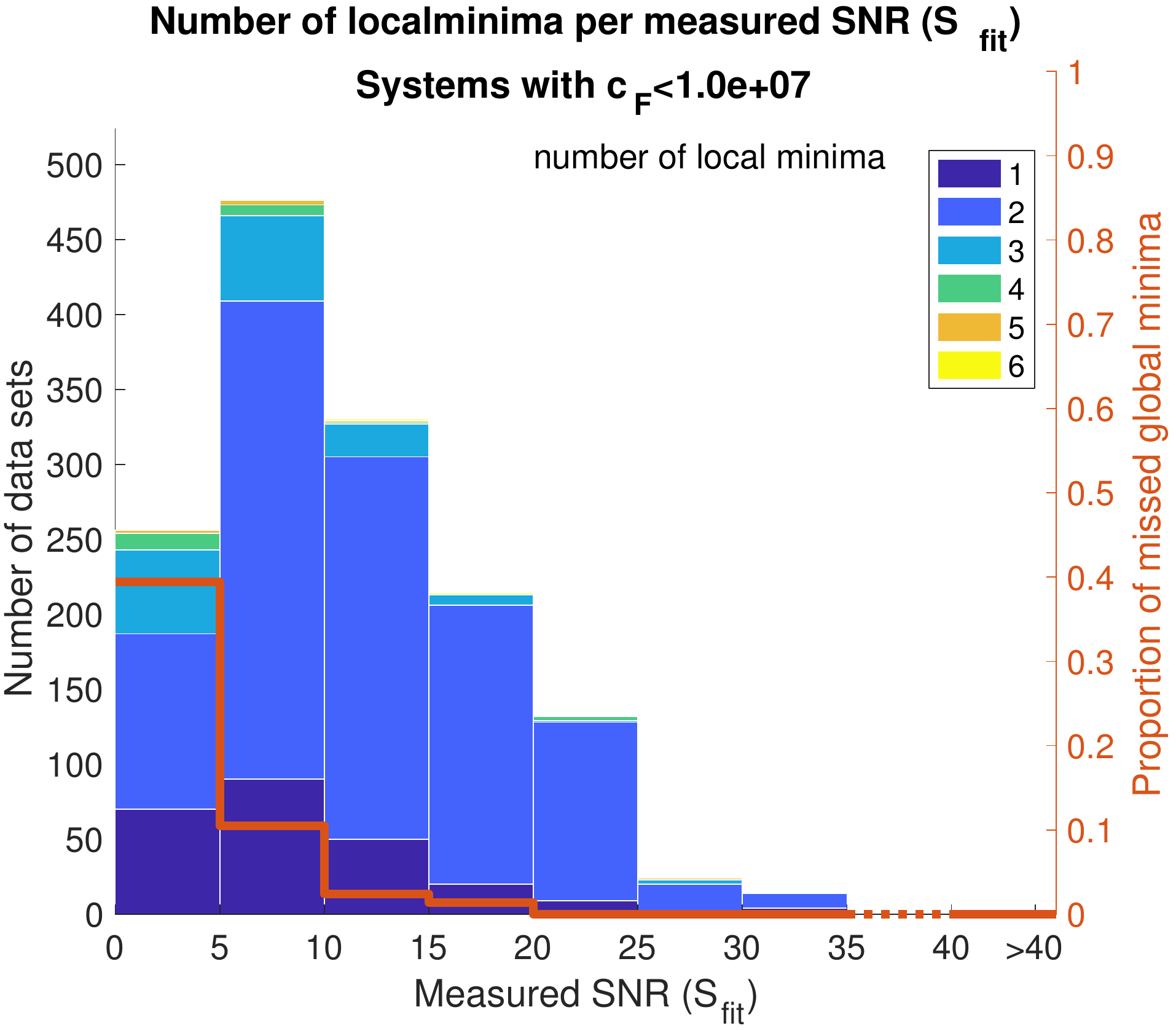}
	\caption{Blue bins: binned values of the number of systems that have a condition number lower than $10^7$ with 1,2,3,4,5, or 6 local minima. The bin size in fitted SNR $S_\mathrm{fit}$ of 5. Red curve: fraction of the binned systems where the global minimum is not attained at the one obtained with a linear fit.}
	\label{fig:nlminsfitcf}
\end{figure}

\subsection{Number of local minima per SNR}
	\label{app:highecc_snr}
We explore through simulations how many local minima should be expected, and how frequently the local minimum given by a non-linear fit starting at a circular orbit is not the global minimum.
Our simulation is structured as follows. We consider the measurement times of CoRoT-9~\citep{bonomo2017a}, Gl96~\citep{hobson2018}, of HD69830~\citep{lovis2006},  HD40307~\citep{mayor2009}, which have respectively 28, 67, 74 and 129 observations and  generate a planet with uniform distribution in $e$, $\omega$, $M_0$ and $K$ (the latter on [0,3]). The period is drawn from a  log normal distribution ($\log_{10} P \sim G(1.5,1)$). The uncertainties on the measurements are taken as the nominal ones that are normalised to obtain a mean variance of 1. We then obtain the fitted uncertainty $\sigma_{\mathrm{fit}}^2$.

For each realisation of the true planet and the noise, we perform first a local minimisation, initialized at a circular orbit. We then compute three quantities at this local minimum: the condition number of the Fisher matrix $\mathbfss{F}$, as well the SNR, similarly to equation~\eqref{eq:snr},
\begin{align}
c_F &= \frac{\max_{i=1..N} \lambda_i}{\min_{i=1..N} \lambda_i} \; \; \; \text{with} \; \; \; (\lambda_i)_{i=1..N} =\mathrm{eigenvalues}(\mathbfss{F}) \\
S_{t} &= \frac{K_t}{\sigma_t} \sqrt{ \frac{N-p}{2} } \label{eq:st} \\
S_{\mathrm{fit}} &= \frac{K_{\mathrm{fit}}}{\sigma_{\mathrm{fit}}} \sqrt{ \frac{N-p}{2} } \label{eq:sfit} 
\end{align}
where $K_t$ and $K_{\mathrm{fit}}$ are respectively the true and fitted values of the semi-amplitude, $N$ is the number of measurements, $p$ the number of fitted parameters.  $\sigma_{\mathrm{fit}}$ is defined as follows.  We adjust a term $\sigma_J^2$ so that the reduced $\chi^2$ is equal to one and take $\sigma_{\mathrm{fit}}$ as the mean of $ (\sigma_J^2 + \sigma_k^2)_{k=1..N}$. The rationale of taking this definition instead of the true SNR is to have a quantity that does not require to know the true orbital elements  and can be computed on a real data set. 

Secondly, we perform a Keplerian periodogram on a closely spaced grid of $e \in [0, 1]$, $\omega \in [0, 2\pi)$ and $n \in [n_t-1/T_\mathrm{obs}, n_t+1/T_\mathrm{obs}]$, where $n_t$ is the true mean motion. We compute the value of the eccentricity where the global minimum of $\chi^2$ is attained, as well as the number of local minima. For instance, in Fig.~\ref{fig:proximab}, the red curve displays four local minima at $e =0.17, 0.75, 0.92$ and $1$, and the global minimum is attained at $e=0.92$. In that case, a local minimisation would not give the global minimum. 

For both definitions of the SNR, we compute the number of system that have a SNR between $5k$ and $5(k+1)$, $k=0..7$ and greater than 35. In each bin, we compute the number of systems with 1,2,3,4,5,6 or 7 local minima. The results are represented in Fig.~\ref{fig:nlminst} and~\ref{fig:nlminsfit} for $S_{t}$ (eq.~\eqref{eq:st}) resp. $S_{\mathrm{fit}}$ (eq.~\eqref{eq:sfit}) by the blue histograms. We then compute in each bin the fraction of system where the global minimum is not the local minimum closest to 0, which means that most likely, a local minimisation does not yield the global minimum (red stair curve, see the right $y$ axis scale). For instance, in Fig.~\ref{fig:nlminst}, we see that out of the 2500 systems
simulated, 472 had a St between 0 and 5, 176 of which had
only one local minimum, 203 had two, 84 had 3, 9 had 4 and
none had more. Among those 472 systems, 34\% of them had
a local minimum that is not the global minimum.

Note that in Fig.~\ref{fig:nlminsfit}, the systems with $S_{\mathrm{ fit}}>35$ present the highest proportion of missed global minima. This is due to the fact that no systems with very high SNR where generated. As a consequence, all the very high values of $S_{\mathrm{ fit}}$ result from datasets with very low $S_t$ where $K$ was very overestimated. To obtain a more reliable diagnostic, we need to determine if a fitted SNR can be trusted. The criterion we used is to select only the data sets where $c_F < 10^7$. The figure~\ref{fig:nlminsfitcf} shows the number of local minima per bin of $S_\mathrm{fit}$ so obtained. The predominance of systems with two local minima is due to the fact that at very high eccentricity, there is in general a decrease of $\chi^2$. As a consequence, the second local minima is attained at the maximum eccentricity used for the calculation of the global periodogram. For instance in the case of Proxima b (Fig.~\ref{fig:proximab}), there is such a minimum at $e  = 0.999$. Note if we do the same analysis per data set (CoRoT-9, Gl96, HD69830, HD40307), we obtain very similar figures, which shows that indeed the SNR is a reliable metric for the number of local minima. 
As a conclusion, we expect that the exploration of the parameter space to be more difficult at low SNR, since there are more local minima to expect.

\subsection{Interpretation}
\label{app:highecc_inter}

We now give a geometrical interpretation of why the high eccentricity region is prone to having local minima. 

Finding the best fitting model amounts to finding the model closest to the observation in a geometrical sense. 
We consider the figure drawn in $\mathbb{R}^N$ by all the  models that have an eccentricity $e$  and a period $P$, $\mathcal{M}_{e,P}$.  This figure might explore more or less dimensions. For instance, if it is close to a plane, it is nearly confined to a two-dimensional space. Otherwise, exploring many dimensions traduces a ``rough'' surface, which increases the chances of finding a local minimum of distance to the data. 

By a procedure based on singular value decomposition, detailed in~\cite{harathese} section 4.3, we obtain an approximate number of dimensions explored by $\mathcal{M}_{e,P}$ as a function of $e$. We here provide  a brief description of our methodology. The measured velocity can be expressed as a linear combination of the velocity components in the orbital frame $\dot{X},\dot{Y}$. With these variables, equation~\eqref{eq:vexprbis} becomes
\begin{align}
y(t,A,B,\omega,e,P) = A\dot{X}(t,\omega,e,P) +B\dot{Y}(t,\omega,e,P).
\label{eq:keplerab}
\end{align} 
The components $\dot{\mathbfit{X}}(P,e,\omega) = (\dot{X}(t_k,P,e,\omega))_{k=1..N}$, $\dot{\mathbfit{Y}}(P,e,\omega)= (\dot{Y}(t_k,P,e,\omega))_{k=1..N}$ are
computed for a grid of $\omega$, $(\omega_k)_{k=1..n}$. Those vectors are concatenated to form a matrix $ M(e,P) = [\dot{\mathbfit{X}}(P,e,\omega_1)..\dot{\mathbfit{X}}(P,e,\omega_n)\dot{\mathbfit{Y}}(P,e,\omega_1)..\dot{\mathbfit{Y}}(P,e,\omega_n)] $, whose columns are  normalized to obtain $\tilde{M}(e,P)$. We  compute  the number of singular values of  $\tilde{M}(e,P)$ that are greater than a tenth of the maximum singular value. This constitutes a proxy for the number of dimensions explored by models with $e,P$ fixed.

 As an example, we perform this calculation on the 214 measurement times of GJ 876~\citep{correia2010}. 
The number of explored dimensions are shown in table~\ref{table:dimension}. As eccentricity increases, the models explore more dimensions.
Since the models at high eccentricity occupy a very large volume in many dimensions, there will often be at least one high eccentricity model closely fitting the data. 

\section{Frequentist methodologies}
\label{sec:fisheriansection}
\subsection{Presentation}

Testing possible eccentricities can also be done in a frequentist framework. This one offers confidence intervals, which are not as easy to interpret as Bayesian credible intervals but have the advantage of being quicker to compute. Furthermore, the associated algorithms have clearer convergence tests.

So far frequentist inferences for eccentricities have been done in several ways.~\citep{lucy1971} and~\cite{husnoo2012} respectively used $p$-values and Bayesian Information Criterion (BIC) to test the hypothesis that eccentricity is non zero. More precisely,~\cite{lucy1971} compute the probability distribution of the eccentricity estimate under the hypothesis that the eccentricity is null and find a Rayleigh distribution whose variance depends on the SNR (which we also obtain as a special case of our analysis section~\ref{sec:moremodels} with $p=0$). For a given measured eccentricity $\widehat{e}$, they measure the probability that the Rayleigh distribution is higher than $\widehat{e}$ and report an eccentric orbit if this probability is lower than a certain threshold.

~\cite{husnoo2012} computes 
\begin{align}
\mathrm{BIC}(\mathcal{M})  = \chi^2_{\mathrm{min}}(\mathcal{M}) + p \ln N + \ln (2\pi|\mathbfss{V}|)
\label{eq:bic}
\end{align}
where $\chi^2_{\mathrm{min}} $ is the minimum $\chi^2$ obtained when minimising the distance between the data and model $\mathcal{M} $, $p$ is the number of degrees of freedom of $\mathcal{M}$ (three for a sine model and five for a Keplerian one) and $|\mathbfss{V}|$ is the determinant of the correlation matrix. The orbit is said to be eccentric  if $\mathrm{BIC}(\mathcal{M}_{\mathrm{ecc}}) \geqslant \mathrm{BIC}(\mathcal{M}_{\mathrm{circ}})$ where $\mathcal{M}_{\mathrm{ecc}}$ and $\mathcal{M}_{\mathrm{circ}}$ are respectively eccentric and circular models.

Though reasonable, these techniques can be improved. First, they both consider the alternative $e$ is zero or non-zero, and do not allow to test if a given value of eccentricity is compatible with the data or not. Secondly, the analytical approximation of the eccentricity distribution is not always accurate.
Finally, the Bayesian information criterion~\eqref{eq:bic} gives equal weight to all parameters, only their number $k$ appears. This approximation is valid in the limit of a large number of observations.

\subsection{New methods}

Our aim is to overcome as much as possible these limitations. It turns out that the procedure to construct confidence intervals outlined in~\cite{casellaberger2001}, chapter 9, allows us to test the hypothesis that the true eccentricity is equal to a certain value $e$ for all $e$.  The idea is to reject the hypothesis that eccentricity is equal to $e$ if all models with eccentricity $e$ have a likelihood lower than a fraction of the maximum likelihood. The following criterion is computed in Appendix~\ref{fisherian}.  We reject the hypothesis that the eccentricity has a certain value $e$ with a confidence level $\alpha$ if   
\begin{align}
&LR :=  \frac{\max\limits_{\btheta \in \Theta_e} f(\mathbfit{y}|\btheta)  }{\max\limits_{\btheta \in \Theta} f(\mathbfit{y}|\btheta) } \leqslant \e^{-\frac{1}{2}\beta} \label{eq:lrt} \\
&\beta = F_{\chi^2_\rho}^{-1}(1 - \alpha) \label{eq:lrt2} \\
&\rho = 2 + 2S'^2\frac{e^2}{1+e^2} - \frac{\pi e}{1+e^2}L_{\frac{1}{2}}\left(-\frac{S'^2}{2}\right)
L_{\frac{1}{2}}\left(-\frac{e^2S'^2}{2}\right) \label{eq:lrt3}.
\end{align}
where  $\Theta_e$ is the set of parameters that have all eccentricity $e$, $f(\mathbfit{y}|\btheta)$ is the likelihood,  $F_{\chi^2_\rho}^{-1}$ is the inverse cumulative distribution function of a $\chi^2$ law with $\rho$ degrees of freedom, $S' = (\sigma /K_t) \sqrt{2/N}$  and $L_{\frac{1}{2}} $ is the Laguerre polynomial of order $1/2$.  The quantity~\eqref{eq:lrt} is simply the ratio of the maximum likelihood obtained by restriction to the models with fixed eccentricity divided by the maximum likelihood on all models. The condition states that if all models that have eccentricity $e$ have too low a likelihood, then $e$ is rejected. The following equations give the value of that threshold, which is obtained by calculating the law followed by the random variable $LR$ under the hypothesis that the true eccentricity is $e$ ($LR|(e_t=e)$). It is in fact easier to compute the law followed by the logarithm of $LR$, to obtain a $\chi^2$ law whose degree depends on a definition of the SNR $S$ and on the eccentricity under study, but is always smaller than 2. Our computations, detailed in Appendix~\ref{fisherian}, also make use of simplifying assumptions, but these are checked to give satisfactory results on simulated signals. 

One of the problems of that expression is that it depends on the true value of the semi-amplitude, $K_t$, which is unknown. There are two ways to circumvent this issue: either by assuming that $\rho = 2$ for all $e$, which is the maximum value $\rho$ can take, of $K_t$, or by approximating $K_t$ by the semi-amplitude of a circular orbit fitted at the period of the signal. The first option can be used to obtain conservative intervals to ensure that $e$ is non zero. The second one gives a more realistic criterion to reject an eccentricity if no extra care is needed. Let us note that $\rho=2$ is obtained for $e=0$. This has a simple interpretation: the model can be approximated by a linear one in $k = e\cos \omega$ and $h=e\sin \omega$. When $e=0$, both $k$ and $h$ are set to zero, which blocks two degrees of freedom. Denoting by $\mathbfit{y}_e$ the model with fixed eccentricity $e$  that has maximum likelihood and $\mathbfit{y}^\star$ the model with maximum likelihood, all parameters free, 
\begin{align}
0.5\ln(LR) = \|\mathbfss{W}(\mathbfit{y} -\mathbfit{y}_e)\|^2 - \|\mathbfss{W}(\mathbfit{y} -\mathbfit{y}^\star)\|^2.
\label{eq:lrt4}
\end{align}
behaves then as a $\chi^2$ law with two degrees of freedom.

To offer a point of comparison of the confidence interval computed equation~\eqref{eq:lrt3}, two other metrics are tested. First, we simply consider 
\begin{align}
F = \frac{N-p}{\rho} \frac{\|\mathbfss{W}(\mathbfit{y} -\mathbfit{y}_e)\|^2 - \|\mathbfss{W}(\mathbfit{y} -\mathbfit{y}^\star)\|^2}{\|\mathbfss{W}(\mathbfit{y} -\mathbfit{y}^\star)\|^2}
\label{eq:lrtfratio}
\end{align}
which is basically equation~\eqref{eq:lrt4}, normalized by $\|\mathbfss{W}(\mathbfit{y} -\mathbfit{y}^\star)\|^2$ so that it depends less on the noise level assumption. The quantity~\eqref{eq:lrtfratio} is assumed to follow a $F$ distribution with $\rho$ and $N-p$ degrees of freedom.

Secondly, we generalize the test suggested by~\citeauthor{lucy1971}. Let us denote by $e^\star$ the estimate of eccentricity obtained by maximum likelihood when all parameters are free. For eccentricity $e$, we fit a Keplerian model that has an eccentricity fixed at $e$. We then compute the probability 
\begin{align}
\mathrm{Pr}\{|\widehat{e} - e| > |e^\star - e| |e, \mathbfss{V}, \widehat{e} \sim \mathrm{Rice}(e, \eta^2) \}
\label{eq:lucyswingen}
\end{align}
that is the probability that an eccentricity estimated by maximum likelihood  $\widehat{e}$ is at least as far from its assumed value $e$ than the distance between $e$ and the best fit actually observed, assuming the noise model is Gaussian with known covariance matrix $\mathbfss{V}$. We also assume that $e$ follows a Rice distribution as in appendix~\ref{appendix_realformula}. A Rice distribution can be seen as the modulus of a vector with two independent Gaussian variables that have the same variance, $X \sim G(a, \eta^2)$ and $Y \sim G(b, \eta^2)$ where $k$ and $h$ are the means of these variables. To specify the distribution, we need therefore two scalars: the variance of both random variables $\eta^2)$ and the modulus of the mean of these two variables, $r = \sqrt{a^2 + b^2}$. Here $a=k$ and $b=h$, so $r=e$. Then $\eta^2$ is the variance of the estimates of $k$ or $h$, which under the hypotheses of section~\ref{sec:moremodels} have the same variance $\eta^2 = (\sigma_{\mathrm{RV}}^2/K_e^2)(\pi/(N-p)) $. Then the quantity~\eqref{eq:lucyswingen} can easily be evaluated by the cumulative distribution function of the Rice distribution, which is a Marcum $Q$-function.

Computing~\eqref{eq:lrt4},~\eqref{eq:lrtfratio} or~\eqref{eq:lucyswingen} necessitates to compute the minimum distance between the observations and a model with fixed eccentricity. To do so, we exploit the fact that Keplerian models are partly linear,
$\mathbfit{y}(\mathbfit{t},\btheta) = A \dot{\mathbfit{X}}(P,e,\omega) + B \dot{\mathbfit{Y}}(P,e,\omega)  + C $
where $\dot{\mathbfit{X}}$ and $\dot{\mathbfit{Y}}$ are the components of the velocity on the orbital plane. For each couple $e,\omega$, we can minimize  $\| \mathbfit{y} - \mathbfit{y}(\mathbfit{t},\btheta) \|$ on $A,B,C$ and $P$, which are respectively three linear parameters and one non-linear parameter. Such a problem is fast to solve with, for instance, a Levenberg-Marquardt algorithm~\citep{levenberg1944, marquardt1963}. If the period is already known (which is supposed here), obtaining an array of $\chi^2$ on a fine grid of $e$ and $\omega$ (60 values each) takes only up to one minute. Let us finally note that the idea of restricting the global $\chi^2$ minimization to a grid of non linear parameters is not new~\citep{hartkopf1989,lucy2014}. There are even further resemblances of our interval calculation with~\cite{lucy2014}, where confidence intervals on orbital parameters are computed in a similar way. However, ~\cite{lucy2014} uses a degree of freedom $\rho = 1$ for all parameters. This is correct only if the model is linear in all the parameters or approximately linear in the vicinity of the best fit and unimodal.

\subsection{Tests}
\label{sec:fisheriantest}

The formula~\eqref{eq:lrt2} and~\eqref{eq:lrt3} have been derived with simplifying assumptions. 
To test and compare them to other options, we proceed as follows. We define the acceptable interval as the set of $e$ where $LR_e \leqslant \exp(-0.5 F_{\chi_\rho^2}(1-\alpha))$.
\begin{enumerate}
	\item We generate a population of exoplanets according to a certain prior density on the orbital elements $p(K,e,P,M_0,\omega)$. The measurement times are taken from existing data sets. The noise generated according to a Gaussian density of covariance matrix $V$. 
	\item For each system, we compute the set of eccentricity that are not rejected, we check that the true eccentricity belongs to this set and compute the measure of its complement in $[0,1]$, that is the measure of the set of rejected eccentricities. 
	\item The results are summarized in two plots. First, the fraction of cases where the true eccentricity is not in the acceptable interval as a function of $\alpha$. Second, the curve drawn when $\alpha$ goes from 0 to 1 by a point whose ordinate is measure of the complement of the set of acceptable eccentricity and whose abscissa is the fraction of cases where the true eccentricity is not in the acceptable interval. 
\end{enumerate}

Such tests were carried out with the following inputs: the measurement times are those of CoRoT-9~\citep{bonomo2017a}. The angles $\omega$ and $M_0$ are chosen uniformly, $e$ follows an uniform distribution. In Fig.~\ref{roc} a) and b)we plot the result of the experiment for a period fixed at 95 days  the semi-amplitude is fixed to $K = 3.5\sigma$ where $\sigma$ is the RMS of the errors. These are the parameters of the detected Jupiter in the system. In Fig.~\ref{roc} c) and d), we let the period vary uniformly in $\log P$ and compute the same quantities.

Plots a) and c) of Fig.~\ref{roc} are labelled ``ROC - like'' curve as a reference to Receiver-Operator Characteristic. These ones are defined when the data are used to decide between two hypothesis. The ROC curve represents the fraction of false positives as a function of the fraction of false negatives for a given decision rule. We adapt this notion to our case, where there is an infinity of hypotheses (each $e$ in $[0, 1]$ is a hypothesis). For a given rate of true eccentricity that is not in the acceptable interval (false negatives), the $y$ axis gives the precision on the estimate. The more eccentricities are rejected, the more precise the estimate.  The closer such a curve is to the upper left corner the better: regardless of the value of $\alpha$, the fraction of true $e$ rejected is zero (no false negatives) and almost all other eccentricities are rejected: the estimation is very precise.

Interestingly enough, the ROC curve (left) is very similar for all the metric considered with a slight advantage for the $F$-ratio and the likelihood ratio tests (formula~\eqref{eq:lrt4} and~\eqref{eq:lrtfratio}), which have a better precision (more eccentricities rejected) when the fraction of true $e$ rejected is low. We now need to set the level of true $e$ rejected. 
As expected, the curve obtained for $\rho=2$ gives an overestimated error rate for a given $\alpha$. For the three other tests, the correspondance seems appropriate. Overall, the  $F$-ratio and the likelihood ratio tests seem to perform best. 

One advantage of the frequentist method is that it relies only on local minimization algorithms, therefore it is fast and convergence is ensured. We have also shown that the parameter $\alpha$ (see eq.~\eqref{eq:lrt2}) allows us to directly control the confidence intervals meaning. Let us  mention that we observed some peculiar behaviour of the estimates for some periods where the matrix of the linearized model is ill-conditioned, that we wish to investigate into more depth in future studies. In those cases, the hypotheses allowing to compute formula~\eqref{eq:lrt3} are not verified and Bayesian analysis or more sophisticated formula would be required. On the other hand, it seems unlikely that someone would want to prove a non zero eccentricity of a planet particularly poorly sampled.

\begin{figure*}
	\centering
	\begin{tikzpicture}
	\path (0.05,0) node[above right]{\includegraphics[scale=0.35]{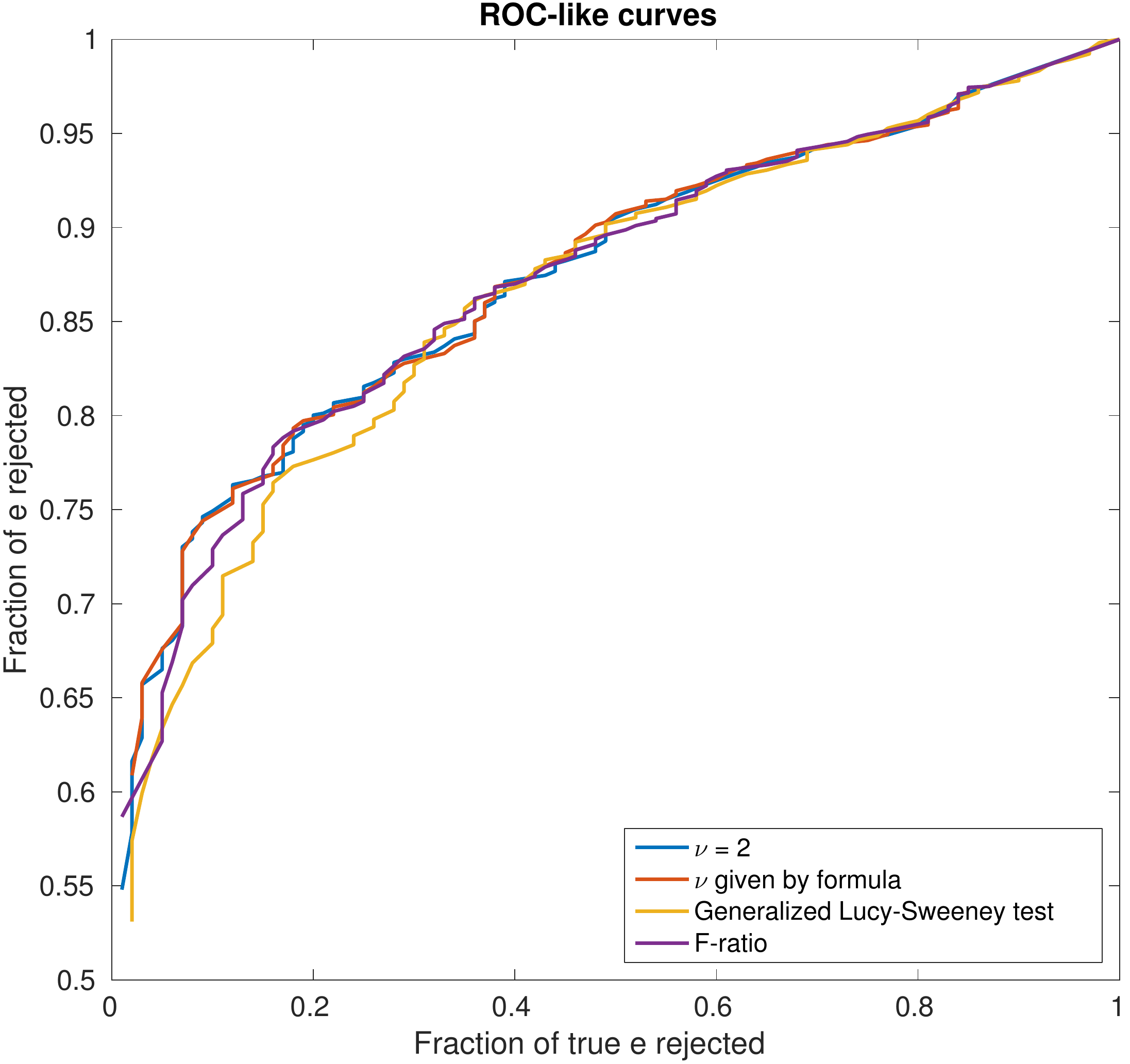}};
	\path (-0.5,7) node[above right]{a)};
	\path (9,0) node[above right]{\includegraphics[scale=0.35]{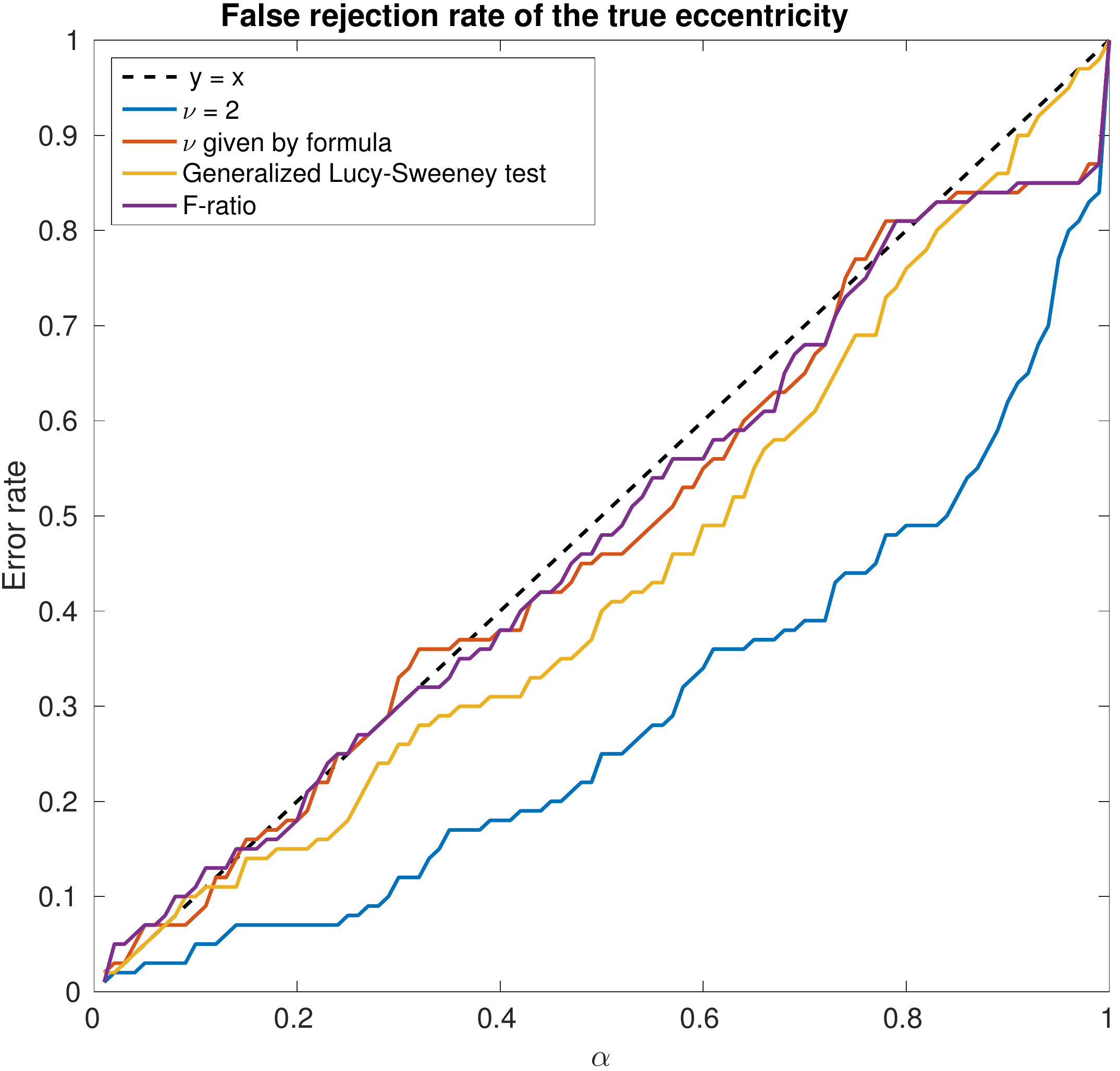}};
	\path (8.5,7) node[above right]{b)};
	\begin{scope}[yshift=-7.2cm]
	\path (0.05,0) node[above right]{\includegraphics[scale=0.4]{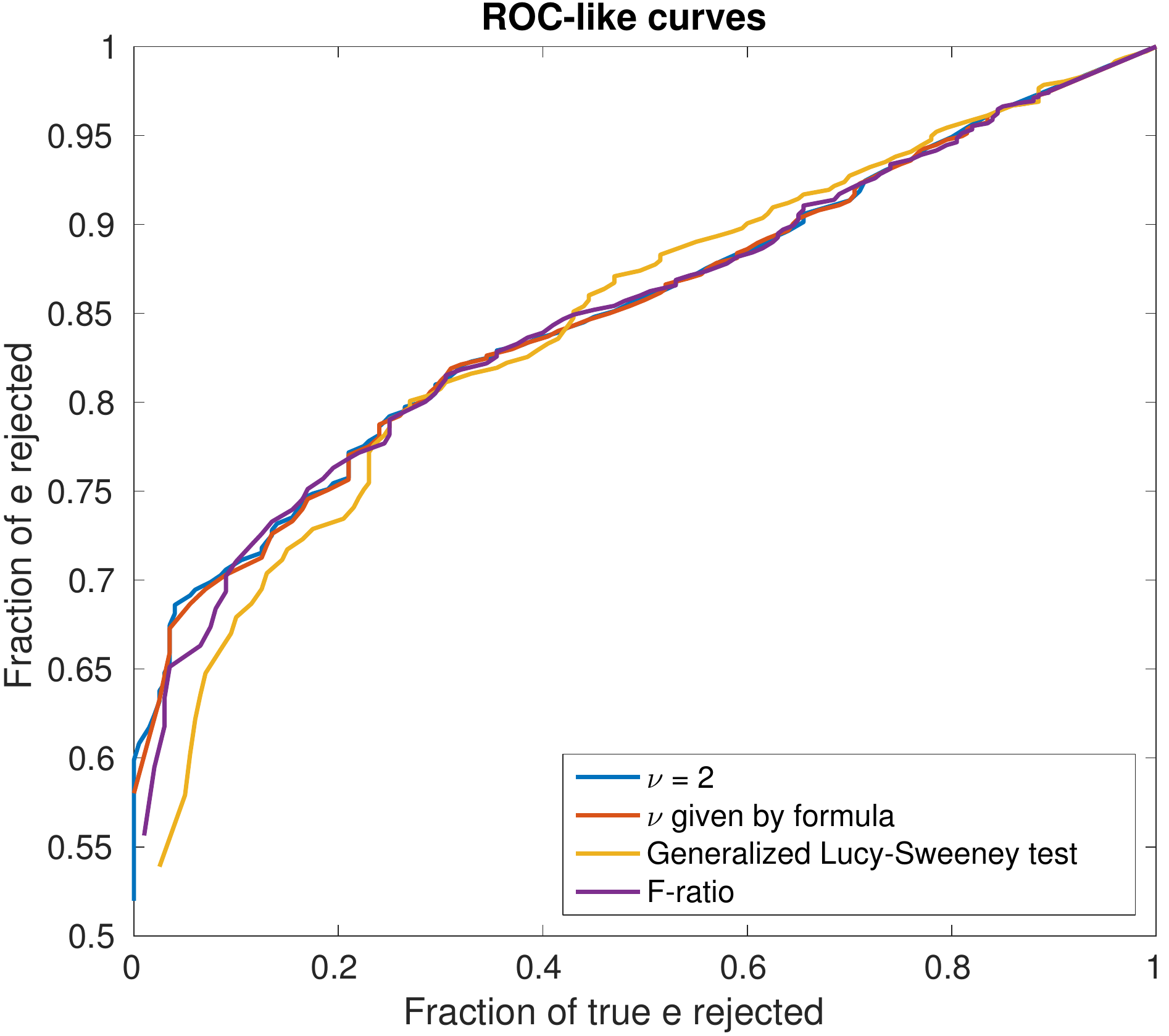}};
	\path (-0.5,6.5) node[above right]{c)};
	\path (9,0) node[above right]{\includegraphics[scale=0.4]{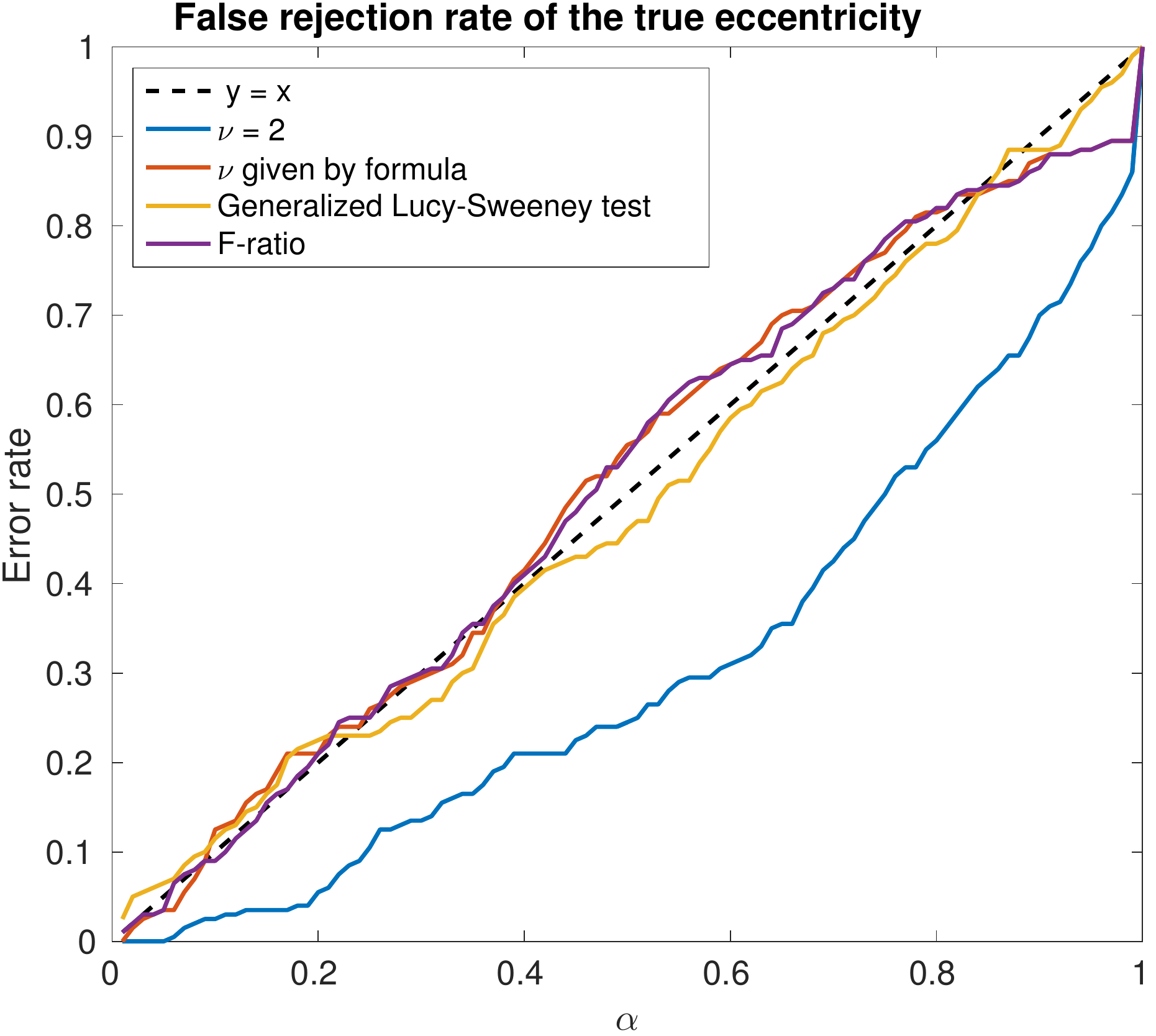}};
	\path (8.5,6.5) node[above right]{d)};
	\end{scope}
	\end{tikzpicture}
	\caption{Results of simulations described by steps (i) to (iii) of section~\ref{sec:fisheriantest}. Here $K$ = 3.5 $\sigma$, $\omega$ and $M_0$ are chosen uniformly in $[0, 2\pi]$, $e$ is chosen uniformly in $[0, 0.99]$. For figures a) and b), the period is fixed to 95 days, while it is chosen uniformly in $\log P$ for plots c) and d). Figures a) and c) represent the fraction of eccentricities rejected (or equivalently, the measure of the set of rejected eccentricities) as a function of the rate of rejection of true eccentricity. Figures b) and c) represent the rate of rejection as a function of $\alpha$. The blue, red, yellow and purple curves are respectively obtained with rejection criteria given by equations~\eqref{eq:lrt}-\eqref{eq:lrt2} with $\rho =2$, equations~\eqref{eq:lrt}-\eqref{eq:lrt2} with $\rho$ given by eq.~\eqref{eq:lrt3}, eq.~\eqref{eq:lrt4} and eq.~\eqref{eq:lrtfratio} with $\rho$ given by eq.~\eqref{eq:lrt3}. In the figure legends, $\nu$ refers to $\rho$. }
	\label{roc}
\end{figure*}

\subsection{Confidence interval calculation}
\label{fisherian}
In this section, we outline the calculation of the confidence intervals for eccentricity. Such an interval is constructed as a set of eccentricities that are not rejected by a hypothesis test. We choose the likelihood ratio test:
\begin{align*}
e \; \text{ is rejected if} \; \; \; \;  R :=  \frac{\max\limits_{\btheta \in \Theta_e} f(\mathbfit{y}|\btheta)}{ \max\limits_{\btheta \in \Theta} f(\mathbfit{y}| \btheta)} \leqslant \beta  
\end{align*}
where $\mathbfit{y}$ denotes the actual observations, $f(\mathbfit{y}|\btheta)$ denotes the likelihood, $\Theta_e$ is the set of parameters that have eccentricity $e$, and $\beta$ is a constant which will be made explicit later. Our aim is to compute the distribution of $R$ under the  assumption that the random variable giving the observations is $\vec Y = \mathbfit{y}_t +\vec \epsilon$, $\vec \epsilon$ being a Gaussian noise.
We further assume the noise is independent and identically distributed, the condition translates to 
\begin{align}
\label{eq:etestr}
e \; \text{ is rejected if} \; \; \; \;  D:=  \| \mathbfit{y} - \mathbfit{y}(\btheta_e) \|^2 -  \| \mathbfit{y} - \mathbfit{y}^\star\|^2 \geqslant -2 \sigma^2 \ln \beta  
\end{align}
where $\btheta_e = \arg \min\limits_{\btheta \in \Theta_e}  \| \mathbfit{y} - \mathbfit{y}(\btheta) \|^2$, $\sigma^2$ is the variance of the observations and $\mathbfit{y}^\star$ is the global minimum. We  now compute the law followed by $D$, so that we can select a $\beta$ that corresponds to a false alarm probability. Since $D$ is defined implicitly, the calculation of its distribution is difficult. We make two simplifying assumptions that allow us to obtain an analytical expression. The expression will then be tested on real cases through numerical simulations. 

Let us first consider the linear approximation $\mathbfit{y} = \mathbfss{M} \mathbfit{x}_t + \bepsilon$ where $\mathbfss{M}$ is defined as in~\eqref{eq:model2} and~\eqref{eq:model2b}. We further suppose that the columns of $\mathbfss{M}$ are orthonormal. Since the columns are originally of the form $\cos n \mathbfit{t}, \sin n \mathbfit{t}, \cos 2 n \mathbfit{t},\sin 2 n \mathbfit{t}$, they must be multiplied by $\sqrt{2/N}$ and the amplitude of the signal is no $K_t$ but $K_t \sqrt{N/2}$.
We look for the solution $\hat{\vec \theta}_e$ defined as
\begin{equation}
\hat{\vec \theta}_e = \mathrm{arg}\min_{\vec x\in\mathbb{R}^p} \|\mathbfit{y} - \mat M
\vec x\| \quad \mathrm{subject\ to} \quad
\sqrt{\frac{x_3^2+x_4^2}{x_1^2+x_2^2}} = e .
\end{equation}
Thanks to the Lagrange multipliers theorem, $\hat{\vec \theta}_e$ satisfies the conditions
\begin{equation}
\frac{\partial L}{\partial \vec x} = \vec 0, \quad
\frac{\partial L}{\partial \lambda} = 0 , \quad\mathrm{where}
\end{equation}
\begin{equation}
L(\vec x, \lambda) = \frac{1}{2}\|\mathbfit{y} - \mat M \vec x\|^2 + \frac{\lambda}{2} \trans{\vec x}\mat E \vec x
\end{equation}
with
\begin{equation}
\mat E = \mathrm{diag}\left(-e^2,-e^2,1,1,0,\ldots,0\right) .
\end{equation}
The condition $\partial L / \partial \vec x = \vec 0$ leads to
\begin{equation}
\left(\trans{\mat M}\mat M + \lambda \mat E\right)\vec x = \trans{\mat M}\mathbfit{y} .
\end{equation}
Since the columns of $\mat M$ are orthonormal, $\trans{\mat M}\mat M$ is the identity, thus
\begin{equation}
x_1 = \frac{u_1}{1-\lambda e^2} , \quad
x_2 = \frac{u_2}{1-\lambda e^2} , \quad
x_3 = \frac{u_3}{1+\lambda} , \quad
x_4 = \frac{u_4}{1+\lambda} , \quad
\end{equation}
and $x_j = u_j , \ \forall j\geq 5$, where we have defined $u_i = \trans{\vec M_i}\mathbfit{y}$, $\vec
M_i$ being the $i$-th column of $\mat M$. The first four components of $\vec x$ are also constrained
by $\partial L / \partial \lambda = 0$. Let $U = u_1^2 + u_2^2$ and $V = u_3^2+u_4^2$. We get
\begin{equation}
\frac{-e^2}{(1-\lambda e^2)^2} U + \frac{1}{(1+\lambda)^2} V = 0 ,
\end{equation} 
or, equivalently,
\begin{equation}
e^2(e^2V-U)\lambda^2 - 2e^2(V+U)\lambda + V-e^2U = 0 ,
\end{equation}
whose solutions are
\begin{equation}
\lambda_\pm = \frac{e^2(U+V)\pm e(1+e^2)\sqrt{UV}}{e^2(e^2V-U)} .
\end{equation}
For the solution $\hat{\vec \theta}_e$ to actually be a minimum of $L$, all its eigenvalues must be
positive, i.e., $\lambda$ must verify $-1<\lambda<1/e^2$. Only $\lambda_-$ fulfils this criterion,
thus
\begin{equation}
\lambda = \frac{e^2(U+V)-e(1+e^2)\sqrt{UV}}{e^2(e^2V-U)} ,
\end{equation}
and
\begin{equation}
x_1 = \frac{1+e_0^2}{1+e^2}u_1 , \quad
x_2 = \frac{1+e_0^2}{1+e^2}u_2 , \quad
x_3 = \frac{e^2}{e_0^2}\frac{1+e_0^2}{1+e^2}u_3 , \quad
x_4 = \frac{e^2}{e_0^2}\frac{1+e_0^2}{1+e^2}u_4 ,
\end{equation}
with $e_0^4 = e^2V/U$. After a few calculation, we show that
\begin{equation}
D = \sum_{k=1}^4(u_k-x_k)^2 = \frac{\left(e\sqrt{u_1^2+u_2^2}-\sqrt{u_3^2+u_4^2}\right)^2}{1+e^2} .
\end{equation}
Let $x = e\sqrt{u_1^2+u_2^2}/K_t$ and $y = \sqrt{u_3^2+u_4^2}/K_t$. These two random variables
follow Rice distributions with parameters
\begin{equation}
\rho_x = e\sqrt{\frac{N}{2}} , \quad \sigma_x = \frac{e\sigma}{K_t} , \quad
\rho_y = e\sqrt{\frac{N}{2}} , \quad \sigma_y = \frac{ \sigma}{K_t} .
\end{equation}
An expansion of the product term shows that $D$ behaves approximately as a weighted sum of variables following a $\chi^2$ distribution. We can then use the Welch-Satterthwaite approximation~\citep{satterthwaite1946,welch1947}: $D$ approximately follows a $\chi^2$ distribution whose number of degrees of freedom $\rho$ is given by $\mathbb{E}\{D \}$.
In the following, we denote by $S' = \frac{K_t}{\sigma}\sqrt{\frac{N}{2}}$ the SNR.
The expected value of $D$ is
\footnotesize
\begin{equation}
\begin{split}
\mathbb{E}\{D\} &= \frac{K_t^2}{1+e^2}
\int_0^\infty \int_0^\infty (x-y)^2 f(x|\rho_x,\sigma_x) f(y|\rho_y,\sigma_y)\,\dd x \dd y , \\
&= \frac{K_t^2}{1+e^2} \left[2\sigma_x^2+\rho_x^2+2\sigma_y^2+\rho_y^2-\pi\sigma_x\sigma_yL_{\frac12}\left(-\frac{\rho_x^2}{2\sigma_x^2}\right)L_{\frac12}\left(-\frac{\rho_y^2}{2\sigma_y^2}\right)\right] \\
&= \frac{K_t^2}{1+e^2}\left[2\frac{\sigma^2}{K_t^2}(1+e^2)+Ne^2-e\pi\frac{\sigma^2}{K_t^2}
L_{\frac{1}{2}}\left(-\frac{S'^2}{2}\right)L_{\frac{1}{2}}\left(-\frac{e^2 S'^2}{2}\right)\right] .
\end{split}
\end{equation}
\normalsize
With $\rho = \mathbb{E}\{D\}/\sigma^2$, we get
\begin{equation}
\rho = 2 + 2S'^2\frac{e^2}{1+e^2} - \frac{\pi e}{1+e^2}L_{\frac{1}{2}}\left(-\frac{S'^2}{2}\right)
L_{\frac{1}{2}}\left(-\frac{e^2S'^2}{2}\right) .
\end{equation}

To obtain a confidence level $\alpha$, then we need to take $-2 \ln \beta =  F_{\chi^2_\rho}^{-1}(1-\alpha)$ where $F_{\chi^2_\rho}^{-1}$ is the inverse cumulative distribution function of a $\chi^2$ distribution with $\rho$ degrees of freedom. Conversely, it is possible to convert a measured $D$ to a probability simply by computing $\alpha_e = 1-F_{\chi^2_\rho}(D)$. The hypothesis $e_t =e$ is rejected if $ \alpha_e$ is below a certain threshold.

This formula was tested numerically. It is in very good agreements with the simulations as soon as $S'$ is above $\approx 20$. As it decreases, the average of estimated eccentricity increases (which is exactly saying that the bias increases) therefore the approximation of low eccentricities does not hold any more. The value of $S'$ can be evaluated keeping in mind that when the linearised model at $e=0$ is poorly conditioned, (matrix $\mathbfss{M}$, as defined equations~\eqref{eq:model2} and~\eqref{eq:model2b}), then the uncertainty on $k$ and $h$ is higher than given by the simple formula~\eqref{eq:sigmak} and the $S'$ analytical approximation is inoperative.

\section{Non Gaussian noise}

\label{sec:nongauss}

\subsection{Simulations}
\label{sec:nongausssim}

In this appendix, we show that the non Gaussianity of the noise has a small impact on the quality of the eccentricity estimates. 

\REWRITE{ We generate 1000 realisations of six different types of noises. These have a null mean, are independent, identically distributed and scaled to have a standard deviation $\sigma = 1$ (their covariance matrix is the identity). We consider noises that are Gaussian, Student $T$ with 3 and 4 degrees of freedom,  uniform, exponential and Poisson. We inject a circular planet on the measurement times of Gl 96 with $K = 4 \times \sigma$ and uniform $\omega$, $M_0$.  Ten periods are chosen randomly according to a log normal distribution on the 5 to 500 days interval.}

For each type of noise, each of the 10 $\times$ 1000 noise realisations chosen, we compute  the error on eccentricity ($|\widehat{e} - e_t|$, $e_t$ being the true eccentricity and $\widehat{e}$ the estimate), and the root mean square (RMS) of the residuals. \REWRITE{The RMS is a s a proxy for the estimated noise level, and thus the width of the error bars.}  The latter is the maximum likelihood estimate of the noise level for an i.i.d. noise model.

The average values of those on the 10 times 1000 realisations are reported in table~\ref{tab:estimates}. Note that we take the square root of the mean squared error (MSE).
We find that the \REWRITE{mean error on eccentricity} for non Gaussian noise is within 2\% of the value of the mean error for a Gaussian noise, and the estimated jitter  is within 5\% of the value of the mean jitter for a Gaussian noise. Only the mean squared error varies by 10\% between the Gaussian and Student $T$ distribution with 3 degrees of freedom. 
Such results are remarkable, since for instance the Poisson noise takes discrete values and is non symmetrical. 

The cumulative distribution function of the error and the estimated noise levels (the standard deviation of the residuals) are shown in Fig.~\ref{fig:nongauss_err}. 
Note that the jitter estimates have a slightly greater dispersion for the Student distributions, but as in the case of the average values we do not see striking differences. We simply note that for noise distributions with heavy tails (here Student), as expected, there is a higher fraction of cases where the noise level is severely underestimated or overestimated. We conclude that non Gaussianity does not play a significant role, except for an increased variability of the noise level estimation when the noise distribution has heavy tails.

\begin{table}
	\caption{Error, mean squared error and noise level estimate of the eccentricity estimates for \REWRITE{Gaussian and non Gaussian noises}.}
	\begin{tabular}{p{1.8cm}|p{0.6cm}|p{0.6cm}|p{0.6cm}|p{0.6cm}|p{0.6cm}|p{0.6cm}}
		Noise type &  Gauss- ian& Stu- dent $T$  3 & Stu- dent $T$ 4  & Uni- form & Expo- nential & Poi- sson  \\ \hline
		Mean error&  0.067 &  0.066 & 0.0672 & 0.067&0.066 &0.067 \\
		$\sqrt{\mathrm{MSE}}$ &  0.084 & 0.094 &0.090 &0.085 & 0.083& 0.084  \\
		Noise level & 0.951 & 0.904 & 0.938& 0.952&0.940 &0.947\\
	\end{tabular}
	\label{tab:estimates}
\end{table}
\begin{figure}
	\includegraphics[width=8.2cm]{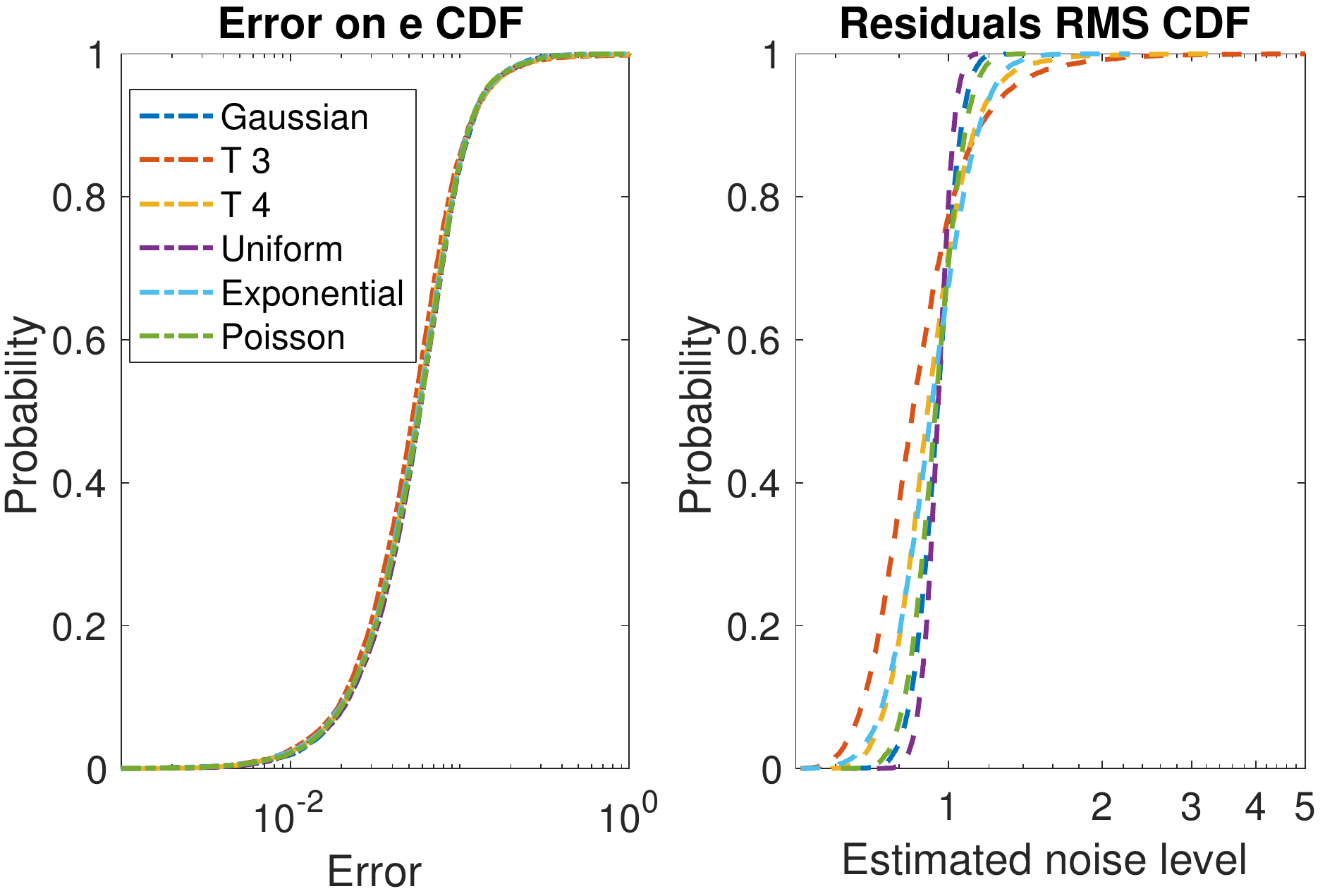}
	\caption{left: cumulative distribution function of the error on eccentricity for six different types of noise:  Gaussian, Student $T$ with 3 and 4 degrees of freedom,  uniform, exponential and Poisson. Right: standard deviation of the residuals after a Keplerian fit.}
	\label{fig:nongauss_err}
\end{figure}

\subsection{Distribution of the bias and the jitter }
\label{sec:appendix_nongauss}

In this section, we state and prove some mathematical results on the estimate of eccentricity obtained via a linear model. We assume as in section~\ref{sec:firstorder} that the eccentricity estimate is given by the linear model, $\mathbfit{y} = \mathbfss{M} \mathbfit{x} + \bepsilon$ where the first four columns of $\mathbfss{M} $ are such that the $i$-th line is evaluated at observation time $t_i$, $\mathbfss{M}_{i} = (\cos(n t_i) \; \sin(n t_i) \; \cos(2n t_i) \; \sin(2 n t_i))$ and $\mathbfit{x} = (A \; B \;  C \; D)^T$. The remaining columns $\mathbfss{M}$ are such that $\mathbfss{M}$ is of maximal rank. The eccentricity estimate is 
\begin{align}
\widehat{e} = \sqrt{\frac{\widehat{C}^2+ \widehat{D}^2}{\widehat{A}^2+ \widehat{B}^2}}.
\label{eq:eestimate}
\end{align}
We denote by $ \widehat{K} =\sqrt{ \widehat{A}^2+ \widehat{B}^2 }$.
The only assumption on the noise $\bepsilon$ is that a vanishing mean and a non degenerate covariance matrix $\mathbfss{V}$. 

The model just described is supposed to be the true model, which is unknown by the data analyst.
We assume that the model with which the analysis is done is $\mathbfit{y} = \mathbfss{M} \mathbfit{x} + \bepsilon$ where $\bepsilon$ is a Gaussian noise identically distributed of variance $\sigma^2$, that is a free parameter. In that model, for a data set $\mathbfit{y}_0$ the maximum likelihood estimates of $\mathbfit{x}$ and $\sigma$ are respectively 
\begin{align}
\widehat{\mathbfit{x}} &= (\mathbfss{M}^T \mathbfss{M})^{-1} \mathbfss{M}^T \mathbfit{y}_0 \\
\widehat{\sigma^2}  & = \frac{ \| \mathbfit{y}_0 -  \mathbfss{M} \widehat{\mathbfit{x}} \|^2 }{N} = \frac{ \| (\mathbfss{I} - \mathbfss{M}(\mathbfss{M}^T \mathbfss{M})^{-1} \mathbfss{M}^T) \mathbfit{y}_0  \|^2 }{N} = \frac{ \| \mathbfss{Q} \mathbfit{y}_0  \|^2 }{N}
\label{eq:mlsigma}
\end{align}
where we denote  by $\mathbfss{I}$ the $N\times N$ identity matrix and by $\mathbfss{Q} = \mathbfss{I} - \mathbfss{M}(\mathbfss{M}^T \mathbfss{M})^{-1} \mathbfss{M}^T$.
 
 The question we address is the dependency of the estimate~\eqref{eq:eestimate} on the noise nature. We show that several quantities relevant to our purposes are only determined by the covariance of the noise. Assertion \textit{iv} below shows that the bias on $e^2$ depends only on the covariance structure of the noise to order 2 in $1/K$. More precisely, 
\begin{theorem}
	\begin{enumerate}
\item $	\mathbb{E}\{ \widehat{\sigma^2}  \} = \frac{1}{N}\mathrm{tr}( \mathbfss{Q}  \mathbfss{V} \mathbfss{Q}^T )  $ where $\mathrm{tr}$ is the sum of the diagonal terms of a matrix (the trace).
\item $\mathbb{E}\{ \widehat{\mathbfit{x}}  \} = \mathbfit{x}_t$
\item  $\mathrm{Cov}\{ \widehat{\mathbfit{x}}  \} = (\mathbfss{M}^T \mathbfss{M})^{-1} \mathbfss{M}^T  \mathbfss{V}   \mathbfss{M}  (\mathbfss{M}^T \mathbfss{M})^{-1} $
\item $\mathbb{E}\{ \widehat{e}^2  \}-e_t^2 =  f( \mathbfss{V})   + o\left( \frac{1}{\widehat{K}^3} \right)$ . 
	\end{enumerate}
\end{theorem}
\begin{proof}
	(i) As $\mathbfit{y}_0 = \mathbfit{y}_t + \bepsilon $ and $\mathbfit{y}_t$ is in the image space of  $\mathbfss{M}$, we have $\mathbfss{Q} \mathbfit{y}_0 = \mathbfss{Q} \bepsilon$.
	
	 The estimate of the noise level is given by eq.~\eqref{eq:mlsigma}. Since 
	 $ \| \mathbfss{Q} \mathbfit{y}_0  \|^2 = \mathbfit{y}_0^T \mathbfss{Q}^T  \mathbfss{Q} \mathbfit{y}_0$ and $ \mathbfss{Q} \mathbfit{y}_0 = \mathbfss{Q} \bepsilon  $,  $	\mathbb{E}\{ \widehat{\sigma^2}  \} = \sum_{i,j} (QQ^T)_{ij}  \mathbb{E}\{ \epsilon_{i}   \epsilon_{j } \}  /N=  \sum_{i,j} (QQ^T)_{ij} V_{ij}/N =\frac{1}{N}\mathrm{tr}( \mathbfss{Q}  \mathbfss{V} \mathbfss{Q}^T )$.
	 
	 (ii)  We have  $\mathbb{E}\{ \widehat{\mathbfit{x}}  \} =$$ \mathbb{E}\{ (\mathbfss{M}^T \mathbfss{M})^{-1} \mathbfss{M}^T \mathbfit{y}_0 \} =$ $ (\mathbfss{M}^T \mathbfss{M})^{-1} \mathbfss{M}^T (\mathbfit{y}_t  +  \mathbb{E}\{ \bepsilon\} $ and by hypothesis, $\mathbb{E}\{ \epsilon \} = 0$, hence the result.
	 
	 (iii) $\mathrm{Cov}\{ \widehat{\mathbfit{x}}  \} =$ $\mathrm{Cov}\{ (\mathbfss{M}^T \mathbfss{M})^{-1} \mathbfss{M}^T \mathbfit{y}_0 \} =$$ \mathrm{Cov}\{(\mathbfss{M}^T \mathbfss{M})^{-1} \mathbfss{M}^T \bepsilon\}  $
	 
	 (iv) We pose $\widehat{A} = \bar{A} + a$, $\widehat{B} = \bar{B} + b$, $\widehat{C} = \bar{C} + a$, $\widehat{D} = \bar{D} + d$. Denoting by $\bar{K} =\sqrt{ \bar{C} + \bar{D} }$.  A Taylor expansion of the denominator $\widehat{ e^2}$ about $\bar{A}, \bar{B},\bar{C},\bar{D}$ at order two in $1/\bar{K}$ is
	 \begin{align}
	 \begin{split}
	 	\widehat{e^2} =& \frac{\bar{C}^2 + \bar{D}^2 + 2\bar{C}c + 2\bar{D}d +c^2 +d^2}{\bar{K}^2}  \\ &\left( 1 - 2\frac{\bar{A} a + \bar{B} b }{\bar{K}^2} - \frac{a^2+b^2}{\bar{K}^2} - 4\frac{(\bar{A} a + \bar{B} b )^2}{\bar{K}^4} +...\right) \\
	 		\mathbb{E}\{\widehat{e^2}\}	- e_t^2	  =&   \sum\limits_{i=1}^{\infty} \frac{\gamma_i}{\bar{K}^i} 
	 		 \end{split} 
	 \end{align}
	  A simple development shows that $\gamma_1=0$ and that $\gamma_2 = f( \bmu^2 )$ where $ \bmu^2 $ is the vector of moments of order 2 of $\mathbfss{x}$, which is a function of the noise covariance as shown in \textit{(iii)}.  For completeness, we give the explicit expression $\gamma_2 = \sigma_c^2 +\sigma_d^2 + 4e_t (\cos \psi \cos\phi  C_{ac} +\sin \psi \cos\phi  C_{cb} + \cos \psi \sin\phi  C_{ad} + \sin \psi \sin\phi  C_{db} )  +  e_t^2( \sigma_a^2 +\sigma_b^2) $ where $C_{ij}$ is the covariance of $\widehat{ I}$ and  $\widehat{ J}$ and $\bar{A} = \bar{K} \cos \phi$, $\bar{B} = \bar{K} \sin \phi$ $\bar{C} = \bar{e} \bar{K} \cos \phi$, $\bar{C} = \bar{e} \bar{K} \sin \phi$.
	 \end{proof}

\section{Residual analysis}
\label{appendixresiduals}

In this section we compute the law followed by the residuals of a linear least-square fit. Let us suppose that we have a model
\begin{align*}
\mathbfit{y} = \mathbfss{A}\mathbfit{x} + \bepsilon, \; \; \bepsilon \sim G(0,\mathbfss{V})
\end{align*}
where $\mathbfit{y}$ is a vector of $N$ observations, modelled as a linear combination of the column of the $N\times p$ matrix $A$, and $\bepsilon$ is a Gaussian noise of covariance matrix $\mathbfss{V} =: \mathbfss{W}^{-1}$. Assuming $\mathbfss{V}$ and $\mathbfss{A}$ are known, the least square estimate of $\mathbfit{y}$ is $\widehat{\mathbfit{y}} = \mathbfss{A}(\mathbfss{A}^T\mathbfss{W}\mathbfss{A})^{-1}\mathbfss{A}^T\mathbfss{W}\mathbfit{y}$ Therefore
\begin{align*}
\mathbfss{W}^{1/2}(\mathbfit{y} - \widehat{\mathbfit{y}}) & =  \mathbfss{W}^{1/2}\left( \mathbfss{A}\mathbfit{x}+ \bepsilon - \mathbfss{A}(\mathbfss{A}^T\mathbfss{W}\mathbfss{A})^{-1}\mathbfss{A}^T\mathbfss{W}(\mathbfss{A}\mathbfit{x}+\bepsilon)   \right) \\
& = \mathbfss{W}^{1/2}(\mathbfss{I}_N -\mathbfss{B})\bepsilon \\
& =: \mathbfit{r}_{W}
\end{align*}
where $\mat I_N$ is the identity matrix of size $N$ and $\mathbfss{B}:= \mathbfss{A }(\mathbfss{A}^T\mathbfss{W}\mathbfss{A})^{-1}\mathbfss{A}^T\mathbfss{W}$. The quantity $\vec r_{W}$, being a product of a matrix ($\mathbfss{W}^{1/2}(\mathbfss{I}_n -\mathbfss{B})$) with a Gaussian random variable of covariance $\mathbfss{V}$ has a covariance $\mathbfss{U}$
\begin{align*}
\mathbfss{U} & = \mathbfss{W}^{1/2}(\mathbfss{I}_N -\mathbfss{B}^T) \mathbfss{V} (\mathbfss{I}_N -\mathbfss{B})\mathbfss{W}^{1/2} \\
& = \mathbfss{W}^{1/2} ( \mathbfss{V} -  \mathbfss{BV} - \mathbfss{VB}^T + \mathbfss{BVB}^T) \mathbfss{W}^{1/2} \; \; \\ \text{since}\; \;   \mathbfss{W}^{1/2}\mathbfss{VB}^T = & \mathbfss{W}^{-1/2}\mathbfss{B}^T \mathbfss{W}^{1/2} = \mathbfss{W}^{1/2} \mathbfss{BVB}^T \mathbfss{W}^{1/2}, \\
\mathbfss{U}   & = \mathbfss{I}_N - \mathbfss{W}^{1/2}\mathbfss{B}^T \mathbfss{W}^{-1/2} \\
& = \mathbfss{I}_N - \mathbfss{C}(\mathbfss{C}^T \mathbfss{C})^{-1} \mathbfss{C}^T 
\end{align*}
where $\mathbfss{C} = \mathbfss{W}^{1/2} \mathbfss{A}$. This notation is convenient because it shows clearly that $\mathbfss{P} = \mathbfss{C}(\mathbfss{C}^T \mathbfss{C})^{-1} \mathbfss{C}^T$ is a projection matrix on the space generated by the columns of $\mat C$. Finally
\begin{align}
\mathbfss{U} = \mathbfss{I}_N - \mathbfss{P}
\label{eq:imp}
\end{align}
Is a projection on the space orthogonal to the one generated by $\mat C$ columns. Therefore, there exists an orthonormal matrix $\mat Q$ such that 
\begin{align*}
\mathbfss{Q}^T \mathbfss{UQ} = \mathbfss{J}_p
\end{align*} 
where  $\mathbfss{J}_p$ is a diagonal matrix whose first $p$ elements are zero and the others are equal to one. Finally, let us remark that $\mathbfit{r}_{QW}:= \mathbfss{Q}^T \mathbfit{r}_{W}$ has a covariance matrix $ \mathbfss{Q}^T \mathbfss{UQ}= \mathbfss{J}_p$, which shows the claim of section~\ref{sec:resanalysis}, $\mathbfss{Q}^T \mathbfss{W}^{1/2}(\mathbfit{y} - \widehat{\mathbfit{y}})$ has $p$ 0 components and the others are Gaussian variables of mean $0$ and variance 1.

Let us finally note that the covariance matrix $\mathbfss{U} = \mathbfss{I}_N - \mathbfss{P}$ of $\mathbfit{r}_{W} = \mathbfss{W}^{1/2}(\mathbfit{y} - \widehat{\mathbfit{y}})$ will be close to identity if there are many more observations than parameters. This explains why the weighted residuals $\mathbfit{r}_{W}$ almost behaves like independent Gaussian variables and, for instance, why plotting $\mathbfit{r}_{W}(t_i) - \mathbfit{r}_{W}(t_j)$ as a function of $t_i - t_j$ gives hints on the correlations.

\end{document}